%% file: main.tex
\documentclass[12pt]{article}

\usepackage[margin=1in]{geometry}
\usepackage{amsmath, amssymb}
\usepackage{empheq}
\usepackage{graphicx}
\usepackage{setspace}
\usepackage{microtype}
\usepackage{natbib}
\usepackage{array,booktabs,tabularx,ragged2e,caption,xcolor,multirow}
\usepackage{longtable,pdflscape}
\usepackage{colortbl}
\usepackage{arydshln}
\newcolumntype{L}{>{\RaggedRight\arraybackslash}X}
\newcolumntype{C}{>{\Centering\arraybackslash}X}
\newcolumntype{R}{>{\RaggedLeft\arraybackslash}X}
\usepackage{tikz}
\usepackage{amsmath,amsthm,amsfonts,amssymb,dsfont,bm,mathrsfs}
\usepackage{enumerate,color,xcolor}
\usepackage{graphicx}
\usepackage[bottom]{footmisc}
\usepackage[colorlinks,linkcolor=blue,citecolor=orange]{hyperref}
\usepackage{url}
\usepackage{placeins}

\numberwithin{equation}{section}
\theoremstyle{plain}

\newtheorem{theorem}{Theorem}
\newtheorem{corollary}{Corollary}

\newtheorem{lemma}{Lemma}

\newtheorem{assumption}{Assumption}

\newcommand{\ep}{\varepsilon}

\title{Residual Learning in Empirical Asset Pricing\thanks{
The code for this paper is available at \url{https://github.com/dexin-peng/ResnetAssetPricing}.
We gratefully acknowledge the computing resources supported by the Advanced Computing Devices (High Performance Computing platform Phase III) at the Hong Kong University of Science and Technology (Guangzhou).
}}

\author{
Dexin Peng\thanks{FinTech Thrust, Hong Kong University of Science and Technology (Guangzhou), email address: dpeng965@connect.hkust-gz.edu.cn.}
\and Xiaoyu Wang\thanks{FinTech Thrust, Hong Kong University of Science and Technology (Guangzhou), email address: xiaoyuwang@hkust-gz.edu.cn.}
}
\date{October 7, 2026}

\begin{document}
\pagenumbering{roman}
\pagestyle{empty}
\maketitle
\thispagestyle{empty}

\begin{abstract}
\setstretch{1}
Shallow models are special cases of deep models, and deep models theoretically have the potential to outperform the shallow ones.
However, the existing empirical asset pricing literature provides strong benchmarks for shallow models.
Residual learning allows neural network models in asset pricing to go deeper by preserving and refining their shallow counterparts.
The out-of-sample Sharpe ratio for value-weighted long-short portfolios of deep residual models (2.07) is higher than that for the corresponding shallow ones (1.92) and more than twice that of the deep feedforward models (0.89).
We show that model depth is a source of additional economic value in asset pricing.
Residual learning can be used to deepen other neural-network-based asset pricing models if they contain intermediate layers.
Our design also provides one way to scale asset pricing models, making native ``large asset pricing models'' more feasible.

\medskip
\noindent\textbf{Keywords:} Residual network; asset pricing; deep learning; return prediction.
\end{abstract}

\vfill
\clearpage

\tableofcontents
\clearpage

\pagenumbering{arabic}
\pagestyle{plain}
\section{Introduction}
\label{sec:introduction}

In equity markets, neural networks are prominent tools for return forecasting, but the models are typically relatively shallow.
\citet{gu2020EmpiricalAssetPricing} show that, in their setting, the performance of neural network models peaks at relatively shallow depths, and simply adding more feedforward layers does not deliver correspondingly greater economic value.
They propose that, ``shallow'' learning outperforms ``deep'' learning.

However, if a shallow neural network model already performs well, adding an identity mapping $f(x) = x$ to its intermediate layer(s) extends its depth, but the performance remains the same.
Besides the identity mapping $f(x) = x$, there are functions and specifications of parameters that can both extend depth and improve model performance at the same time.
Theoretically, shallow models are subsets of deep models, and deep models have the potential to outperform the shallow ones.
The potential should apply to the asset pricing field.
We propose that,
\begin{center}
\textit{``Shallow'' is a special case of ``deep''.}
\end{center}

Recent studies examine deeper architectures in related asset pricing problems.
\citet{zhang2022AssetPricingDeep} compares feedforward and residual architectures for stock-return forecasting and reports only modest incremental gains; \citet{lo2023DeeplearningModelsForecasting} use neural networks with skip connections for risk-premium forecasting and \citet{kelly2026LargeDeepFactor} study how depth and width affect stochastic discount factors based on deep neural networks.
But none of them study residual models with varying depths.
Our study explores the potential of additional depth by isolating the effect of residual learning and examining performance as depth gradually increases.
We also examine whether residual learning allows additional depth to preserve a profitable shallow ranking and improve its out-of-sample economic value in an asset-pricing-tailored model architecture.

Following the residual learning convention in \citet{he2016deep,he2016identity}, we add residual blocks with shortcut paths to feedforward asset pricing models.
A residual block is a combination of a shortcut path and a residual branch, in which the shortcut path keeps an earlier representation intact through an identity mapping and the residual branch learns a residual update.
We find that successive layers with residual learning can make the ranking of stocks in the cross-section more accurate, thereby achieving better out-of-sample portfolio performance, and the same design can be used in other neural-network-based asset pricing models if they contain intermediate layers.

We also develop a more compatible feedforward model to adapt residual learning to asset pricing.
This model is tailored to financial data with a low signal-to-noise ratio and to the training protocol of asset pricing neural network models.
We compare machine-learning models and the corresponding training protocols across several areas of empirical asset pricing, including stock-return forecasting \citep{gu2020EmpiricalAssetPricing,gu2021AutoencoderAssetPricinga,jiang2023ReImaginingPriceTrends,chen2024DeepLearningAsset}, option-return forecasting \citep{bali2023optionml} and corporate-bond return forecasting \citep{bell2024bondml}; finally, we develop a more compatible feedforward model NN+ (and its corresponding residual learning version) in asset pricing.\footnotemark

\footnotetext{See \ref{sec:resnetp} for the model definitions.}

Two comparison sets are developed to evaluate the residual-learning framework across varying depths.
The first is the raw benchmark comparison set, which compares a feedforward neural network model (NN) with the same model directly equipped with residual learning (ResNet).
The second is the main comparison set, which compares the asset-pricing-tailored feedforward neural network model (NN+) with the same model directly equipped with residual learning (ResNet+).
In each comparison set, we ensure that the models are evaluated under the same conditions such as data, preprocessing, training protocol, and portfolio evaluation methods to isolate the effect of residual learning.
Most of the conditions, as well as the NN model in the raw benchmark comparison set, are the same as those in \citet{gu2020EmpiricalAssetPricing}.

In the deep group, the out-of-sample Sharpe ratios for value-weighted long-short portfolios are 2.07 for ResNet+ and 0.89 for NN+, while in the shallow group, the corresponding Sharpe ratios are 1.92 for ResNet+ and 1.84 for NN+.
The change from shallow to deep indicates that the additional depth through residual learning generates additional economic value; similar trends are found in the raw benchmark comparison set as well.
To find out how the additional depth through residual learning generates additional economic value, we follow the portfolio-sort tradition in empirical asset pricing \citep{FamaFrench2008Dissecting,HouXueZhang2015} by examining realized returns across all forecast deciles.
Beyond decile sorting, we examine every pair of stocks in the cross-section to further explore how deep residual neural networks interact with their shallow counterparts.
The results show that deep residual models are capable of preserving the ranking already made by the corresponding shallow ones and refining the ranking through accurate revisions as depth increases, thereby achieving a more accurate ranking of stocks in the cross-section, which is then translated into better out-of-sample portfolio performance.
The ResNet+ versus NN+ gap stays large and distinct after excluding extremely large- or small-cap stocks, survives transaction costs, and persists across most of the NBER business cycles.

The residual-learning framework introduced in \citet{he2016deep} allows for training models with substantially more parameters.
Asset pricing models exhibit a trend toward including more parameters.
The CAPM presented in \citet{sharpe1964CapitalAssetPrices} is a starting point.
If we view the formula used in the OLS estimation of the CAPM as a linear machine-learning model, its empirical regression estimates two parameters: an intercept $\alpha$ and a market loading $\beta$.
Adding factors increases the number of parameters in the linear model, as in the models discussed by \citet{FamaFrench1993,harvey2016CrossSectionExpectedReturns}.
Recent uses of large models, or attempts to use them, in financial studies typically involve language models.
Several studies use large language models to extract return-predictive signals from financial text \citep{chen2022ExpectedReturnsLarge,chen2025ChatGPTDeepseekCan,lopez-lira2023CanChatGPTForecast}, but they treat the models as given tools developed in the machine learning field.
Some studies train large language models natively on financial text, such as BloombergGPT \citep{wu2023BloombergGPTLargeLanguage}, or fine-tune existing language models for financial applications, such as FinGPT \citep{yang2025FinGPTOpenSourceFinancial}; they generally follow the approach ``information $\xrightarrow{\text{LLM}}$ structured data $\rightarrow$ return forecasts''.
Our study is compatible with these existing approaches because we focus on the path ``structured data $\xrightarrow{\text{LAPM}}$ return forecasts'' by training a model directly on structured panel data to output return forecasts.
In one study, \citet{kelly2026LargeDeepFactor} develop native large asset pricing models, but these models use relatively shallow networks and do not directly output individual stock return forecasts.
Our study shows the possibility of scaling up native large asset pricing models by using residual learning to deepen the model architecture and improve out-of-sample performance.

\section{Framework and Model Design}
\label{sec:framework}
We use information at time $t$ to forecast the one-month-ahead excess return of stock $i$ according to
\begin{equation}
\label{eqn:assetpricing:basic}
\widehat{R}_{i,t+1}=f_{\theta_t}(X_{i,t}),
\qquad
\varepsilon_{i,t+1}:=R_{i,t+1}-\widehat{R}_{i,t+1},
\qquad i = 1,\ldots,N_t, \ t = 1,\ldots,T,
\end{equation}
where $X_{i,t}$ is the information set for stock $i$ at time $t$, and $f_{\theta_t}$ is the forecasting function with parameter set $\theta_t$.
Both the information set $X_{i,t}$ and the parameter set $\theta_t$ are time dependent, and the forecasting function $f_{\theta_t}$ can take different forms, representing different model architectures.
The information set $X_{i,t}$ includes currently available firm characteristics, macroeconomic predictors, and other predictive signals observed before and at time $t$.
The realized return $R_{i,t+1}$ is the target, and $\varepsilon_{i,t+1}$ is the forecast error, measured as the difference between the realized return $R_{i,t+1}$ and the return forecast $\widehat{R}_{i,t+1}$.
We estimate the parameter set $\theta_t$ for the model $f$ at time $t$ by minimizing the forecast mean squared error (MSE),
\begin{equation}
\label{eq:loss_fn}
\theta_t^* = \arg\min_{\theta} \mathbb{E}_t\left\Vert R_{t+1} - \widehat{R}_{\theta,t+1} \right\Vert^2.
\end{equation}
The parameter set $\theta_t$ for the model $f$ is optimized across time and firms, and is held fixed when applied to out-of-sample observations $X_{i,\tau}$ to forecast the corresponding future returns $R_{i,\tau+1}$.

We modify the forecasting function $f_{\theta_t}$ by introducing different models with or without residual learning here, and hold all other aspects of Equation \eqref{eqn:assetpricing:basic}, as well as the training protocol, the same.
Existing forecasting functions range from linear models to nonlinear methods such as random forests and neural networks, but we confine the analysis to neural-network-based asset pricing models because the residual neural network architecture is the focus of this study.

\subsection{From Feedforward Networks to Residual Neural Networks}

Prior evidence from \citet{gu2020EmpiricalAssetPricing} shows that nonlinear methods are capable of improving out-of-sample portfolio performance, but the performance of their neural network models peaks at a relatively shallow depth.

A feedforward neural network (NN) stacks hidden layers in sequence; more layers make the model more complex but do not guarantee better performance.
\citet{he2016deep} provide a residual block design in which shortcut paths carry the current hidden representation forward and residual branches learn residual updates to the hidden representation.
For return forecasting, the shortcut path carries the current (useful) pricing signal forward, while the residual branch searches for residual updates to that pricing signal.

\begin{figure}[htpb]
\centering
\includegraphics[width=0.75\textwidth]{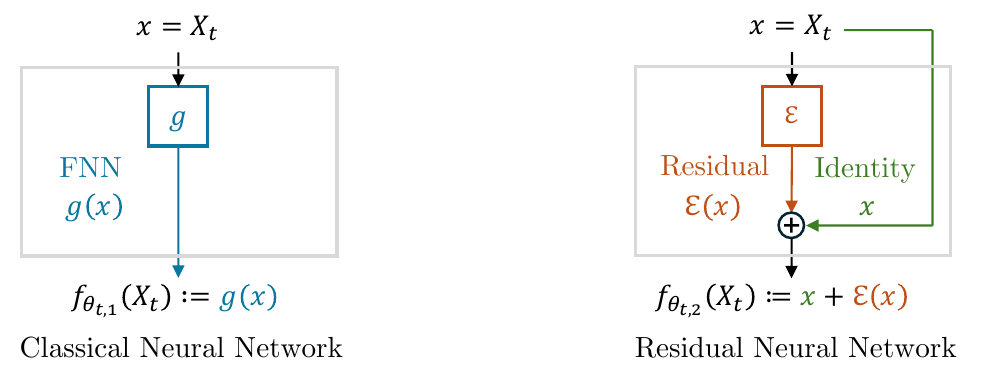}
\caption{\textbf{From Feedforward Networks to Residual Neural Networks.}
\newline\footnotesize
The left panel shows a feedforward block, where FNN denotes a feedforward neural network.
The right panel shows a residual block.
A feedforward block performs a simple function mapping, whereas a residual block keeps the input $x$ intact through a shortcut path and adds a learned residual update $\mathcal{E}(x)$ from the residual branch.}
\label{fig:resnet}
\end{figure}

Figure~\ref{fig:resnet} highlights the block-level architectural difference between feedforward and residual networks \citep{he2016deep,he2016identity}.
In a feedforward block, the learned transformation becomes the next representation.
A residual block supplies a second route, in which the shortcut path carries the input and the residual branch contributes a residual update.
A deep residual neural network contains multiple residual blocks (or their variants) together with other blocks, such as feedforward transition blocks.

Suppressing stock and time subscripts, let the initial map $\mathcal I(\cdot)$ convert the predictors into a hidden representation that is further processed according to
\begin{equation}
\label{eq:recursion}
h_0=\mathcal I(X),
\qquad
h_d = h_{d-1} + g_d(h_{d-1}),
\qquad d=1,\ldots,D.
\end{equation}

Here, $h_d$ is the hidden representation after residual block $d$, $D$ is the number of residual blocks in this idealized recursion, and $g_d(\cdot)$ is the residual update that modifies that output while the identity mapping keeps $h_{d-1}$ intact.
For any shallower specification in this idealized recursion with depth $d_0<D$, the deeper function class can reproduce it by setting $g_d(\cdot)\equiv 0$ for all $d>d_0$, so the shallow specification is a subset of the deeper specifications, and is nested inside the idealized deeper function class with residual learning.
``Shallow'' is a special case of ``deep'', and deep models have the potential to outperform the shallow ones.

A feedforward neural network in \citet{gu2020EmpiricalAssetPricing} uses the corresponding recursion $h_d:=g_d\left(h_{d-1}\right), d = 1,\ldots,D,$ so each feedforward block applies a new nonlinear mapping to the previous hidden representation, while Equation~\eqref{eq:recursion} carries the previous hidden representation $h_{d-1}$ directly forward and adds a residual update, allowing the residual branch to focus on the residual of the previous representation instead of entirely reevaluating $h_{d-1}$.
We build an empirical sequence that begins with a shallow model, then add layers under residual and feedforward rules, and examine whether these models retain, deteriorate, or improve upon the out-of-sample performance of the shallow model.

\subsection{Asset Pricing Residual Neural Network Specification}
\label{sec:resnetp}

\begin{figure}[htpb]
\centering
\includegraphics[width=0.9\textwidth]{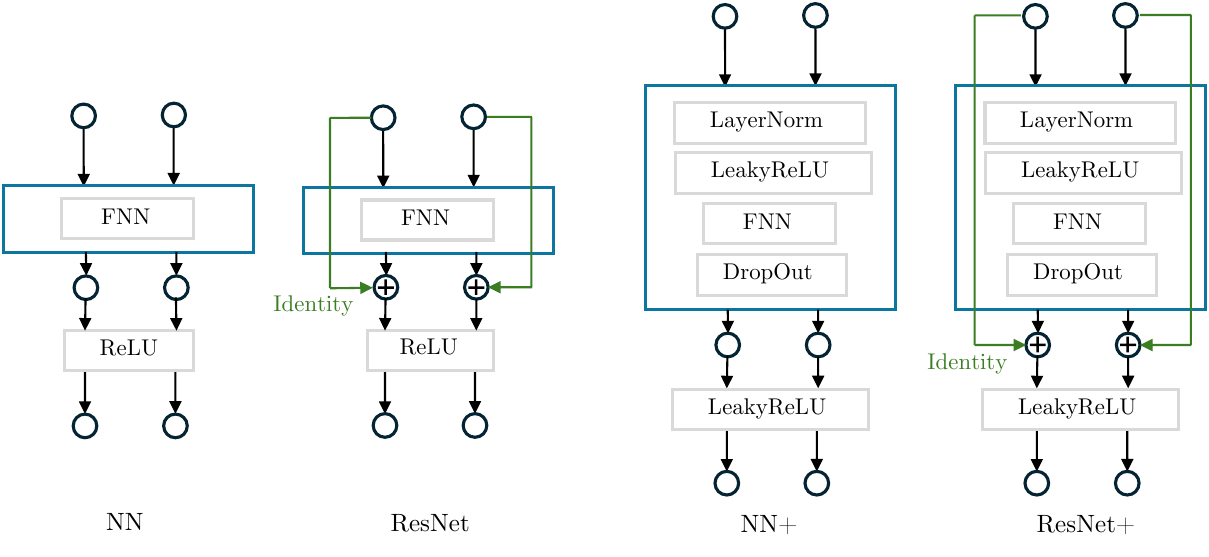}
\caption{\textbf{Simplified Building Blocks.}
\newline\footnotesize
Circles represent hidden states; arrows represent sequences of operations; text boxes, e.g., FNN and BatchNorm, represent functions.
The drawings show the core building blocks of the four models when the input and output dimensions of the building block are the same.
For projections and technical formulations, see Appendix~\ref{appendix:resnetp_specification}.
}
\label{fig:archsimplespec}
\end{figure}

Generally speaking, an identity shortcut requires the $h_{d-1}$ and $h_d$ representations to have the same hidden width, but the hidden width can change across layers in a neural network.
For example, the hidden width can be 32 in one layer and 16 in the next layer, but we cannot use an identity shortcut path between representations with different hidden widths because the integer 16 is not equal to 32.

However, our ResNet+ allows the hidden width to change while still being able to maintain the residual learning structure by using projection paths.
To better unlock the potential of the residual learning as well as the neural network approach in asset pricing, we tailor\footnotemark \ ResNet+ to the nature of financial data and asset pricing tasks; ResNet+ preserves the ordinary feedforward transition when a new hidden width is first introduced and applies residual learning to the representation when additional depth is allocated within that stage.
\footnotetext{
The implementation uses layer normalization \citep{ba2016layer}, LeakyReLU, dropout \citep{Srivastava2014}, and pre-activation residual branches.
Given the use of cross-sectional rank normalization in preprocessing, layer normalization is used in the model architecture to reduce the effects of scale differences across time periods; LeakyReLU is used to accommodate the preceding rank normalization step; and dropout is used to alleviate overfitting.
To better follow residual network practice, we use the pre-activation ordering in \citet{he2016identity}, where normalization and activation occur before the affine transformation in each branch.
}

To isolate the sole effect of residual learning from the combined effect of the tailored architecture and residual learning, we remove the residual-learning components from ResNet+ and replace them with tailored feedforward blocks to form its counterpart NN+, which is identical in overall structure to ResNet+ except that it replaces residual blocks with feedforward blocks.

\section{Empirical Study}
\label{sec:empirical}

Our empirical design follows \citet{gu2020EmpiricalAssetPricing} in evaluating out-of-sample stock return forecasts and uses the U.S. stock-level panel data and screens of \citet{jensen2023ThereReplicationCrisisa}.
Within this setting, comparisons between residual and feedforward model families at the same width and depth specifications test whether residual learning improves future return forecasting by evaluating the corresponding out-of-sample portfolio performance.

\subsection{Data and Preprocessing}
The stock-level panel dataset is constructed using data from the Center for Research in Security Prices (CRSP) and annual accounting data from Compustat, following the predefined rules of \citet{jensen2023ThereReplicationCrisisa}\footnotemark.
\footnotetext{The screening rules are applied to remove duplicate records, non-standard securities, and secondary share classes to ensure a clean, reliable equity panel with one observation per security per month for ordinary common shares traded on primary exchanges and make sure they are uniquely identified across data sources.
Details of the screening rules and the data acquisition process are available in the online ``ReplicationCrisis'' repository maintained by bkelly-lab on GitHub.
See \url{https://github.com/bkelly-lab/ReplicationCrisis}.}
The dataset spans January 1963 through December 2023, totaling 61 years.
It contains eligible U.S. common stocks after the Jensen, Kelly, and Pedersen security screens are applied.
The beginning year is selected to match the starting year of the online dataset for the Fama--French five-factor model in \citet{fama2015FivefactorAssetPricing}\footnotemark.
The resulting sample contains approximately 5,000 stocks per month and 3.57 million stock-month observations in total.

\footnotetext{
The online dataset is available on Kenneth R. French's website.
See \url{https://mba.tuck.dartmouth.edu/pages/faculty/ken.french/data_library.html}.}

Predictors combine $153$ firm characteristics from the stock-level panel data with $8$ macroeconomic predictors from \citet{welch2008ComprehensiveLookEmpirical}\footnotemark.

\footnotetext{The eight macroeconomic predictors are the dividend-price ratio (\texttt{dp}), earnings-price ratio (\texttt{ep}), book-to-market ratio (\texttt{bm}), net equity expansion (\texttt{ntis}), Treasury-bill rate (\texttt{tbl}), term spread (\texttt{tms}), default spread (\texttt{dfy}), and stock variance (\texttt{svar}).
The dataset and detailed variable descriptions are available on Amit Goyal's website.
See \url{https://sites.google.com/view/agoyal145}.}

All predictors are lagged by one month, so that information at month $t$ is used to forecast future returns in month $t+1$.
For example, at the end of September 2023, for a randomly selected stock, the predictor ``total accruals'' is the most recent publicly available value of total accruals on or before September 30, 2023.
The same timing applies to all other $152+8$ predictors, and the response variable is the realized excess return in October 2023, which is the stock return from September 30, 2023, to October 31, 2023, adjusted by the risk-free rate.

For data preprocessing, in each month, we first fill missing firm characteristic values with their cross-sectional medians and then apply cross-sectional rank normalization.
We apply rank normalization to the firm characteristics; the macroeconomic predictors are not rank-normalized because every stock shares the same value at the same time.
If a firm characteristic is missing for every stock in a month, we fill it with zero because zero represents the cross-sectional median after rank normalization.
Therefore, filling firm characteristics that are entirely missing within a cross-section with zero is consistent with the intuition of median imputation.

\begin{figure}
    \centering
    \includegraphics[width=1\textwidth]{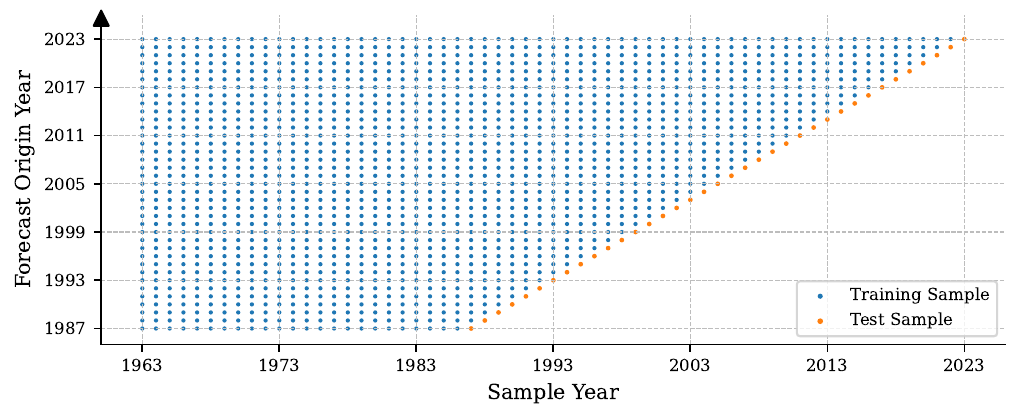}
    \caption{\textbf{Expanding Training Window.}
    \newline
    Each row represents the information observable in one forecast origin year.
    Time flows along the y-axis, from bottom to top, and the x-axis represents the sample years observable at the current forecast origin year.
    The blue points in each row indicate the data in the training sample for that forecast origin year; the orange point indicates the corresponding out-of-sample test sample.
    The model parameters are re-estimated for the latest forecast origin year to adapt to the latest market conditions.
    Forecast origins from January through December are used to forecast returns from February through the following January.
    }
    \label{fig:ExpandingTrainTestSplitStrategy}
\end{figure}

With expanding training windows, the model parameters are optimized using an initial training sample containing predictor observations from January 1963 through December 1986 and their one-month-ahead returns.
The last response observation in the initial training sample is the realized return in January 1987.
After the January 1987 returns are observed, model training is completed, and the model uses forecast origin observations from January through December 1987 to forecast out-of-sample future returns from February 1987 through January 1988.
Subsequent annual refits follow this expanding approach, yielding 443 forecast origins from January 1987 through November 2023 and realized returns from February 1987 through December 2023.
The expanding training window is designed to avoid forward-looking bias while preserving the mapping $X_t \rightarrow R_{t+1}$ throughout model training and out-of-sample forecasting.

Every model in our study shares the same source data, preprocessing methods, expanding training window rule, optimization objective, and evaluation measures, to isolate the effects of model architecture and depth.

\subsection{Experimental Design}

Our experimental design studies depth by comparing out-of-sample results from residual and feedforward model families within prespecified width and depth specifications.
At each width and depth specification, the two model families in each comparison set share every implementation detail except for the model architecture, with the only difference being whether the residual learning is applied.
The design is intended to reveal out-of-sample performance as model depth (and the model width) gradually increases, while avoiding ex post model selection, which is important given multiple testing concerns \citep{harvey2016CrossSectionExpectedReturns}.
The shallowest specification of one width schedule is used to anchor the out-of-sample performance for the specific width schedule of the model of interest, and then we track the performance as model depth gradually increases across a range of width schedules.

\input{ExperimentalDesignGrid.tex}

Table~\ref{tab:ExperimentalDesignGrid} describes the two comparison sets used throughout the empirical analysis.
ResNet versus NN anchors the analysis to the feedforward neural network benchmark in \citet{gu2020EmpiricalAssetPricing}, which documents limited gains from simply adding feedforward layers, but this comparison set uses the data from \citet{jensen2023ThereReplicationCrisisa}.
This comparison set shows how the residual learning changes the out-of-sample performance of NN as depth (and width) gradually increases.
ResNet+ versus NN+ is the main comparison set: it asks whether the residual learning provides an out-of-sample performance advantage when the feedforward model is tailored to the nature of financial data and is more compatible with the asset pricing task.
Both comparison sets yield evidence on the effect of residual learning on out-of-sample performance, and comparing corresponding width and depth specifications across them shows how NN+ performs relative to NN.

Each width schedule $s$ in both comparison sets starts with the minimum depth under that schedule.
In each comparison set, the residual model is the same as the feedforward model at the minimum depth of each schedule, and the added layers then make their residual and feedforward update rules diverge.
For a given schedule, the models start from an identical forecast and differ only in how later blocks update that anchor.
The deeper the model is, the larger the difference between the residual and feedforward model architectures becomes, and the depth counts hidden layers from 1 through 20.
Each model family has 74 planned width and depth specifications: 14 shallow (1--5), 20 medium (6--10), and 40 deep (11--20).
Equivalently, each comparison set contains 74 planned width and depth specifications.
Completed specifications receive equal\footnotemark \ weight within their group, giving each width and depth specification equal influence.

\footnotetext{Equal weighting is applied to the specifications within each group to ensure fairness when aggregating across different width and depth specifications.
It does not imply equal-weighted portfolio formation.}

In each width and depth specification, every model is trained in-sample and evaluated out-of-sample to forecast one-month-ahead future excess returns from the same stock panel.
Specifically, in every forecast origin year, models are trained\footnotemark \ to minimize mean squared error at the individual stock level at every point in time in the given training sample; $\theta_t^*$ in Equation~\eqref{eq:loss_fn} is then updated to its latest fitted value.
We then use the currently optimal parameter set $\theta_t^*$, together with the out-of-sample test predictors, in the model $f$ specified in Equation~\eqref{eqn:assetpricing:basic} to forecast out-of-sample one-month-ahead future returns.
\footnotetext{Annual January refits use expanding training windows and warm starts based on the preceding annual fit parameters.
Training is performed for at most 100 epochs, with batches of 10,000 observations.
Early stopping monitors training loss over the most recent five years of the training sample and stops after five consecutive epochs showing no improvement.
NN and ResNet share the same Adam optimizer and the same $\ell_1$ optimization parameters; NN+ and ResNet+ share the same AdamW optimizer and the same $\ell_2$ optimization parameters.
Appendix~\ref{appendix:normalization_implementation} gives the full protocol.}
The ten-member ensemble means that, for every model specification defined by width and depth, training is repeated ten times with ten different random seeds; the out-of-sample forecasts from the corresponding ten random-seed-based models are averaged to produce a single out-of-sample forecast before portfolio formation to prevent the results from being driven by a lucky seed.
Finally, the out-of-sample forecasts for the model at the specified width and depth are obtained.

On the last day of each month, eligible stocks are sorted into deciles by their out-of-sample forecasts.
The analysis mainly uses value-weighted decile portfolios, with stock weights proportional to market capitalization.
Using the magnitude of the forecast itself as a performance measure is statistically and economically meaningful, but when it comes to forming a portfolio, using the ranking of the stocks is a more direct way.
Because the ranking of decile assignments does not depend on the distances between forecast values, and ranking reduces the sensitivity to raw forecast scale and outliers, some studies \citep{NguyenLo2012,gu2020EmpiricalAssetPricing,freyberger2020DissectingCharacteristicsNonparametrically} use rankings in empirical asset pricing.

The out-of-sample Sharpe ratio for value-weighted long-short portfolios is our main economic statistic.
Full decile realized returns serve as a proxy for the overall ranking accuracy of the model.
Out-of-sample $R^2$ is a statistical metric of forecasts \citep{RapachStraussZhou2010,gu2020EmpiricalAssetPricing}; our version calibrates out-of-sample $R^2$ to better handle the scale and extreme value issues.\footnotemark \
\footnotetext{
When evaluating the out-of-sample $R^2$, relative to \citet{gu2020EmpiricalAssetPricing}, we remove observations below the 2.5th percentile or above the 97.5th percentile of each month's forecast distribution and retain the central approximately 95\%.
We then center and rescale the retained forecasts as $\widehat R^{c}_{i,t+1}=\widehat\kappa(\widehat R_{i,t+1}-\overline{\widehat R}_{t})$, where $\overline{\widehat R}_{t}$ is the mean of the retained forecasts in month $t$ and $\widehat\kappa\geq0$ is one scale fitted by least squares over all retained out-of-sample stock-month observations for each specification.
We compute $R^2_{\mathrm{OOS}}=1-\frac{\sum(R_{i,t+1}-\widehat R^{c}_{i,t+1})^2}{\sum R_{i,t+1}^2}$ over the retained observations.
For example, multiplying every forecast $\widehat R_{i,t+1}$ by 100 leaves the rank-sorted portfolio unchanged but substantially worsens the unadjusted $R^2_{\mathrm{OOS}}$ because the raw forecast scale is distorted quadratically.
Our $R^2_{\mathrm{OOS}}$ absorbs this common change in scale, leaving $R^2_{\mathrm{OOS}}$ unchanged.
Similarly, multiplying $\widehat R_{i,t+1}$ in both tails by 100 while preserving their ranking leaves the rank-sorted portfolio unchanged but also substantially worsens the unadjusted $R^2_{\mathrm{OOS}}$ because extreme forecasts affect $R^2_{\mathrm{OOS}}$.
Since these tail observations are excluded by our trimming rule, the calibrated statistic is also unchanged.
See Appendix~\ref{appendix:forecast_portfolio_measures}.}
The Fama--French five-factor alpha measures the long-short return not explained by standard factor exposures \citep{fama2015FivefactorAssetPricing}.
Turnover records the trading required to rebalance the portfolio.
Market capitalization screens examine whether portfolio results depend on stocks with extreme market capitalizations or not.
Transaction cost assessments measure whether the portfolio results and conclusions still hold when transaction costs are taken into account.
Market state assessments show whether the portfolio results and conclusions are sensitive to the overall market conditions.
These metrics measure forecast accuracy, cross-sectional ranking, risk-adjusted portfolio performance, and the robustness of the results.

\subsection{Long-Short Portfolio Performance}
\label{sec:portfolio_performance}

During the out-of-sample period, each month, we form a value-weighted long-short portfolio by buying stocks in the highest forecast decile and short-selling stocks in the lowest forecast decile.
In the next month, we rebalance by the latest forecast.
We track the realized returns of the portfolio every month, forming a time series of out-of-sample long-short portfolio returns.
All performance metrics are calculated using the value-weighted (VW) approach, except for the equal-weighted (EW) Sharpe ratio in Table~\ref{tab:FinalBaselinePerformance} and the equal-weighted portfolio analysis in Appendix~\ref{appendix:equal_weighted_portfolios}.

\input{FinalBaselinePerformance.tex}

Table~\ref{tab:FinalBaselinePerformance} reports the performance of each model family across shallow, medium, and deep specifications, and out-of-sample portfolios formed from every model family are profitable.
The reported means and volatilities are based on the monthly long-short returns averaged within each model family and depth group.
To express Sharpe ratios on an annual basis, we estimate the annualized Sharpe ratio by scaling the monthly Sharpe ratio by the square root of 12.
We then calculate the Sharpe ratio for each width and depth specification and average these ratios within each model family and depth group to produce the reported Sharpe ratios.\footnotemark

\footnotetext{See Appendix~\ref{appendix:forecast_portfolio_measures} for formulations.}

ResNet and NN remain close across depths, but ResNet outperforms NN in every depth group.
ResNet's mean monthly long-short return and volatility change little across depths, with the corresponding value-weighted Sharpe ratio moving from 0.95 to 0.93, while the mean monthly long-short return for NN falls from 1.86 to 1.65 percent and its Sharpe ratio falls from 0.94 to 0.82 across those groups.
The out-of-sample performance of ResNet remains nearly the same across all three groups, while NN shows a monotonic decline in these performance metrics, indicating that directly adding the residual learning provides protection against deterioration with depth.

The mean monthly long-short return and volatility of ResNet+ also change little across depths, but as depth increases, the value-weighted Sharpe ratio increases monotonically and gradually from 1.92 in the shallow group to 1.94 in the medium group, reaching 2.07 in the deep group.
NN+ begins with a value-weighted Sharpe ratio of 1.84 and then falls to 1.47 and 0.89 as the specifications move from shallow through medium to deep.
With the tailored architecture, the residual learning provides an improvement in out-of-sample performance as network depth increases, while the performance of NN+ deteriorates as depth increases.

Regressions on the five Fama--French factors test whether the long-short portfolio has a significant alpha, with the market, size, value, profitability, and investment factors as explanatory variables and inference based on 12 Newey\mbox{-}West lags \citep{fama2015FivefactorAssetPricing,NeweyWest1987}.
The $R^2_{OOS}$ statistic shows whether the forecasts are close to the realized returns at the stock level after trimming and scaling.

The alpha, $R^2_{OOS}$, and equal-weighted portfolio results show patterns similar to the value-weighted Sharpe ratio, indicating that, as depth increases, the residual learning in tailored asset pricing neural network architectures is associated with more accurate forecasts, higher out-of-sample risk-adjusted returns, and return patterns that are harder for the common factor model to explain.

\begin{figure}[!t]
    \centering
    \includegraphics[width=\textwidth]{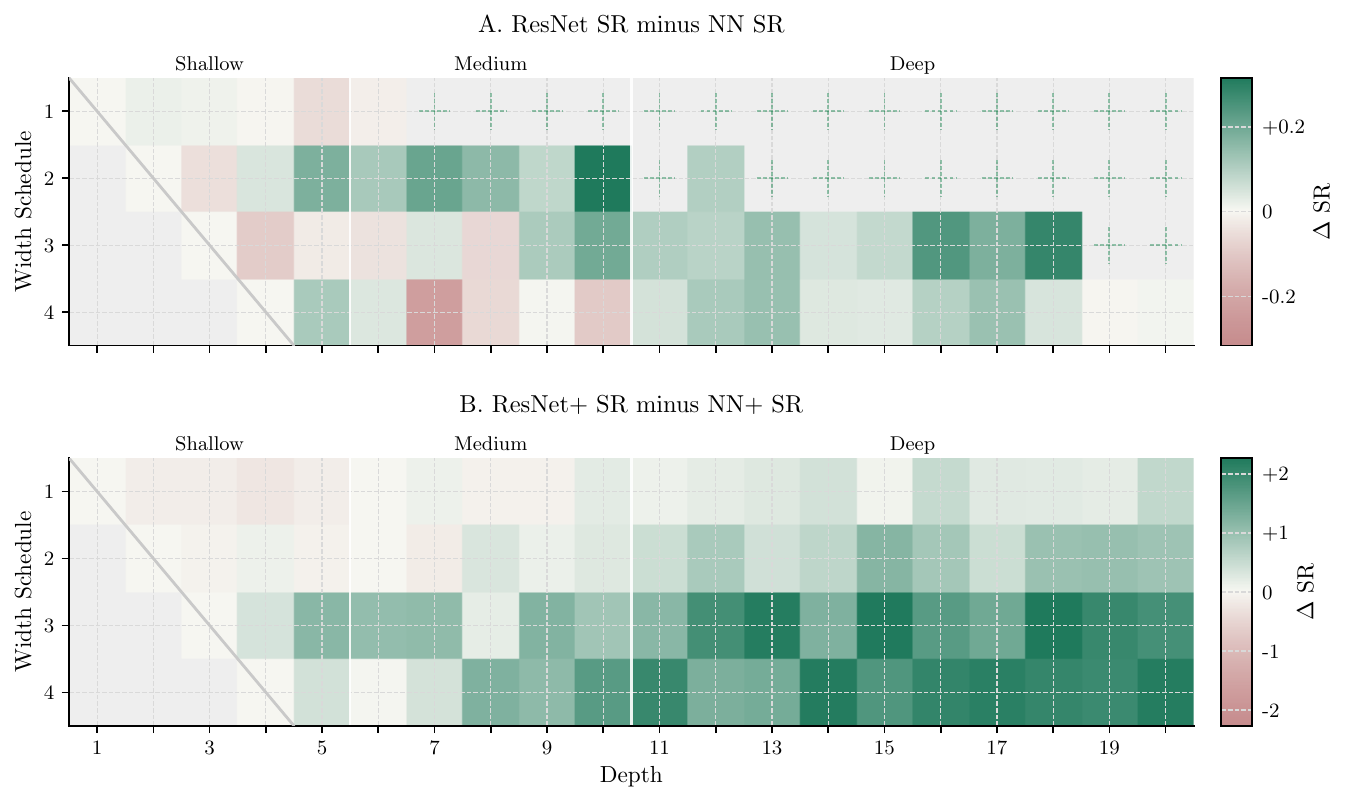}
    \caption{\textbf{Sharpe Ratio Differences.}
    Each grid entry represents a width and depth specification within one comparison set.
    Each entry reports the out-of-sample Sharpe ratio difference between the residual model and the feedforward model for value-weighted long-short portfolios.
    Green entries favor the residual model and red entries favor the feedforward model.
    Dashed green crosses indicate specifications in which the ResNet forecast remains valid while the NN forecast collapses.
    The diagonal line on the left shows the minimal depth for each width schedule.}
    \label{fig:SpecificationDeltaSRHeatmaps}
\end{figure}

To unpack the aggregated width and depth results, we examine differences in the out-of-sample Sharpe ratio for every width and depth specification within each comparison set in value-weighted long-short portfolios.
Figure~\ref{fig:SpecificationDeltaSRHeatmaps} shows that ResNet+ has a higher SR than NN+ in 4 of 14 shallow specifications; the count rises to 17 of 20 medium specifications and all 40 deep specifications.
ResNet has a higher SR than NN in 5 of 14 regular shallow specifications, 10 of 16 regular medium specifications, and 18 of 19 regular deep specifications, showing a trend similar to that of ResNet+.
The NN model is not shown in the figure at some width and depth specifications because it collapses in these specifications and yields a constant forecast; thus, it cannot form a valid portfolio.\footnotemark
\footnotetext{Appendix~\ref{appendix:s1d7_trainability_diagnostics} reports the diagnostic for one randomly selected case.}
When collapse occurs, the constant-forecast NN specification is excluded from the NN model family summary in Table~\ref{tab:FinalBaselinePerformance}, and the corresponding $\Delta$ SR is unavailable in Figure~\ref{fig:SpecificationDeltaSRHeatmaps}.
The heatmap marks that location with a dashed green cross.
Across the grid, residual models are more likely to achieve higher portfolio Sharpe ratios than their feedforward counterparts at medium and deep depths.
In the ResNet versus NN comparison set, the residual learning protects the model from collapse in deeper specifications of width schedules 1, 2 and 3, while the NN model collapses in these specifications.

\subsection{Forecast Ranking across Long-Only Decile Portfolios}
\label{sec:forecast_ranking_slopes}

Long-only value-weighted portfolios across all ten forecast deciles show how realized performance varies over the full ranking.
Each month, instead of forming a long-short portfolio as in Section~\ref{sec:portfolio_performance}, we take long positions in the stocks in each decile to form long-only value-weighted portfolios and track the monthly time series of their realized returns.
Then, we calculate their mean monthly realized returns, volatilities, and annualized Sharpe ratios over the out-of-sample period from the single forecast panel obtained after averaging the ten ensemble members.
Figure~\ref{fig:DecileMonotonicitySelected} averages the results across completed specifications shared by both model families in each comparison set. Table~\ref{tab:FullDecileMonotonicity} reports depth-group statistics for each model family's nonconstant-forecast specifications.

\begin{figure}[!htbp]
    \centering
    \includegraphics[width=\textwidth]{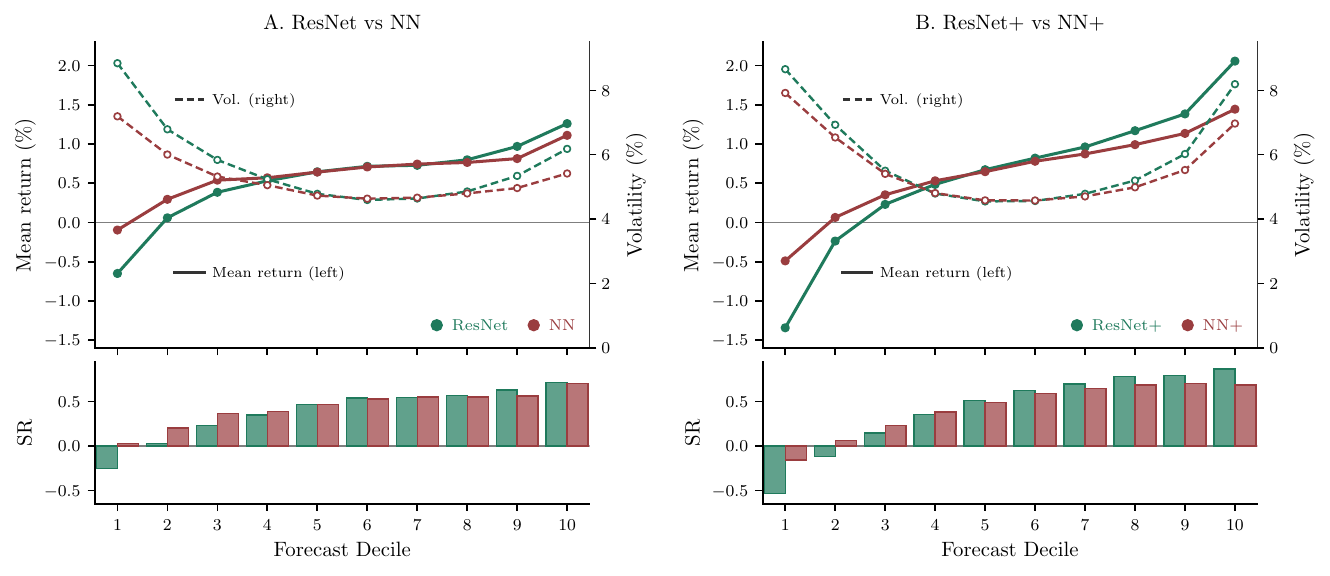}
    \caption{\textbf{Realized Returns, Volatility, and Sharpe Ratios across Deciles.}
    Solid lines report mean monthly realized returns on the left axis, dashed lines report mean volatility of monthly realized returns on the right axis, and the bars report annualized Sharpe ratios.
    The statistics are calculated from monthly value-weighted portfolios and then averaged across completed width and depth specifications shared by both model families in each comparison set.}
    \label{fig:DecileMonotonicitySelected}
\end{figure}

Figure~\ref{fig:DecileMonotonicitySelected} first shows a common rising trend in realized returns across forecast deciles, directly indicating that every model is useful for ranking stocks by their return forecasts.
The return profiles in each subfigure are close in the middle deciles and separate toward the tails, showing that the residual learning is more effective at distinguishing the most promising and the most undesirable stocks.
Deciles 3 and above show positive realized returns; this is due to the fact that the majority of stocks in the U.S. stock market have positive realized returns overall from 1987 to 2023.
The volatility profiles show a U-shaped pattern, with the lowest volatility in the middle deciles and higher volatility in the extreme deciles, showing that stocks in the extreme forecast deciles are more volatile.
In both comparison sets, the volatility of the right tail is lower than that of the left tail; together with the U-shaped volatility profiles, these results show an asymmetric volatility pattern across the forecast deciles.
Right-tail volatility is only slightly lower than the left-tail volatility in the ResNet+ versus NN+ comparison set, whereas there is an apparent decline toward the right tail in the ResNet versus NN comparison set, indicating that the tailored architecture is more effective in reducing the volatility asymmetry in the extreme forecast deciles.

The Sharpe ratios for all models also show a common rising trend across deciles, indicating that every model is economically useful in ranking stocks.
For stocks forecast to perform poorly in the low deciles, the Sharpe ratios of the residual models are generally lower than those of the feedforward models; for stocks forecast to perform promisingly in the high deciles, the Sharpe ratios of the residual models are generally higher than those of the feedforward models.

\input{FullDecileMonotonicity.tex}

Table~\ref{tab:FullDecileMonotonicity} documents how the aggregated results change across depth groups.
The most obvious anomaly is that the magnitude of the raw forecasts shrinks as depth increases for the NN and NN+ models, making the stock rankings less distinct and less accurate, as measured by mean realized returns, while the magnitudes of the raw forecasts of the ResNet and ResNet+ models remain stable.
Without the residual learning, as depth increases, the extra gain from the tailored architecture gradually vanishes, causing the deep NN+ models to approximate the performance of the NN model and eventually collapse.
By comparison, the ResNet and ResNet+ models show stable raw forecasts across depth groups, indicating that the residual learning is effective in preserving the shallow forecasts and preventing the collapse of the deep model.

The performance of ResNet remains stable across depth groups, while the performance of ResNet+ improves as depth increases.
Specifically, for the ResNet+ model, the long-short portfolio Sharpe ratio increases monotonically from 1.92 to 1.94 and finally reaches 2.07; for the ResNet model, the long-short portfolio Sharpe ratio decreases from 0.95 to 0.94 and finally reaches 0.93.
The residual learning in the tailored architecture is effective in refining the forecasts as depth increases.
For ResNet, the residual learning preserves the shallow forecasts as depth increases.

\subsection{Preservation and Refinement}
\label{sec:preservation_refinement}

Next, we examine how additional depth changes the economic value already present in a profitable shallower forecast.

Identity shortcut paths carry an earlier representation forward, while residual branches add residual updates \citep{he2016identity}.
Motivated by this ``architectural decomposition'', we develop an output-level decomposition that separates preservation and revision from the overall change in the forecast rankings as depth increases.
A deeper model may reproduce the ranking of the shallow model, discard it, or preserve most of the ranking and revise the rest.

For each model family and width schedule $s$, the shallow anchor is defined at the minimum possible depth $d_0(s)=s$.
Each deeper forecast is evaluated relative to the shallow anchor for its own model family and width schedule.
For example, the shallow anchor of ResNet+ under width schedule 2 is the forecast at depth 2, and the deeper forecasts to be decomposed are those in the medium and deep groups of ResNet+, spanning depths 6 through 20.

For a given model, width schedule, and target depth, after averaging the forecasts across the ensemble, we compare the shallow and deeper forecast rankings for every stock pair in each month $t$.\footnotemark \
\footnotetext{A stock pair is an unordered pair of stocks.
Suppose there are $N$ stocks in month $t$; then there are $\binom{N}{2}$ stock pairs.}
Let $n_t$ be the total number of pairs and $m_t$ be the number whose ordering changes.
For each changed pair $k$, let $c_{k,t+1}$ denote its pairwise revision score based on next month's realized returns.
We define the preservation score, revision accuracy, and revision economic value as
\begin{equation}
\begin{aligned}
P_t=1-\frac{m_t}{n_t},
\qquad
C_{t+1}&=\frac{1}{m_t}\sum_{k=1}^{m_t}c_{k,t+1},\\
V_{t+1}&=\frac{1}{n_t}\sum_{k=1}^{m_t}c_{k,t+1}
\left|R_{i_k,t+1}-R_{j_k,t+1}\right|,
\end{aligned}
\label{eq:forecast_revision}
\end{equation}
where the preservation score $P_t$ measures the fraction of stock pairs whose forecast ordering is unchanged from the corresponding minimum-depth shallow anchor forecast, ranging from $0$ to $1$, with $0$ indicating no preservation and $1$ indicating complete preservation.
$c_{k,t+1}$ is assigned $+1$ if the changed deeper model's ordering agrees with realized returns ordering, which means the deeper model corrects the wrong ordering of the shallow anchor, and $c_{k,t+1}$ is assigned $-1$ if the shallow anchor's ordering is correct but the deeper model's ordering is wrong.
Changes to or from a forecast tie receive half of the score, and a tie in realized returns contributes zero.\footnotemark \
Averaging $c_{k,t+1}$ gives the revision accuracy $C_{t+1}$, with $100\%$ indicating every revision is correct and $-100\%$ indicating every revision is wrong.
Because a correct revision may be more valuable when the realized returns of the two stocks are far apart, we weight each $c_{k,t+1}$ by the realized return difference to get the revision economic value $V_{t+1}$.
Revision economic value $V_{t+1}$ measures the economic value of these revisions over all $n_t$ stock pairs, where $i_k$ and $j_k$ are the stocks in changed pair $k$ and unchanged pairs contribute zero.
$R_{i_k,t+1}$ and $R_{j_k,t+1}$ are the subsequent realized returns of $i_k$ and $j_k$, respectively.
A positive revision economic value indicates refinement on average, and a negative value indicates deterioration.
We report $P_t$ on its original scale, $C_{t+1}$ in percent, and $V_{t+1}$ in basis points.

\footnotetext{See Appendix~\ref{appendix:performance_metrics} for formalized definitions.}

We then analyze the preservation score $P_t$, the revision accuracy $C_{t+1}$, and revision economic value $V_{t+1}$ over 443 forecast origins for every combination of model family, width schedule, and target depth.
Table~\ref{tab:ForecastRefinement} reports model family means for the medium and deep groups.
Figure~\ref{fig:ForecastPreservationAndCorrectionByDepth} reports individual depths from 6 through 20 for each model family.

\input{ForecastRefinement.tex}

\begin{figure}[!htbp]
    \centering
    \includegraphics[width=\textwidth]{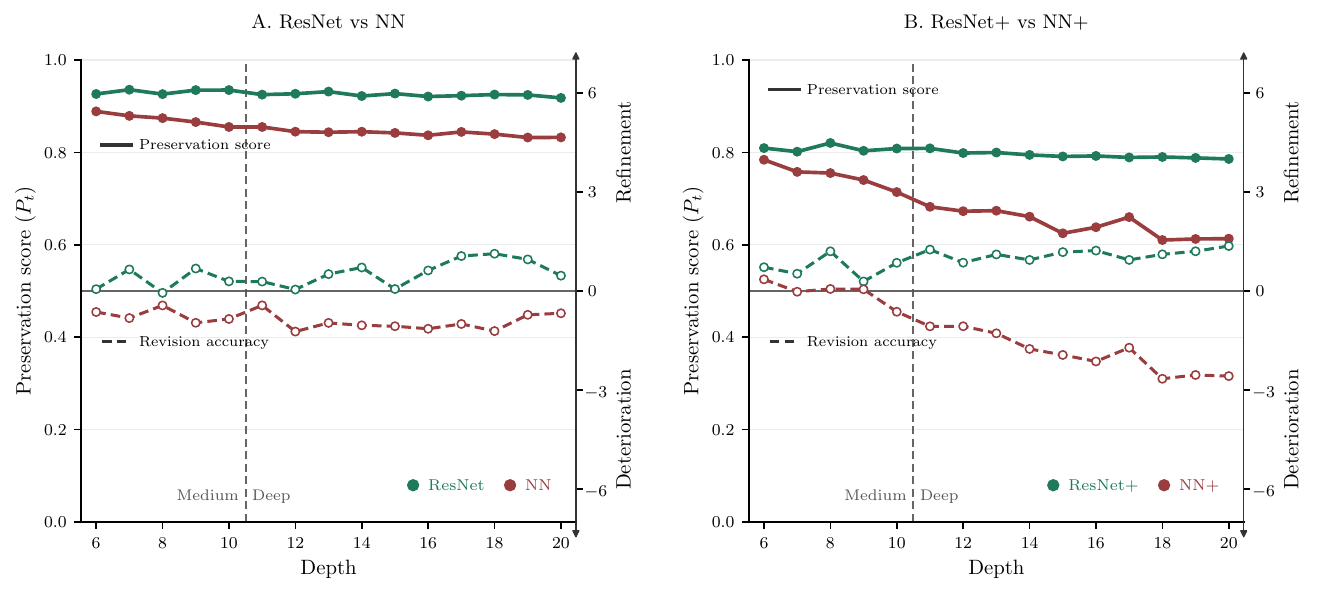}
    \caption{\textbf{Forecast Preservation and Refinement by Depth.}
    Solid lines report the preservation score $P_t$ relative to the corresponding shallow anchor model on the left axis.
    Dashed lines report the revision accuracy $C_{t+1}$ (\%) relative to the corresponding shallow anchor model on the right axis.
    Revision accuracy larger than zero indicates refinement, while accuracy less than zero indicates deterioration.
    Green and red indicate the residual and feedforward models, respectively, in each panel.}
    \label{fig:ForecastPreservationAndCorrectionByDepth}
\end{figure}

Preservation scores $P_t$ for ResNet are close to 1 and nearly the same throughout the depth range, and its revision accuracy $C_{t+1}$ is positive for both the medium and the deep group; for NN, the preservation scores are also close to 1 but lower than those for ResNet, while its revision accuracy is negative.
However, the revision economic values $V_{t+1}$ for ResNet and NN are both negative, with ResNet showing a smaller negative value.
The ResNet model shows the potential of residual-learning framework to make the model deeper.
NN shows a lower and decreasing ability to preserve the shallow forecast and does not yield statistically significant refinement.
This finding is consistent with the portfolio results in Section~\ref{sec:portfolio_performance} and Section~\ref{sec:forecast_ranking_slopes}, where ResNet preserves the out-of-sample performance and NN shows deteriorating out-of-sample performance as depth increases.
For ResNet+, the preservation scores are nearly the same throughout the depth range, while its revision accuracy and revision economic values show refinement and are higher than those of NN+.
For NN+, by contrast, the preservation scores fall from 0.75 to 0.64, and its revision accuracy shows deterioration as depth increases.
The finding is also consistent with previous portfolio results, where ResNet+ shows improving out-of-sample performance and NN+ shows deteriorating out-of-sample performance as depth increases.

To further explore the economic intuition underlying the preservation and refinement, we examine which firm characteristics the deeper models rely on to retain the shallow forecast ranking and make accurate and economically valuable revisions.
We group the firm characteristics into 13 economic themes according to \citet{jensen2023ThereReplicationCrisisa} and set the characteristics in each theme to zero, one theme at a time.
Holding the fitted models and other inputs fixed, we set the same theme to zero in both shallow and deep models, recompute their forecasts, and measure the resulting changes.

For each model specification and month, the performance losses from removing economic theme (characteristic group) $g$ are
\begin{equation}
\begin{aligned}
\Delta P_{g,t}=P_t-P_t^{(-g)},
\qquad
\Delta C_{g,t+1}&=C_{t+1}-C_{t+1}^{(-g)},\\
\Delta V_{g,t+1}&=V_{t+1}-V_{t+1}^{(-g)},
\end{aligned}
\label{eq:theme_losses}
\end{equation}
where the original $P$, $C$, and $V$ are the comparisons between the deep and shallow forecasts, and adding the superscript $(-g)$ to indicate that the characteristics in theme $g$ are set to zero in both models.
A positive performance loss indicates that removing the theme lowers the corresponding metrics, which means the theme is important for the performance of depth in that metric; a negative performance loss (improvement in the metric) indicates that keeping the theme contributes less to the deep model's than to the shallow model's.
We report $\Delta P_t$ on its original scale, $\Delta C_{t+1}$ in percent, and $\Delta V_{t+1}$ in basis points.
$\Delta P_t > 0$ indicates that the theme helps the deep and shallow models produce more consistent rankings; $\Delta C_{t+1} > 0$ indicates the theme helps achieve higher revision accuracy for deep models; and $\Delta V_{t+1} > 0$ indicates the theme helps achieve a higher revision economic value for deep models.

\input{ThemePairwiseLosses.tex}

Table~\ref{tab:theme_pairwise_losses} shows that for most themes, the residual-learning framework in deep ResNet+ preserves the ranking made by the shallow anchor well.
For the ``Low Leverage'' and ``Size'' themes, the residual-learning framework in deep ResNet+ reorders the ranking made by the shallow anchor more extensively and achieves a significantly higher revision economic value, showing economic refinement.
By contrast, in the NN+ model, the significantly heavier reordering is accompanied by deterioration in revision economic value, and preservation is accompanied by deterioration under the ``Low Risk'', ``Quality'', ``Short-Term Reversal'' and ``Value'' themes.
For ResNet and NN, there are only two significantly negative $\Delta P$ values in NN, but three significantly positive $\Delta P$ values appear in ResNet, showing that the residual-learning framework in deep ResNet can better preserve the ranking made by the shallow anchor.

Note that all $\Delta V$ values in NN+ are negative; the reason for the negative $\Delta V$ is that the shallow NN+ model utilizes the theme better than the deep ones.\footnotemark \
Another observation is that the results in the ResNet+ and NN+ model families are more likely to be significant than those in the ResNet and NN model families.
The tailored architecture is more likely to bring significant changes when a theme is removed.
Also, we can observe that only the ``Low Leverage'' and ``Profit Growth'' themes are significant in $\Delta P$ for all four model families, but neither is significant in both $\Delta C$ and $\Delta V$.
It is rare to find a theme that shows a significant effect on all model families.
Adding the ``Low Leverage'' theme characteristics generally makes the deep models reorder the rankings, while the ``Profit Growth'' theme characteristics generally make the deep models preserve the rankings.

\footnotetext{
Revision economic value is the deep ranking economic value minus the shallow ranking economic value.
A theme can therefore improve both forecasts but yield a negative $\Delta V$ if it benefits the shallow forecast more.
Appendix~\ref{appendix:ranking_value_components} provides the decomposition and separate shallow and deep theme losses.}

\section{Robustness Tests}
\label{sec:robustness}

Robustness tests change the stock universe, trading cost, or market state one at a time, with other settings kept the same throughout each exercise.
We report results by model family and depth group for market capitalization and trading frictions, and by model family and subperiod for market states.
The analysis focuses on differences between residual and feedforward models in the out-of-sample performance of value-weighted long-short portfolios, measured by Sharpe ratios in the first two tests and mean returns in the last.

\subsection{Market Capitalization}
\label{sec:robust_market_cap}

Market capitalization raises two concerns for the portfolios previously constructed in Section~\ref{sec:portfolio_performance}, one at each tail.
The first concern is that, at the lower tail, portfolio returns can be concentrated in small-cap stocks, which are prone to manipulation and exhibit unusually high cross-sectional returns \citep{FamaFrench2008Dissecting,hou2020ReplicatingAnomalies}.
Our previous results in Table~\ref{tab:FinalBaselinePerformance} show that the equal-weighted portfolios exhibit higher Sharpe ratios than those of the value-weighted portfolios.
However, limited liquidity can also make such excess returns difficult to exploit.
The second concern is that, at the upper tail, value-weighted portfolios allow a few large firms to dominate portfolio returns, making the measured returns unrepresentative of the broader cross-section \citep{avramov2012AnomaliesFinancialDistress}.
We remove either the smallest or the largest 20 percent of stocks from the stock universe each month based on current market capitalization, then give the models the trimmed universe to generate raw forecasts, and finally form value-weighted long-short portfolios based on the forecast deciles.

\input{SizeExclusionRobustness.tex}

Table~\ref{tab:SizeExclusionRobustness} shows that excluding the smallest 20 percent of stocks lowers the Sharpe ratio in every model and depth group.
This result is consistent with existing research: although the value-weighted portfolio allocates relatively low weights to small stocks, they remain a comparatively strong source of long-short portfolio returns \citep{FamaFrench2008Dissecting,hou2020ReplicatingAnomalies}.\footnotemark \
Removing the largest stocks raises every Sharpe ratio through two channels.
The first channel is the direct effect of removing large stocks that have substantial weights but contribute lower returns, and the second channel is reweighting the remaining stocks, which substantially increases the weights of smaller stocks that contribute higher returns.
\footnotetext{
In the value-weighted portfolio, the excluded smallest stocks account for only a small portion of portfolio value.
Removing the smallest stocks changes the relative value weights of the remaining stocks only modestly.
So, it can be concluded that the SR decline after excluding the smallest stocks can be attributed more directly to the loss of their strong returns.
}

Excluding the largest 20 percent of stocks makes the $\Delta$ SR statistically significant in every depth group for both comparison sets.
In the full universe, the ResNet versus NN differences in all three depth groups and the ResNet+ versus NN+ difference in the shallow group are not statistically significant, but each becomes significant after the largest stocks are excluded.
On the other hand, every $\Delta$ SR that is already significant in the full universe remains significant under either market capitalization exclusion.
These results show that residual models make the result more robust when a particular market capitalization segment is excluded.
And the residual models produce more accurate forecasts for smaller stocks.

Notice that, the deep NN+ Sharpe ratio falls to 0.79 after the smallest stocks are excluded, close to the deep NN Sharpe ratios of 0.82 in the full universe and 0.80 after the same exclusion.
This pattern suggests that, the deep NN+ converges toward NN when the extra gain from the tailored architecture gradually vanishes as depth increases, and the convergence is faster when the smallest stocks are excluded.

\subsection{Trading Frictions}
\label{sec:robust_trading_frictions}

The portfolio performance reported in Section~\ref{sec:portfolio_performance} is based on monthly rebalancing, with more than \$2 traded each month for every \$1 of previous portfolio value.\footnotemark \
However, the reported portfolio performance does not account for any trading frictions.
To examine whether the advantage of residual learning persists after accounting for trading frictions, we consider two linear transaction cost rates as trading frictions: 10 and 50 basis points (bps) per dollar traded.
For example, with monthly turnover of 260\%, these charges reduce monthly realized returns by 0.26\% and 1.30\%, respectively.

\footnotetext{See Equation~\eqref{eq:appendix_turnover}.
Starting with a \$1 capital base, our long-short portfolio has a \$1 long position and a \$1 short position, making it a zero-net-investment portfolio.
Turnover in the first month is 200\%.
In subsequent months, turnover includes both changes in portfolio constituents due to updated forecasts, and changes in the weights of existing stock holdings due to updated forecasts and return drift.
The mean monthly turnover values reported earlier all exceed 200\%, meaning that rebalancing the combined long and short positions requires more than \$2 traded per \$1 of previous portfolio value.}

\input{TransactionCostAdjustedPerformance.tex}

Table~\ref{tab:TransactionCostAdjustedPerformance} shows that transaction costs lower the Sharpe ratio for every model family and depth group, while increasing $\Delta$ SR within each depth group in both comparison sets.
The difference between ResNet+ and NN+ in the shallow group is not statistically significant without transaction costs but becomes significant at 50 bps.
Every significant $\Delta$ SR without transaction costs, remains significant after costs are introduced.
These results show that the advantage of residual learning persists under transaction costs, and the advantage widens as the cost rate rises.

And there are two apparent patterns in turnover and changes in Sharpe ratios.
One is that although residual models have turnover similar to that of the feedforward models within the same depth group, their Sharpe ratios decline less as transaction costs increase.
The other is that ResNet and NN have lower turnover than ResNet+ and NN+ and experience a smaller decline in the Sharpe ratio per unit of turnover.
These two patterns can be explained by one reason: transaction costs mainly reduce mean portfolio returns and have a relatively small effect on volatility.
Table~\ref{tab:FinalBaselinePerformance} reports higher portfolio volatility for ResNet and NN than for ResNet+ and NN+, and also reports higher volatility for residual models than for their feedforward counterparts within each depth group.
When volatility changes little, the same reduction in mean returns causes a smaller decline in the Sharpe ratio for a portfolio with higher volatility.

\subsection{Market States}
\label{sec:robust_market_states}

The market states analysis examines how the model mean return varies across expansions and recessions.
Using the National Bureau of Economic Research (NBER) business cycle chronology \citep{NBERBusinessCycleDating}, we divide the 443 out-of-sample realized return months from February 1987 through December 2023 into nine subperiods, consisting of five expansions and four recessions.
Within each comparison set, each model family's monthly value-weighted long-short returns are first averaged equally across the same width and depth specifications, and then averaged over time to obtain the reported ResNet(+) Mean and NN(+) Mean.\footnotemark

\footnotetext{Some subperiods contain too few monthly observations to estimate return volatility reliably, so here we report mean returns rather than Sharpe ratios.
Inference for mean return differences uses 12 Newey\mbox{-}West lags, and statistical significance is not reported for subperiods with fewer than 13 months.}

\input{MarketStateRobustness.tex}

Table~\ref{tab:MarketStateRobustness} shows that the ResNet+ versus NN+ mean return difference is positive in all five expansion periods, which together cover 407 of 443 months.
Four of the five differences are significant at the 5 percent level.
The two longest expansions account for more than half of all months, and both have significantly positive differences.
During recessions, however, only three of the four mean return differences are positive, and these differences are either statistically insignificant or based on too few months for inference.
The negative difference occurs in March and April 2020, the two months of the COVID recession; both ResNet+ and NN+ have negative mean returns in this subperiod, showing that the return patterns learned from earlier data do not translate into profitable portfolios during this short period.

For ResNet versus NN, the mean return difference is positive in four of the five expansion periods, covering 279 of 443 months, and three of these differences are significant at the 10 percent level.
During the two months of the COVID recession, ResNet also has a negative mean return, whereas NN has a positive mean return.
The short duration of this period makes it difficult to distinguish differences in forecasting ability from sampling variation.
Overall, residual models outperform their feedforward counterparts in most subperiods, but even models with strong and stable performance over the full sample can incur losses during extreme market states.

\section{Conclusions}

We propose that, shallow models are special cases of deep models.

A deeper model can at least reproduce its shallow counterpart with identity mappings $f(x) = x$ and has the potential to outperform the shallow counterpart.
And we achieve an out-of-sample Sharpe ratio of 2.07 for value-weighted long-short portfolios of deep residual models, higher than that for the corresponding shallow ones (1.92) and more than twice that of the deep feedforward models (0.89).

We tailor the feedforward architecture to the features of financial panel data and the training protocol, obtaining NN+ and its residual-learning counterpart, ResNet+.
Within each comparison set, we hold the data, preprocessing, training protocol, and portfolio evaluation methods fixed and compare residual and feedforward models across different depths and width schedules.
We find that residual models outperform their feedforward counterparts, and the edge of residual learning becomes sharper as depth increases, while the feedforward counterparts deteriorate as depth increases.
Further exploration shows that the residual models preserve the shallow forecasts and refine them as depth increases, while their feedforward counterparts fail to do so.
The results remain robust after excluding extremely large- or small-cap stocks and after incorporating transaction costs, and the results persist across most of the NBER business cycles.

The residual-learning framework makes it possible to train deeper neural networks with more parameters directly on structured panel data to output return forecasts, forming native ``large asset pricing models''.

\clearpage

\phantomsection
\addcontentsline{toc}{part}{Appendix}
\noindent{\Large\bfseries Appendix}
\par
\vspace{1em}

\appendix
\captionsetup{font=small, labelfont=bf}

\section{Data, Estimation, and Model Specification}

\subsection{Normalization and Training Protocol}
\label{appendix:normalization_implementation}

After filling missing values with monthly cross-sectional medians, we rank-normalize firm characteristic $x_{i,t}$ as
\begin{equation}
\widetilde x_{i,t}=2\frac{\operatorname{rank}_t(x_{i,t})}{N_t}-1.
\label{eq:appendix_rank_normalization}
\end{equation}
An entirely missing monthly cross-section is assigned zero, and macroeconomic predictors are not rank-normalized across stocks.
Batch and layer normalization take the form
\begin{equation}
\operatorname{Norm}(y)
=\gamma\odot\frac{y-\mu}{\sqrt{\sigma^2+\ep}}+\beta.
\end{equation}
For batch normalization, $(\mu,\sigma^2)$ use the observations in a mini-batch for each feature.
For layer normalization, $(\mu,\sigma^2)$ use the features within one hidden vector.
NN and the transition blocks of ResNet use batch normalization \citep{ioffe2015BatchNormalizationAccelerating}; NN+ and ResNet+ use layer normalization in their corresponding blocks and before the final forecast map \citep{ba2016layer}.

Annual January refits use the preceding year's fitted parameters as starting values.
Each training batch contains 10,000 observations.
We minimize mean squared error for at most 100 epochs, stopping after five consecutive epochs without improvement on the most recent five years of training data.
NN and ResNet use Adam with $\ell_1$ regularization; NN+ and ResNet+ use AdamW with $\ell_2$ regularization.
For each width and depth specification, we average the forecasts from ten random seeds before sorting stocks.

\subsection{ResNet+ Specification Details}
\label{appendix:resnetp_specification}

For $d=1,\dots,D$, the block update is
\begin{align}
h_0 & := \mathrm{Act}(\mathcal{I}(X)),
\label{eq:resnetp_init}
\\
h_d & :=
\begin{cases}
\mathrm{Act}(h_{d-1}+\mathcal{E}_d(h_{d-1})), & n_{d-1}=n_d, \\[4pt]
\mathrm{Act}\!\left(\mathcal{E}_d(h_{d-1})+\sum_{m=1}^{M_d}\mathcal{Q}_{d,m}(h_{d-1})\right), & n_{d-1}\neq n_d,
\end{cases}
\label{eq:resnetp_block}
\end{align}
where $\mathcal{I}(\cdot)$ is the initial map, $\mathrm{Act}(\cdot)$ is the activation function, $h_d \in \mathbb{R}^{n_d}$ is the hidden representation after block $d$, and $n_d$ is the hidden width.
At a width change, $\mathcal{E}_d$ maps to width $n_d$; if it is the only branch, $M_d=0$.
Additional depth at the same transition adds $M_d\geq1$ projection paths $\mathcal{Q}_{d,m}(h)=\mathrm{Act}(\mathrm{LN}_{d,m}(W_{d,m}h))$ from the same input $h_{d-1}$.

\begin{figure}[!htbp]
    \centering
    \includegraphics[width=0.9\textwidth]{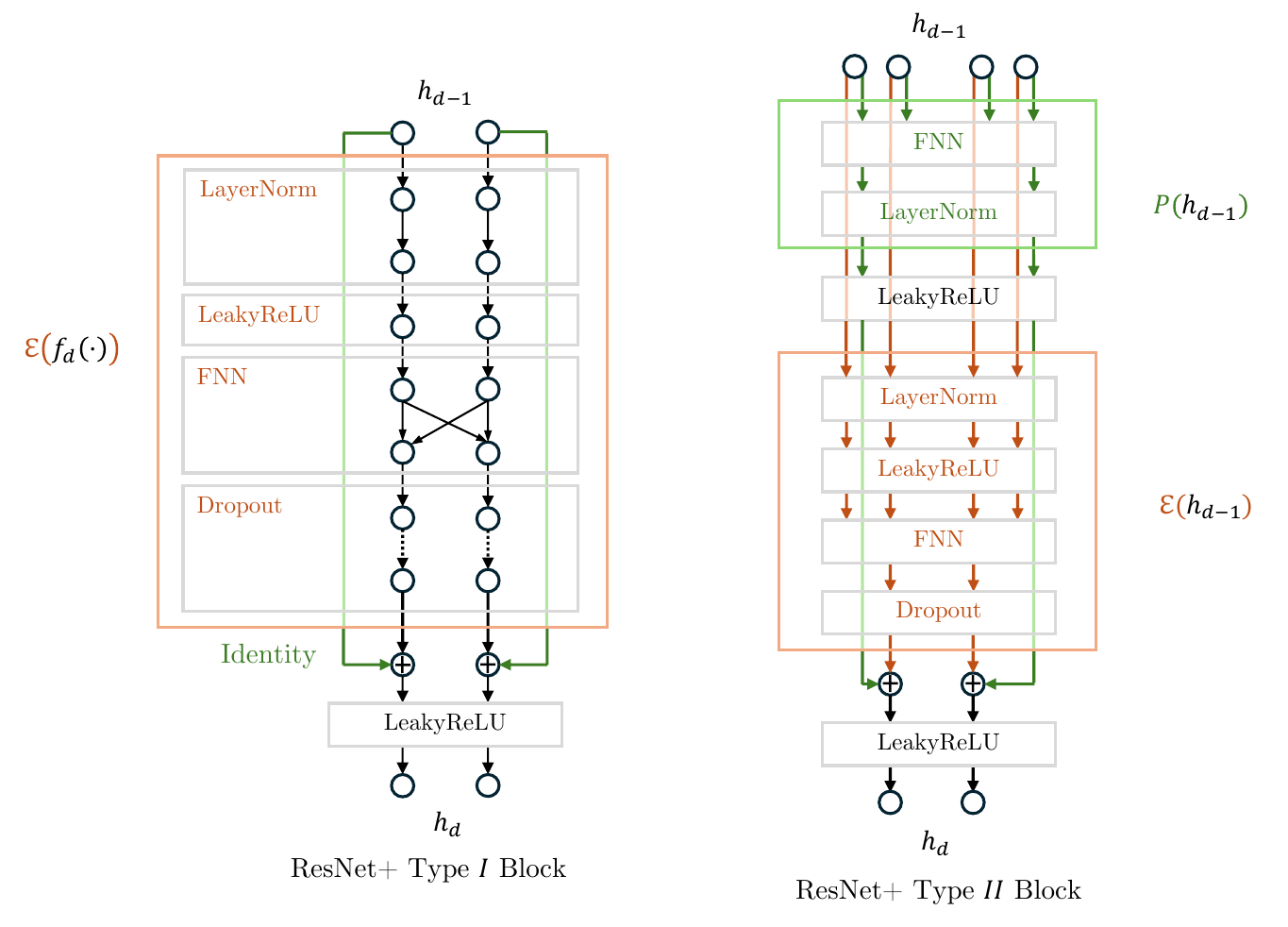}
\caption{
\textbf{ResNet+ Forecasting Blocks.}
\newline\small
The left panel shows Type \emph{I}, which adds the learned residual update $\mathcal{E}_d(h_{d-1})$ to the identity shortcut path when the hidden width is unchanged.
Type \emph{II} shows a width change with additional layers: $\mathcal{P}$ denotes $\mathcal{Q}_{d,m}$ in equation~\eqref{eq:resnetp_block}, and FNN denotes the feedforward branch $\mathcal{E}_d$.
The circled ``+'' denotes elementwise addition.
    }
    \label{fig:CoreBuildingBlockResNetPlus}
\end{figure}

The activation function $\mathrm{Act}(\cdot)$ is LeakyReLU.
Residual branches use layer normalization and dropout \citep{Srivastava2014}.

The return forecast is
\begin{equation}
\widehat R_{i,t+1}:=\mathcal{O}(h_{D,i,t}),
\label{eq:resnetp_output}
\end{equation}
where $\mathcal{O}(\cdot)$ is the final return forecast map.

\FloatBarrier

\clearpage
\begin{landscape}
\subsection{Predictor List}
\label{appendix:predictor_list}

\input{PredictorList.tex}
\end{landscape}

\FloatBarrier

\clearpage
\section{Measurement and Construction Details}
\label{appendix:performance_metrics}

\subsection{Forecast and Portfolio Measures}
\label{appendix:forecast_portfolio_measures}

Let $\widehat R_{i,t+1}$ and $R_{i,t+1}$ denote the forecast and realized one-month-ahead excess returns.
Each month, we exclude stocks with forecasts below the 2.5th or above the 97.5th percentile.
On the retained stock set $\mathcal U_t^{\mathrm{trim}}$, define
\begin{equation}
\widehat R^\circ_{i,t+1}
=\widehat R_{i,t+1}
-\frac{1}{|\mathcal U_t^{\mathrm{trim}}|}
\sum_{j\in\mathcal U_t^{\mathrm{trim}}}\widehat R_{j,t+1}.
\end{equation}
Realized returns are not demeaned.
Let $\mathcal O^{\mathrm{trim}}=\{(i,t):i\in\mathcal U_t^{\mathrm{trim}}\}$.
One nonnegative scale is fitted for each specification,
\begin{equation}
\widehat\kappa
=\max\left\{0,
\frac{\sum_{\mathcal O^{\mathrm{trim}}}\widehat R^\circ_{i,t+1}R_{i,t+1}}
{\sum_{\mathcal O^{\mathrm{trim}}}(\widehat R^\circ_{i,t+1})^2}
\right\}.
\end{equation}
The calibrated out-of-sample statistic is
\begin{equation}
R^2_{\mathrm{OOS}}
=1-
\frac{\sum_{\mathcal O^{\mathrm{trim}}}
(R_{i,t+1}-\widehat\kappa\widehat R^\circ_{i,t+1})^2}
{\sum_{\mathcal O^{\mathrm{trim}}}R_{i,t+1}^2}.
\label{eq:appendix_oos_r2}
\end{equation}
The tables report $100R^2_{\mathrm{OOS}}$.
We scale forecasts only when calculating $R^2_{\mathrm{OOS}}$.

For the common monthly stock sample, let $\mathcal P_t=\{(i,j):i<j\}$ and $|\mathcal P_t|=\binom{N_t}{2}$.
With shallow forecast $\widehat R^S$, deeper forecast $\widehat R^D$, and realized return $R$, define
\begin{align}
S_{ij,t}&=\operatorname{sgn}(\widehat R^S_{i,t+1}-\widehat R^S_{j,t+1}),\\
D_{ij,t}&=\operatorname{sgn}(\widehat R^D_{i,t+1}-\widehat R^D_{j,t+1}),\\
Z_{ij,t+1}&=\operatorname{sgn}(R_{i,t+1}-R_{j,t+1}),\\
\mathcal C_t&=\{(i,j)\in\mathcal P_t:S_{ij,t}\ne D_{ij,t}\}.
\label{eq:appendix_pairwise_signs}
\end{align}
Then
\begin{equation}
P_t=1-\frac{|\mathcal C_t|}{|\mathcal P_t|},
\qquad
C_{t+1}
=\frac{1}{|\mathcal C_t|}
\sum_{(i,j)\in\mathcal C_t}
\frac{(D_{ij,t}-S_{ij,t})Z_{ij,t+1}}{2}.
\end{equation}
Months with $\mathcal C_t=\varnothing$ do not enter the within-specification mean of $C_{t+1}$.

Each month, stocks are sorted into ten forecast deciles $S_{t+1,k}$.
For $w\in\{\mathrm{VW},\mathrm{EW}\}$,
\begin{equation}
R^w_{t+1,k}=\sum_{i\in S_{t+1,k}}w^w_{i,t,k}R_{i,t+1},
\quad
w^{\mathrm{EW}}_{i,t,k}=\frac{1}{|S_{t+1,k}|},
\quad
w^{\mathrm{VW}}_{i,t,k}=\frac{\mathrm{ME}_{i,t}}{\sum_{j\in S_{t+1,k}}\mathrm{ME}_{j,t}}.
\label{eq:appendix_decile_return}
\end{equation}
The long-short return and annualized Sharpe ratio are
\begin{equation}
R^w_{t+1,\mathrm{HL}}=R^w_{t+1,10}-R^w_{t+1,1},
\qquad
\mathrm{SR}^w=\sqrt{12}\frac{\bar R^w_{\mathrm{HL}}}{\sigma^w_{\mathrm{HL}}}.
\label{eq:appendix_hl_sr}
\end{equation}
The Fama--French five-factor alpha satisfies
\begin{equation}
R^w_{t+1,\mathrm{HL}}
=\alpha^w_{\mathrm{FF5}}+\beta^{w\top}f_{t+1}+\varepsilon_{t+1},
\label{eq:appendix_ff5_alpha}
\end{equation}
where $f_{t+1}$ contains $\mathrm{Mkt}{-}\mathrm{RF}$, SMB, HML, RMW, and CMA.
We report monthly alpha in percent.
Its $t$ statistic uses Newey\mbox{-}West standard errors with 12 lags.
For Table~\ref{tab:FinalBaselinePerformance}, monthly long-short returns are averaged equally across the included regular specifications within each model and depth group before one regression is estimated.

If $a_{i,t}$ is the signed long-short weight, turnover is
\begin{equation}
a^{\mathrm{drift}}_{i,t+1}
=\frac{a_{i,t}(1+R_{i,t+1})}{1+\sum_j a_{j,t}R_{j,t+1}},
\qquad
\mathrm{TO}_{t+1}=\sum_i|a_{i,t+1}-a^{\mathrm{drift}}_{i,t+1}|.
\label{eq:appendix_turnover}
\end{equation}
The long and short legs sum to one and minus one, and first-month gross turnover is $\mathrm{TO}_1=2$.

\subsection{Comparison Sets, Depth Summaries, and Heatmaps}
\label{appendix:table_figure_construction}
\label{appendix:depth_summaries_heatmaps}

Let $j=(s,d)$ index width schedule $s\in\{1,\ldots,4\}$ and depth $d$.
For residual model family $r$, feedforward model family $b$, and depth set $A$, define
\begin{equation}
\mathcal C_{r,b}(A)
=\{(s,d):s\in\{1,\ldots,4\},\ d\in A,\
\mathrm{SR}^{\mathrm{VW}}_{r,s,d}\text{ and }
\mathrm{SR}^{\mathrm{VW}}_{b,s,d}\text{ are regular}\}.
\label{eq:appendix_completed_pairs}
\end{equation}
We exclude specifications that produce constant forecasts from the portfolio performance tables and statistical tests.
Ties are allowed if forecasts still differ across stocks.
The shallow, medium, and deep groups cover hidden depths $1$--$5$, $6$--$10$, and $11$--$20$.
Within each model family, completed regular specifications receive equal weight.

Let $J(A)=|\mathcal C_{r,b}(A)|$.
The depth-group Sharpe ratio difference and monthly return difference are
\begin{align}
\Delta\mathrm{SR}_{r,b}(A)
&=\frac{1}{J(A)}\sum_{(s,d)\in\mathcal C_{r,b}(A)}
(\mathrm{SR}^{\mathrm{VW}}_{r,s,d}-\mathrm{SR}^{\mathrm{VW}}_{b,s,d}),
\label{eq:appendix_depth_delta_sr}\\
\Delta R^{r,b}_{t+1}(A)
&=\frac{1}{J(A)}\sum_{(s,d)\in\mathcal C_{r,b}(A)}
(R^{\mathrm{VW}}_{r,s,d,t+1,\mathrm{HL}}-R^{\mathrm{VW}}_{b,s,d,t+1,\mathrm{HL}}).
\label{eq:appendix_cell_avg_delta_return}
\end{align}
The FF5 alpha difference is the intercept in a five-factor regression of $\Delta R^{r,b}_{t+1}(A)$.
The heatmaps report
\begin{equation}
H_{r,b}(s,d)=\mathrm{SR}^{\mathrm{VW}}_{r,s,d}-\mathrm{SR}^{\mathrm{VW}}_{b,s,d}.
\label{eq:appendix_heatmap_delta}
\end{equation}

\subsection{Sharpe Ratio Difference Inference}
\label{appendix:sharpe_difference_inference}

The residual-minus-feedforward tests target equation~\eqref{eq:appendix_depth_delta_sr}.
For model $m$, let $\mathcal C_m(A)$ be its regular specifications in depth set $A$, with $S=\{1,\ldots,5\}$ and $D=\{11,\ldots,20\}$.
The deep-minus-shallow comparison in Table~\ref{tab:FinalBaselinePerformance} tests
\begin{equation}
\Delta\mathrm{SR}^{D-S}_m
=\frac{1}{|\mathcal C_m(D)|}\sum_{j\in\mathcal C_m(D)}\mathrm{SR}^{\mathrm{VW}}_{m,j}
-\frac{1}{|\mathcal C_m(S)|}\sum_{j\in\mathcal C_m(S)}\mathrm{SR}^{\mathrm{VW}}_{m,j}.
\end{equation}
For both comparisons, we calculate each specification's Sharpe ratio before averaging across specifications.
For model $m$, define
\begin{equation}
\widehat{\mathrm{IF}}^m_{j,t}
=\sqrt{12}\left[
\frac{u^m_{j,t}}{\widehat\sigma^m_j}
-\frac{\widehat\mu^m_j\{(u^m_{j,t})^2-(\widehat\sigma^m_j)^2\}}
{2(\widehat\sigma^m_j)^3}
\right],
\quad
u^m_{j,t}=R^m_{j,t}-\widehat\mu^m_j.
\end{equation}
The two monthly influence series are
\begin{align}
\overline\ell^{r,b}_t(A)
&=\frac{1}{J(A)}\sum_{j\in\mathcal C_{r,b}(A)}
\left(\widehat{\mathrm{IF}}^r_{j,t}-\widehat{\mathrm{IF}}^b_{j,t}\right),\\
\overline\ell^{D-S}_{m,t}
&=\frac{1}{|\mathcal C_m(D)|}\sum_{j\in\mathcal C_m(D)}\widehat{\mathrm{IF}}^m_{j,t}
-\frac{1}{|\mathcal C_m(S)|}\sum_{j\in\mathcal C_m(S)}\widehat{\mathrm{IF}}^m_{j,t}.
\end{align}
For either contrast, let $\widehat s=\mathrm{SE}_{\mathrm{NW}(12)}(\overline\ell_t)$.
With $M=1{,}000$ circular block bootstrap draws and fixed block length $b=12$ months,
\begin{equation}
d=\frac{|\widehat\Delta|}{\widehat s},
\qquad
d^{*,q}=\frac{|\widehat\Delta^{*,q}-\widehat\Delta|}{\widehat s^{*,q}},
\qquad
p=\frac{1+\sum_{q=1}^{M}\mathds{1}\{d^{*,q}\ge d\}}{M+1}.
\label{eq:appendix_lw_pvalue}
\end{equation}
Each bootstrap draw samples the same months for all return series.
We adapt the studentized circular block method of \citet{LedoitWolf2008Sharpe} to the differences above using Bartlett Newey\mbox{-}West(12) standard errors.

\subsection{Robustness Measures}
\label{appendix:inference_robustness_tables}

Market capitalization robustness tests use the full universe and the universes obtained after excluding the smallest or largest 20 percent of stocks by formation-month market capitalization.
For transaction cost rate $\kappa$ basis points per dollar traded,
\begin{equation}
R^{\mathrm{net}}_{f,s,d,t+1}(\kappa)
=R^{\mathrm{VW}}_{f,s,d,t+1,\mathrm{HL}}
-\mathrm{TO}^{\mathrm{VW}}_{f,s,d,t}\frac{\kappa}{10{,}000},
\qquad \kappa\in\{0,10,50\}.
\label{eq:appendix_net_cost_return}
\end{equation}

\FloatBarrier

\section{Shallow and Deep Ranking Economic Values}
\label{appendix:ranking_value_components}

\subsection{Definitions and Relation to Revision Economic Value}

For a common monthly stock sample, let $\mathcal P_t=\{(i,j):i<j\}$ and $n_t=|\mathcal P_t|$.
Define the shallow and deep forecast ordering signs as
\begin{equation}
\begin{aligned}
S_{ij,t}&=\operatorname{sgn}(\widehat R^S_{i,t+1}-\widehat R^S_{j,t+1}),\\
D_{ij,t}&=\operatorname{sgn}(\widehat R^D_{i,t+1}-\widehat R^D_{j,t+1}).
\end{aligned}
\label{eq:appendix_value_orderings}
\end{equation}
The shallow and deep ranking economic values are
\begin{equation}
\begin{aligned}
V_{t+1}^{S}&=\frac{1}{2n_t}\sum_{(i,j)\in\mathcal P_t}
S_{ij,t}(R_{i,t+1}-R_{j,t+1}),\\
V_{t+1}^{D}&=\frac{1}{2n_t}\sum_{(i,j)\in\mathcal P_t}
D_{ij,t}(R_{i,t+1}-R_{j,t+1}).
\end{aligned}
\label{eq:appendix_individual_ranking_values}
\end{equation}
These are measures of forecast rankings, not portfolio returns.

Let $\mathcal C_t=\{(i,j)\in\mathcal P_t:S_{ij,t}\ne D_{ij,t}\}$ and $m_t=|\mathcal C_t|$.
For a changed pair, the pairwise revision score is
\begin{equation}
c_{ij,t+1}=\frac{D_{ij,t}-S_{ij,t}}{2}
\operatorname{sgn}(R_{i,t+1}-R_{j,t+1}).
\label{eq:appendix_value_revision_score}
\end{equation}
Since unchanged pairs have $D_{ij,t}=S_{ij,t}$ and $\operatorname{sgn}(x)|x|=x$, subtraction yields
\begin{equation}
\begin{aligned}
V_{t+1}^{D}-V_{t+1}^{S}
&=\frac{1}{2n_t}\sum_{(i,j)\in\mathcal P_t}
(D_{ij,t}-S_{ij,t})(R_{i,t+1}-R_{j,t+1})\\
&=\frac{1}{n_t}\sum_{(i,j)\in\mathcal C_t}
c_{ij,t+1}|R_{i,t+1}-R_{j,t+1}|\\
&=\frac{1}{n_t}\sum_{k=1}^{m_t}
c_{k,t+1}|R_{i_k,t+1}-R_{j_k,t+1}|
=V_{t+1}.
\end{aligned}
\label{eq:appendix_value_decomposition}
\end{equation}
\subsection{Ranking Economic Values and Theme Losses}

Table~\ref{tab:appendix_forecast_ranking_values} reports the shallow and deep ranking economic values using the same specification samples and aggregation weights as Table~\ref{tab:ForecastRefinement}.

\input{AppendixForecastRankingValues.tex}
\FloatBarrier

Table~\ref{tab:appendix_theme_value_losses} separates the shallow and deep ranking economic value losses underlying Table~\ref{tab:theme_pairwise_losses}, which satisfy
\begin{equation}
\Delta V_{g,t+1}=\Delta V^D_{g,t+1}-\Delta V^S_{g,t+1}.
\label{eq:appendix_theme_value_decomposition}
\end{equation}
\input{AppendixThemeValueLosses.tex}
\FloatBarrier

\section{Theoretical Performance Bounds}
\label{appendix:theoretical_performance_bounds}

In this section, we first introduce the general idea of residual blocks preventing gradient vanishing, e.g.~\citet{he2016deep,he2016identity}, then we study the stability of ResNet+ by providing supporting theoretical results. 

Referring to~\citet{he2016deep,he2016identity}, we consider a pre-activation residual network with $D$ residual blocks, let $(h_d)_{d=0}^{D-1}$ be the sequence of hidden states, it can be formulated as below.
$$
h_d:=h_{d-1}+\psi_d\left(\theta;h_{d-1}\right),
$$
where $\theta$ is the learnable parameter of the neural networks. 
Denoting $\mathcal{L}$ the loss function acts at the final layer, by chain rule, we can get the backpropagation updates follow
\begin{equation}
\label{eqn:loss:change}
\nabla_{h_{d-1}}\mathcal{L} = \left[\frac{\partial h_{d}}{\partial h_{d-1}}\right]^{\top} \nabla_{h_{d}}\mathcal{L} = I + J_{\psi_d}(h_{d-1})^{\top} \nabla_{h_{d}}\mathcal{L}, \qquad J_{\psi_d}(h_{d-1}):= \frac{\partial \psi_{\theta}(h_{d-1})}{\partial h_{d-1}},
\end{equation}
where $J$ is the Jacobian matrix describing how a tiny change in the input features propagates through all the operations inside the residual block. Hence, even if the derivative $\frac{\partial \psi_{\theta}(h_{d-1})}{\partial h_{d-1}}$ is small, the backpropagation updates still preserve gradient magnitude in the sense that $\nabla_{h_{d-1}}\mathcal{L} \approx \nabla_{h_d}\mathcal{L}$. Comparing to nonresidual networks, e.g.~\citet{gu2020EmpiricalAssetPricing}, where the backpropagation updates compose many Jacobians such that $$
\nabla_{h_{d-1}}\mathcal{L} = \left[\frac{\partial h_{d}}{\partial h_{d-1}}\right]^{\top} \nabla_{h_{d}}\mathcal{L} = J_{\psi_d}(h_{d-1})^{\top} \nabla_{h_{d}}\mathcal{L},
$$
with $J_{\psi_d}(h_{d-1}):= \frac{\partial \psi_{\theta}(h_{d-1})}{\partial h_{d-1}}$. After iterating over many layers, we can have 
$$\nabla_{h_0}\mathcal{L} = \prod_{i=1}^{d-1}J_i^{\top}\nabla_{h_{d-1}}\mathcal{L}.
$$
Therefore, once the network equips with with many activation functions and initializations, each Jacobian usually admits $\Vert J_i \Vert < 1$ under~\citet{he2016deep} setting, also see~\citet{he2015delving}, then the products of Jacobian shrinks exponentially with depth, the gradient vanishes and make deep networks untrainable.

Now we adapt the residual structure to our learning model and specify the residual function, the ResNet architecture is showing in \autoref{fig:resnet}. Under the setting, the effective unit block of ResNet+ is formulated by
\begin{equation}
\label{def:psi}
h_d:=h_{d-1}+\psi_d\left(\theta;h_{d-1}\right), \qquad \psi = \mathrm{FNN} \circ \mathrm{LeakyReLU} \circ \mathrm{LayerNorm},
\end{equation} 
For $a\in(0,1)$, define $\phi(u)=\mathrm{LeakyReLU}_a(u)$ element-wise by
\begin{equation}
\label{eqn:lrelu}
\phi(u_i) = \begin{cases}
u_i, & u_i \geq  0,\\
a u_i, & u_i<0.
\end{cases}
\end{equation}
And $\mathrm{LN}$ is the layer normalization (LayerNorm, LN) proposed by~\citet{ba2016layer}. Let $x\in\mathbb{R}^d$ and $\mathbf{1}$ be all-one vector in $d$-dimension, denoting
$$
\alpha(x)=\frac{1}{d}\mathbf{1}^T x, \quad \hat x=x-\alpha(x)\mathbf{1},
$$
and for any $\ep > 0$, we further define
$$
\sigma(x)^2=\tfrac{1}{d}\Vert  \hat x\Vert ^2, \quad v(x)=\sqrt{\sigma(x)^2+\ep}.
$$
LN function with gain $\gamma\in\mathbb{R}^d$ and bias $\beta\in\mathbb{R}^d$ is defined as below,
\begin{equation}
\label{eqn:layer:norm}
\mathrm{LN}_{\gamma,\,\beta}(x) = \gamma\odot \frac{\hat x}{v(x)}+\beta,
\end{equation}
where $\odot$ is the element-wise product. We note that the parameters $\gamma,\beta \in \mathbb{R}^d$ are learnable and LN function is nonlinear, see~\citet{ni2024nonlinearity}. In the following, we compute the lower bound of the norm of Jacobian $J_{\psi_d}(h_{d-1}):=I + \frac{\partial \psi_{\theta}(h_{d-1})}{\partial h_{d-1}}$ to show the nonvanishing gradient. We first compute Jacobian with respect to LN function and give a norm bound in the following lemma.
\begin{lemma}[LayerNorm Jacobian]
\label{lemma:ln:jacobian}
Under the setting, assume $\Vert \gamma \Vert_{\infty}< \infty $ and $\Vert \beta \Vert_{\infty}< \infty $, the Jacobian of~~$\mathrm{LN}_{\gamma,\,\beta}$ satisfies, for all $x\in\mathbb{R}^n$,
\begin{equation}
J_{\mathrm{LN}}(x) = D_\gamma\left(\frac{1}{v(x)}\left(I - \frac{1}{n}\mathbf{1}\mathbf{1}^T\right)-\frac{1}{d v(x)^3} \hat{x}\hat{x}^T\right),~\text{ where $D_\gamma=\mathrm{diag}(\gamma)$}.
\end{equation}
Moreover, its operator norm is uniformly bounded as
\begin{equation}
\label{eqn:ln:lip}
\Vert J_{\mathrm{LN}}(x)\Vert  \leq \frac{\Vert \gamma\Vert _\infty}{\sqrt{\ep}}\left(1 + \frac{1}{\sqrt{n}}\right),~\text{ for all $x\in\mathbb{R}^n$},
\end{equation}
so $\mathrm{LN}_{\gamma,\,\beta}$ is Lipschitz and depth-independent with constant
$\frac{2}{\sqrt{\ep}}\Vert \gamma \Vert_\infty$.
\end{lemma}
\begin{proof}
We denote $D_{\gamma} = \mathrm{diag}(\gamma)$, a diagonal matrix with $\gamma$ on its diagonal line, and the orthogonal projector $P := I-\frac{1}{n}\mathbf{1}\mathbf{1}^{T}$ onto the mean-zero subspace. Then by the definition, we can have
\begin{equation}
\hat x = Px \quad\text{ and }\quad \sigma(x)^2 = \frac{1}{n}\Vert Px \Vert^2 = \frac{1}{n}xPx^T,
\end{equation}
where we note that $P^2 = P$ by the symmetric and idempotent. Hence, LN function follows
\begin{equation}
\mathrm{LN}_{\gamma,\,\beta}(x)  =  \frac{1}{v(x)}D_\gamma P x + \beta, \quad v(x) = \sqrt{\sigma(x)^2 + \ep}.
\end{equation}
Now we differentiate to get
\begin{align}
J_{\mathrm{LN}}(x) & = D_\gamma\left( \frac{1}{v(x)}P - \frac{1}{v(x)^2}\left(\nabla v(x)\right)\,(Px)^T \right)
= D_\gamma\left(\frac{1}{v(x)} P-\frac{1}{n v(x)^3} \hat{x}\hat{x}^T\right),
\end{align}
where we can compute
\begin{equation}
\nabla v(x) = \frac{1}{2v(x)}\nabla\sigma(x)^2=\frac{1}{2v(x)} \cdot \frac{2}{n}Px = \frac{1}{dv(x)}Px.
\end{equation}
Moreover, we compute the norm of Jacobian, by the definition, we have $\Vert D_\gamma\Vert =\Vert \gamma\Vert_\infty = \sup_{i \in [h]} |\gamma_i|$ and $\Vert \hat x\hat x^T\Vert =\Vert \hat x\Vert ^2=n\sigma(x)^2 \leq n v(x)^2$. Recall $P = I - \frac{1}{n}\mathbf{1}\mathbf{1}^T$, we can have $\Vert P\Vert \leq 1 + \frac{1}{\sqrt{n}}$. So we can get
\begin{align}
\Vert J_{\mathrm{LN}}(x)\Vert & = \left\Vert D_\gamma\left(\frac{1}{v(x)} P-\frac{1}{n v(x)^3} \hat{x}\hat{x}^T\right) \right\Vert
\nonumber
\\
& \leq \Vert \gamma\Vert _\infty\left(\frac{1 + 1/\sqrt{n}}{v(x)} + \frac{v(x)^2}{v(x)^3}\right)
\leq \frac{\Vert \gamma\Vert _\infty}{\sqrt{\ep}}\left(1 + \frac{1}{\sqrt{n}}\right).
\end{align}
Thus $\mathrm{LN}_{\gamma,\,\beta}$ is Lipschitz. And the proof is complete.
\end{proof}
In our Jacobian bound~\eqref{eqn:ln:lip}, this norm bound of Jacobian is uniformly controlled by $\Vert \gamma\Vert _\infty/\sqrt{\ep}$. We notice from LN function defined in~\eqref{eqn:layer:norm}, LayerNorm plays the role of normalization over the training, hence, we have $\Vert \gamma \Vert_{\infty}/\sqrt{\ep}$ is a moderate term in the sense that for any small $\ep$, the gain parameter $\Vert \gamma \Vert_{\infty}$ is relatively small.


We assume weights are initialized with bounded variance and  the spectral norms of $\Vert W_d\Vert $ are controlled under the \autoref{assump:second-moment}.
\begin{assumption}[The boundedness of second moment]
\label{assump:second-moment}
For each residual block $d\in\{0,\dots,D-1\}$, let $W_d$ be the affine weight matrix in $F_d(h)=W_d h+b_d$. Assume $\mathbb{E}\left\Vert W_d \right\Vert^2 : = C_{\Vert W_d \Vert} < \infty$.
\end{assumption}
This assumption is mild and natural in the setting of initialization and normalization and it is standard in random-matrix-based analyses of deep networks, see~\citet{vershynin2018high,hanin2019finite,xiao2018dynamical,bartlett2017spectrally,lu2017expressive,ruthotto2020deep,telgarsky2016benefits,yang2020feature}. For example,~\citet{vershynin2018high,hanin2019finite} initializes weights by the computation per neuron distributed over its inputs, in this case, the spectral norm concentrates around $\mathcal{O}(1)$ with high probability, and $\mathbb{E} \Vert W_l \Vert^2<\infty$ holds automatically. 


\begin{corollary}[Jacobian norm of $\psi$]
\label{cor:ln:lrelu:affine}
At each layer $h_d$, $d = 0,1,\ldots,D-1$, let $\mathrm{LN}_{\gamma,\,\beta}:\mathbb{R}^{n_d} \rightarrow \mathbb{R}^{n_d}$ be per-sampled LayerNorm, $\phi_d=\mathrm{LeakyReLU}_a$ with $a\in(0,1)$, and let $F(z)=W_d z + b_d$ be affine with $W_d\in\mathbb{R}^{m_d \times n_d}$, $b_d\in\mathbb{R}^{m_d}$ such that $\Vert W_d \Vert < \infty$ and $\Vert b_d\Vert < \infty$. $\psi$ is the composed function defined early in~\eqref{def:psi}. Then, for $h_d$, $d = 0,1,\ldots,D-1$, it holds
\begin{equation}
\label{eqn:Jacobian:sharp}
\Vert J_\psi(h_d)\Vert \leq \Vert W_d\Vert \frac{\Vert \gamma_d \Vert_\infty}{\sqrt{\ep}}\left(1+\frac{1}{\sqrt{n_d}}\right), \quad d = 0,1,\ldots, D-1.
\end{equation}
\end{corollary}

\begin{proof}
Denoting by $W$ the weighted matrix of the FNN such that $F(h) = Wh + b$, then the Jacobian is $J_F(h) = W$ with $(W)_{ij} = \frac{\partial (F(h))_i}{\partial h_j}=W_{ij}$. Moreover, the Jacobian for LeakyReLU is $J_\phi(u)=\mathrm{diag}(\phi'(u))$ follows its definition in~\eqref{eqn:lrelu} with entries in $\{a,1\}$ where the subgradient at $0$ is chosen in $[a,1]$. Hence, $\left\Vert J_\phi(u) \right\Vert \leq \max\{1,a\} = 1$ for $a\in(0,1)$. We finally complete the proof by using chain rule:
\begin{align}
\Vert J_\psi(h) \Vert = \Vert J_F J_{\phi}(\mathrm{LN}(h)) J_{\mathrm{LN}}(h)\Vert
\leq \Vert J_F \Vert \Vert J_{\phi}(\mathrm{LN}(h)) \Vert \Vert J_{\mathrm{LN}}(h) \Vert
\leq \Vert W \Vert \frac{\Vert \gamma\Vert_\infty}{\sqrt{\ep}}\left(1+\frac{1}{\sqrt{n}}\right),
\label{eqn:ln:lrelu:affine}
\end{align}
where $J_F = W$, $\left\Vert J_\phi(u) \right\Vert \leq \max\{1,a\} = 1$ and \autoref{lemma:ln:jacobian}  to get Jacobian of LN. The proof is complete.
\end{proof}
Considering the full input to output Jacobians after $D$ blocks is the product as the following.
\begin{equation}
\frac{\partial}{\partial h_0}h_D  =  \frac{\partial}{\partial h_{D-1}}h_D \cdot \frac{\partial}{\partial h_{D-2}}h_{D-1} \cdots \frac{\partial}{\partial h_{0}}h_1= \prod_{d=1}^D \Big(I+ \,J_\psi(h_{d-1})\Big).
\end{equation}
Then we show the stability of ResNet+ by studying the norm of $\prod_{d=1}^D \Big(I+ \,J_\psi(h_{d-1})\Big)$. 

\begin{theorem}[The stability of gradients]
Under \autoref{assump:second-moment}, there exists $s \geq e > r > 0$ holding
\begin{equation}
0 < r\Vert \nabla_{h_D}\mathcal{L} \Vert \leq \Vert \nabla_{h_0}\mathcal{L} \Vert < s\Vert \nabla_{h_D}\mathcal{L} \Vert <\infty, \quad\text{with positive probability.}
\end{equation}
There is no exploding or vanishing in gradients updates through backpropagation updates.
\end{theorem}

\begin{proof}
\textbf{Nonvanishing gradients.} Denoting matrices $J_d:=J_{\psi}(h_d)$, $d=0,\ldots,d-1$, noticing the smallest eigenvalue $\sigma_{\min}(I + J_d) \geq 1 - \Vert J_d\Vert$, we can have 
$
\left\Vert \prod_{d=0}^{D-1} \Big(I+ \,J_\psi(h_{d-1})\Big)\right\Vert \geq \prod_{d=0}^{D-1}\left( 1 - \Vert J_d \Vert \right),
$
then for $0 < r < 1$, the event 
$$
\left\{\max _d\left\Vert J_d\right\Vert  \leq 1-r^{1 / D}\right\} \subseteq \left\{\prod_{d=0}^{D-1}\left(1-\left\Vert J_d\right\Vert\right) \geq r\right\}.
$$
Using \autoref{cor:ln:lrelu:affine}, we can compute 
\begin{align}
\mathbb{P}\left(\prod_{d=0}^{D-1}\left( 1 - \Vert J_d \Vert \right) \geq r \right) & \geq \mathbb{P}\left(\max _d\left\Vert J_d\right\Vert \leq 1-r^{1/D}\right)
\nonumber
\\
& \geq \mathbb{P}\left(\max_d\left(\Vert W_d \Vert \frac{\Vert \gamma_d\Vert_\infty}{\sqrt{\ep}}\left(1+\frac{1}{\sqrt{n_d}}\right)\right) \leq 1-r^{1/D} \right),
\end{align}
where the probability is taken over the distribution of $W_d$ and $\gamma_d$ denotes the LayerNorm gain parameter and $\ep > 0$ is the stabilization constant, $\Vert\gamma_d \Vert_{\infty}/\sqrt{\ep}$ is moderate. In particular, under the common continuous initializations of $W_d$ such as Gaussian, the setup in~\citet{he2016deep}, this usually does have positive probability.
Hence, we can compute from~\eqref{eqn:loss:change}, the backpropagation updates follow
\begin{equation}
\Vert \nabla_{h_0}\mathcal{L} \Vert \geq \prod_{d=0}^{D-1}(1 - \Vert J_d \Vert)\Vert \nabla_{h_D}\mathcal{L} \Vert \geq r\Vert \nabla_{h_D}\mathcal{L} \Vert > 0, \quad\text{with positive probability.}
\end{equation}

\textbf{Nonexploding gradients.} Moreover, we can compute
\begin{equation}
\left\Vert \prod_{d=0}^{D-1} (I+ J_d) \right\Vert \leq \prod_{d=0}^{D-1} \Vert I+ J_d \Vert \leq \prod_{d=0}^{D-1} \left(1+ \Vert J_l \Vert \right)
 \leq \exp\left( \sum_{d=0}^{D-1} \Vert J_d \Vert \right),
\end{equation}
where we use $1+x \leq e^x$ to get the last inequality. Thus for the threshold $s>0$,
\begin{align}
\mathbb{P}\left(\left\Vert \prod_{d=0}^{D-1} (I+ J_d) \right\Vert \geq s \right)
& \leq \mathbb{P}\left(\exp\left( \sum_{d=0}^{D-1} \Vert J_d \Vert \right) \geq s \right)
\nonumber
\\
& \leq \mathbb{P}\left(\sum_{d=0}^{D-1} \Vert W_d \Vert \frac{\Vert \gamma_d \Vert_\infty}{\sqrt{\ep}}\left(1+\frac{1}{\sqrt{n_d}}\right) \geq \ln s \right)
\nonumber
\\
& \leq \mathbb{P}\left(\left(\sum_{d=0}^{D-1} \Vert W_d \Vert^2\right)^{1/2}\left(\sum_{d=0}^{D-1}\frac{\Vert \gamma_d\Vert^2_\infty}{\ep}\left(1+\frac{1}{\sqrt{n_d}}\right)^2 \right)^{1/2} \geq \ln s \right)
\nonumber
\\
& \leq \mathbb{E}\left[\sum_{d=0}^{D-1} \left\Vert W_d \right\Vert^2\right]\frac{1}{(\ln s)^2} \cdot \sum_{d=0}^{D-1}\frac{\Vert \gamma_d\Vert^2_\infty}{\ep}\left(1+\frac{1}{\sqrt{n_d}}\right)^2
\nonumber
\\
& \leq \left(\sum_{d=0}^{D-1} C_{\Vert W_d \Vert}\right) \cdot \frac{1}{(\ln s)^2} \cdot \sum_{d=0}^{D-1}\frac{\Vert \gamma_d\Vert^2_\infty}{\ep}\left(1+\frac{1}{\sqrt{n_d}}\right)^2,
\end{align}
where we use~\ref{eqn:ln:lrelu:affine} to get the second inequality, and then we use Cauchy-Schwartz inequality to get the third inequality and finally we use Markov's inequality and \autoref{assump:second-moment}, where the second moment is bounded, to get the last two inequalities. Hence, we get from ~\eqref{eqn:loss:change}, the backpropagation updates follow
\begin{equation}
\Vert \nabla_{h_0}\mathcal{L} \Vert \leq \left\Vert \prod_{d=0}^{D-1} (I + J_d) \right\Vert \Vert \nabla_{h_D}\mathcal{L} \Vert  < s, \quad\text{with positive probability.}
\end{equation}
The proof is complete.
\end{proof}

\FloatBarrier

\section{Supplementary Depth and Portfolio Displays}
\label{appendix:supplementary_depth_figures}

\subsection{Forecast Ranking Diagnostic}
\label{appendix:s1d7_trainability_diagnostics}

Table~\ref{tab:S1D7DiagnosticSummary} compares NN and ResNet at depths 3 and 7 under width schedule 1.
Forecast dispersion determines whether the fitted model can distinguish stocks for portfolio sorting.
The comparison holds the width schedule and training protocol fixed while varying architecture and depth.

\input{S1D7DiagnosticSummary.tex}

\FloatBarrier

\subsection{Equal-Weighted Portfolios}
\label{appendix:equal_weighted_portfolios}

Table~\ref{tab:EqualWeightedRobustness} repeats the paired portfolio tests using equal-weighted portfolios.

\input{EqualWeightedRobustness.tex}

\FloatBarrier

\clearpage
\begin{landscape}
\subsection{Width and Depth Profiles}

Tables~\ref{tab:PairedSpecificationDepthProfile} and~\ref{tab:ResNetVsNNDepthProfile} report results by width and depth for the completed specifications in the grouped summaries and heatmaps.

\input{PairedSpecificationDepthProfile.tex}

\input{ResNetVsNNDepthProfile.tex}
\end{landscape}

\FloatBarrier

\clearpage
\section{Model Scale Reference}
\label{appendix:parameter_scale}

Figure~\ref{fig:ParameterScaleGap} shows the fitted parameter ranges for our four models alongside familiar asset pricing specifications and published neural networks.

\begin{figure}[!htbp]
\centering
\includegraphics[width=\textwidth]{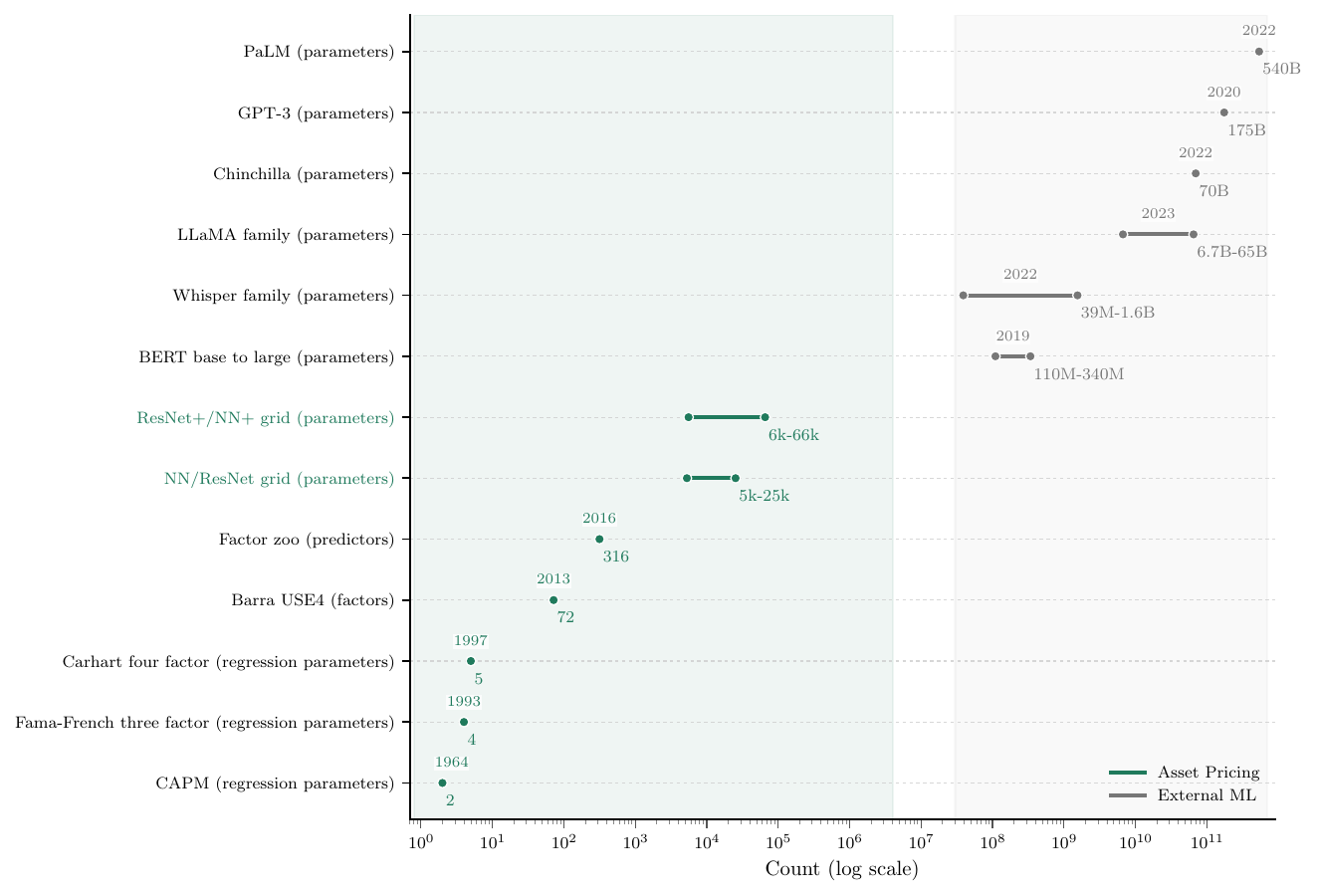}
\caption{
\textbf{Model Scale Reference for Asset Pricing Specifications.}
\newline\small
The horizontal axis uses a logarithmic scale.
CAPM, the Fama--French model, and the Carhart model count regression parameters; Barra USE4 counts factors, and the factor zoo counts predictors.
The remaining rows count neural-network parameters.
The models are from \citet{sharpe1964CapitalAssetPrices,FamaFrench1993,Carhart1997,msci2013BarraUSE4,harvey2016CrossSectionExpectedReturns,devlin2019BERT,radford2023RobustSpeechRecognition,touvron2023LLaMA,hoffmann2022TrainingComputeOptimal,brown2020LanguageModelsAre,chowdhery2022PaLM}
}
\label{fig:ParameterScaleGap}
\end{figure}

\FloatBarrier

\bibliographystyle{chicago}
\bibliography{ref}
\end{document}

%% file: ExperimentalDesignGrid.tex
\begin{table}[htbp]
\centering
\scriptsize
\setstretch{1.0}
\renewcommand{\arraystretch}{1.05}
\setlength{\tabcolsep}{0pt}
\captionsetup{skip=4pt,justification=raggedright,singlelinecheck=false,labelsep=newline,labelfont=bf,textfont=bf}
\caption{\textbf{Empirical Design}}
\label{tab:ExperimentalDesignGrid}
\begin{tikzpicture}
\node[inner sep=0pt, outer sep=0pt] (designmatrix) {%
\begin{tabular}{@{}>{\raggedright\arraybackslash}p{\dimexpr14\textwidth/62\relax}@{}>{\raggedright\arraybackslash}p{\dimexpr15\textwidth/62\relax}@{\hspace{6pt}}>{\centering\arraybackslash}p{\dimexpr16\textwidth/62\relax}@{}>{\centering\arraybackslash}p{\dimexpr16\textwidth/62\relax}@{}}
\toprule
 &  & ResNet vs NN & ResNet+ vs NN+ \\
\midrule
\multirow{7}{*}{Data} & Forecast Target & \multicolumn{2}{@{}>{\centering\arraybackslash}p{\dimexpr32\textwidth/62\relax}@{}}{Excess Return\textsuperscript{a}} \\
 & Forecast Horizon & \multicolumn{2}{@{}>{\centering\arraybackslash}p{\dimexpr32\textwidth/62\relax}@{}}{One Month} \\
	 & Predictors & \multicolumn{2}{@{}>{\centering\arraybackslash}p{\dimexpr32\textwidth/62\relax}@{}}{153 JKP + 8 Macro\textsuperscript{b}} \\
	 & Sample & \multicolumn{2}{@{}>{\centering\arraybackslash}p{\dimexpr32\textwidth/62\relax}@{}}{Jan 1963--Dec 2023} \\
	 & Universe & \multicolumn{2}{@{}>{\centering\arraybackslash}p{\dimexpr32\textwidth/62\relax}@{}}{U.S. Common Stocks} \\
	 & Forecast Origins & \multicolumn{2}{@{}>{\centering\arraybackslash}p{\dimexpr32\textwidth/62\relax}@{}}{Jan 1987--Nov 2023} \\
	 & Realized Returns & \multicolumn{2}{@{}>{\centering\arraybackslash}p{\dimexpr32\textwidth/62\relax}@{}}{Feb 1987--Dec 2023} \\
\arrayrulecolor{gray!45}\midrule[0.25pt]\arrayrulecolor{black}
\multirow{3}{*}{Training} & Training Window & \multicolumn{2}{@{}>{\centering\arraybackslash}p{\dimexpr32\textwidth/62\relax}@{}}{Annually Expanding} \\
 & Loss Function & \multicolumn{2}{@{}>{\centering\arraybackslash}p{\dimexpr32\textwidth/62\relax}@{}}{Mean Squared Error} \\
 & Forecast Averaging & \multicolumn{2}{@{}>{\centering\arraybackslash}p{\dimexpr32\textwidth/62\relax}@{}}{10 Ensemble Members} \\
\arrayrulecolor{gray!45}\midrule[0.25pt]\arrayrulecolor{black}
\multirow{2}{*}{Variation} & Width Schedules & \multicolumn{2}{@{}>{\centering\arraybackslash}p{\dimexpr32\textwidth/62\relax}@{}}{{[32]}; [32, 16]; [32, 16, 8]; [32, 16, 8, 4]} \\
 & Depths & \multicolumn{2}{@{}>{\centering\arraybackslash}p{\dimexpr32\textwidth/62\relax}@{}}{1--20\textsuperscript{c}} \\
\arrayrulecolor{gray!45}\midrule[0.25pt]\arrayrulecolor{black}
\multirow{3}{*}{Model} & Feedforward Model & NN\textsuperscript{d} & NN+\textsuperscript{d} \\
 & Residual Model & ResNet\textsuperscript{d} & ResNet+\textsuperscript{d} \\
 & Residual Shortcut Path & Identity & Identity / Projection \\
\bottomrule
\end{tabular}%
};
\draw[line width=\arrayrulewidth] ([xshift=\dimexpr29\textwidth/62\relax]designmatrix.north west) -- ([xshift=\dimexpr29\textwidth/62\relax]designmatrix.south west);
\end{tikzpicture}
\par\vspace{4pt}
\begin{minipage}{\textwidth}
\setstretch{0.96}\footnotesize Centered entries apply to both comparison sets; side entries are specific to one comparison set.
\textsuperscript{a} Excess return is measured relative to the 30-day Treasury bill return.
\textsuperscript{b} JKP denotes the firm characteristics from \citet{jensen2023ThereReplicationCrisisa}; Macro denotes the eight monthly predictors following \citet{welch2008ComprehensiveLookEmpirical}.
\textsuperscript{c} Depth reports hidden layer counts; width schedules 1 to 4 start at depths 1 to 4 and all extend through 20.
\textsuperscript{d} NN is the \citet{gu2020EmpiricalAssetPricing} feedforward neural network benchmark; NN+ is the feedforward counterpart to ResNet+; ResNet and ResNet+ add a residual learning mechanism to the corresponding feedforward specifications.
\end{minipage}
\end{table}

%% file: FinalBaselinePerformance.tex
\begin{table}[!htb]
\centering
\scriptsize
\setstretch{1.0}
\renewcommand{\arraystretch}{1.06}
\setlength{\tabcolsep}{2pt}
\newcolumntype{Q}[1]{>{\hsize=#1\hsize\linewidth=\hsize\Centering\arraybackslash}X}
\newcolumntype{D}[1]{>{\hsize=#1\hsize\linewidth=\hsize\Centering\arraybackslash}X}
\newcommand{\finalbaselinepaneldivider}{\tikz[baseline=-0.6ex]{\draw[gray!65,dash pattern=on 1.2pt off 1.2pt,line width=0.25pt] (0,-0.9ex) -- (0,1.7ex);}}
\captionsetup{skip=3pt,justification=raggedright,singlelinecheck=false,labelsep=newline,labelfont=bf,textfont=bf}
\caption{\textbf{Long-Short Portfolio Performance by Depth Group and Model Family}}
\label{tab:FinalBaselinePerformance}
\begin{tabularx}{\linewidth}{@{}>{\Centering\arraybackslash}p{44pt}@{\hspace{5pt}\finalbaselinepaneldivider\hspace{5pt}}cQ{1.05}Q{0.95}Q{0.85}D{1.10}Q{0.85}Q{0.95}Q{1.00}Q{1.25}@{}}
\toprule
 & \mbox{Model} & \mbox{Mean} & \mbox{SD} & \mbox{SR} & \mbox{$\Delta$ SR} & \mbox{SR(EW)} & \mbox{$R^2_{\mathrm{OOS}}$} & \mbox{FF5 $\alpha$} & \mbox{Turnover} \\
\midrule[0.55pt]
\multirow{4}{*}{Shallow} & NN & 1.86 & 6.89 & 0.94 & \multirow{2}{=}{\centering\hspace{4pt}0.009} & 2.31 & 0.21 & \ensuremath{2.00^{***}} & 212.22 \\
 & ResNet & 1.92 & 7.02 & 0.95 &  & 2.40 & 0.22 & \ensuremath{2.04^{***}} & 211.75 \\
\arrayrulecolor{gray!45}
\cmidrule[0.25pt](l{2pt}){2-10}
\arrayrulecolor{black}
 & NN+ & 3.15 & 5.80 & 1.84 & \multirow{2}{=}{\centering\hspace{4pt}0.085} & 3.28 & 0.24 & \ensuremath{3.03^{***}} & 257.53 \\
 & ResNet+ & 3.36 & 6.03 & 1.92 &  & 3.53 & 0.30 & \ensuremath{3.28^{***}} & 255.68 \\
\addlinespace[2pt]
\midrule[0.55pt]
\multirow{4}{*}{Medium} & NN & 1.81 & 6.92 & 0.90 & \multirow{2}{=}{\centering\hspace{4pt}0.046} & 2.29 & 0.20 & \ensuremath{1.94^{***}} & 209.64 \\
 & ResNet & 1.92 & 7.09 & 0.94 &  & 2.41 & 0.21 & \ensuremath{2.03^{***}} & 212.21 \\
\arrayrulecolor{gray!45}
\cmidrule[0.25pt](l{2pt}){2-10}
\arrayrulecolor{black}
 & NN+ & 2.29 & 5.18 & 1.47 & \multirow{2}{=}{\centering\hspace{4pt}\ensuremath{0.472^{***}}} & 2.60 & 0.16 & \ensuremath{2.09^{***}} & 250.78 \\
 & ResNet+ & 3.25 & 5.75 & 1.94 &  & 3.52 & 0.29 & \ensuremath{3.09^{***}} & 256.67 \\
\addlinespace[2pt]
\midrule[0.55pt]
\multirow{4}{*}{Deep} & NN & 1.65 & 6.94 & 0.82 & \multirow{2}{=}{\centering\hspace{4pt}0.100} & 2.13 & 0.18 & \ensuremath{1.81^{***}} & 209.58 \\
 & ResNet & 1.90 & 7.11 & 0.93 &  & 2.38 & 0.21 & \ensuremath{2.01^{***}} & 210.47 \\
\arrayrulecolor{gray!45}
\cmidrule[0.25pt](l{2pt}){2-10}
\arrayrulecolor{black}
 & NN+ & 1.33 & 4.32 & 0.89 & \multirow{2}{=}{\centering\hspace{4pt}\ensuremath{1.179^{***}}} & 1.60 & 0.09 & \ensuremath{1.26^{***}} & 261.97 \\
 & ResNet+ & 3.48 & 5.82 & 2.07 &  & 3.83 & 0.33 & \ensuremath{3.35^{***}} & 259.49 \\
\bottomrule[0.55pt]
\end{tabularx}
\par\nointerlineskip
{\scriptsize
\renewcommand{\arraystretch}{1.06}
\newcommand{\depthcontrastdivider}{\tikz[baseline=-0.6ex]{\draw[gray!65,dash pattern=on 1.2pt off 1.2pt,line width=0.25pt] (0,-0.9ex) -- (0,1.7ex);}}
\begin{tabularx}{\linewidth}{@{}>{\Centering\arraybackslash}p{130pt}@{\hspace{5pt}\depthcontrastdivider\hspace{5pt}}*{4}{>{\Centering\arraybackslash}X}@{}}
\noalign{\vskip\aboverulesep}
\multirow{2}{=}[-3pt]{\centering Deep SR minus Shallow SR} & NN & ResNet & NN+ & ResNet+ \\
\arrayrulecolor{gray!45}
\cmidrule[0.25pt](l{2pt}){2-5}
\arrayrulecolor{black}
 & \ensuremath{-0.11^{**}} & -0.02 & \ensuremath{-0.94^{***}} & \ensuremath{0.15^{**}} \\
\bottomrule
\end{tabularx}
}
\par\vspace{4pt}
\begin{minipage}{\linewidth}
\setstretch{1.0}\footnotesize Mean and SD are the monthly mean and standard deviation of value\mbox{-}weighted long-short portfolio returns.
SR and SR(EW) are the annualized Sharpe ratios for value\mbox{-}weighted and equal\mbox{-}weighted long-short portfolios, respectively.
Except for the portfolios underlying SR(EW), all portfolios are value\mbox{-}weighted.
$\Delta$ SR is the value\mbox{-}weighted residual-minus-feedforward Sharpe ratio difference.
Within each depth group, model-family results use each model's available width and depth specifications, whereas $\Delta$ SR uses the set completed by both models; specifications receive equal weight.
Deep SR minus Shallow SR reports, within each model family, the mean Sharpe ratio across available Deep specifications minus the corresponding mean across available Shallow specifications; specifications receive equal weight within each group.
Stars on $\Delta$ SR and Deep SR minus Shallow SR use two-sided Ledoit--Wolf studentized circular block bootstrap tests that the respective contrasts equal zero, with 12-month blocks; stars on FF5 $\alpha$ use two-sided tests with Newey\mbox{-}West standard errors and 12 lags.
*, **, and *** denote the 10\%, 5\%, and 1\% levels.
Mean, SD, FF5 $\alpha$, and Turnover are monthly percentages; $R^2_{\mathrm{OOS}}$ is reported in percent.
Constant forecasts are excluded.
\end{minipage}
\end{table}

%% file: FullDecileMonotonicity.tex
\begin{table}[p]
\centering
\setstretch{1.0}
\fontsize{8.0}{9.5}\selectfont
\renewcommand{\arraystretch}{0.97}
\setlength{\tabcolsep}{2.5pt}
\captionsetup{width=\linewidth,font=footnotesize,skip=5pt,justification=raggedright,singlelinecheck=false,labelsep=newline,labelfont=bf,textfont=bf}
\caption{\textbf{Decile Portfolio Performance by Depth Group and Model Family}}
\label{tab:FullDecileMonotonicity}
\begin{tabular*}{\linewidth}{@{\extracolsep{\fill}}l*{12}{r}@{}}
NN & \multicolumn{4}{c}{Shallow} & \multicolumn{4}{c}{Medium} & \multicolumn{4}{c}{Deep} \\
\cmidrule(r){1-1} \cmidrule(lr){2-5} \cmidrule(lr){6-9} \cmidrule(lr){10-13}
Rank & \multicolumn{1}{c}{Fcst} & \multicolumn{1}{c}{Mean} & \multicolumn{1}{c}{SD} & \multicolumn{1}{c}{SR} & \multicolumn{1}{c}{Fcst} & \multicolumn{1}{c}{Mean} & \multicolumn{1}{c}{SD} & \multicolumn{1}{c}{SR} & \multicolumn{1}{c}{Fcst} & \multicolumn{1}{c}{Mean} & \multicolumn{1}{c}{SD} & \multicolumn{1}{c}{SR} \\
\midrule[0.45pt]
Low & -2.74 & -0.56 & 8.54 & -0.23 & -2.18 & -0.53 & 8.50 & -0.21 & -1.00 & -0.42 & 8.08 & -0.18 \\
2 & -1.34 & 0.07 & 6.62 & 0.04 & -1.10 & 0.07 & 6.68 & 0.04 & -0.44 & 0.14 & 6.58 & 0.08 \\
3 & -0.59 & 0.44 & 5.68 & 0.27 & -0.49 & 0.46 & 5.72 & 0.28 & -0.10 & 0.45 & 5.62 & 0.28 \\
4 & -0.07 & 0.53 & 5.15 & 0.36 & -0.05 & 0.53 & 5.17 & 0.36 & 0.16 & 0.54 & 5.10 & 0.37 \\
5 & 0.35 & 0.67 & 4.73 & 0.49 & 0.30 & 0.65 & 4.77 & 0.47 & 0.39 & 0.66 & 4.78 & 0.48 \\
6 & 0.72 & 0.72 & 4.58 & 0.55 & 0.62 & 0.72 & 4.57 & 0.55 & 0.61 & 0.71 & 4.58 & 0.54 \\
7 & 1.10 & 0.75 & 4.63 & 0.56 & 0.95 & 0.77 & 4.61 & 0.58 & 0.84 & 0.81 & 4.57 & 0.61 \\
8 & 1.51 & 0.83 & 4.85 & 0.60 & 1.31 & 0.82 & 4.82 & 0.59 & 1.09 & 0.86 & 4.73 & 0.63 \\
9 & 2.05 & 0.95 & 5.31 & 0.62 & 1.77 & 0.94 & 5.17 & 0.63 & 1.42 & 0.89 & 4.92 & 0.63 \\
High & 2.98 & 1.30 & 6.05 & 0.75 & 2.59 & 1.27 & 5.88 & 0.75 & 1.96 & 1.23 & 5.46 & 0.78 \\
\addlinespace[3pt]
H$-$L & 5.71 & 1.86 & 6.89 & 0.94 & 4.77 & 1.81 & 6.92 & 0.90 & 2.96 & 1.65 & 6.94 & 0.82 \\
\midrule[0.45pt]
\addlinespace[5pt]
ResNet & \multicolumn{4}{c}{Shallow} & \multicolumn{4}{c}{Medium} & \multicolumn{4}{c}{Deep} \\
\cmidrule(r){1-1} \cmidrule(lr){2-5} \cmidrule(lr){6-9} \cmidrule(lr){10-13}
Rank & \multicolumn{1}{c}{Fcst} & \multicolumn{1}{c}{Mean} & \multicolumn{1}{c}{SD} & \multicolumn{1}{c}{SR} & \multicolumn{1}{c}{Fcst} & \multicolumn{1}{c}{Mean} & \multicolumn{1}{c}{SD} & \multicolumn{1}{c}{SR} & \multicolumn{1}{c}{Fcst} & \multicolumn{1}{c}{Mean} & \multicolumn{1}{c}{SD} & \multicolumn{1}{c}{SR} \\
\midrule[0.45pt]
Low & -3.19 & -0.63 & 8.72 & -0.25 & -3.18 & -0.66 & 8.86 & -0.26 & -3.19 & -0.65 & 8.87 & -0.25 \\
2 & -1.61 & 0.04 & 6.76 & 0.02 & -1.56 & 0.06 & 6.80 & 0.03 & -1.57 & 0.06 & 6.79 & 0.03 \\
3 & -0.78 & 0.42 & 5.76 & 0.25 & -0.72 & 0.38 & 5.83 & 0.23 & -0.73 & 0.37 & 5.87 & 0.22 \\
4 & -0.21 & 0.52 & 5.18 & 0.35 & -0.15 & 0.52 & 5.23 & 0.34 & -0.16 & 0.54 & 5.26 & 0.36 \\
5 & 0.24 & 0.68 & 4.73 & 0.50 & 0.30 & 0.65 & 4.77 & 0.48 & 0.29 & 0.63 & 4.81 & 0.45 \\
6 & 0.64 & 0.71 & 4.58 & 0.53 & 0.70 & 0.72 & 4.60 & 0.54 & 0.69 & 0.72 & 4.60 & 0.54 \\
7 & 1.04 & 0.74 & 4.61 & 0.56 & 1.09 & 0.74 & 4.64 & 0.55 & 1.08 & 0.72 & 4.65 & 0.54 \\
8 & 1.48 & 0.83 & 4.83 & 0.59 & 1.52 & 0.80 & 4.88 & 0.57 & 1.51 & 0.78 & 4.86 & 0.56 \\
9 & 2.05 & 0.98 & 5.32 & 0.64 & 2.07 & 0.97 & 5.35 & 0.63 & 2.05 & 0.96 & 5.35 & 0.62 \\
High & 3.04 & 1.29 & 6.12 & 0.73 & 3.04 & 1.26 & 6.20 & 0.71 & 3.00 & 1.25 & 6.19 & 0.70 \\
\addlinespace[3pt]
H$-$L & 6.23 & 1.92 & 7.02 & 0.95 & 6.22 & 1.92 & 7.09 & 0.94 & 6.19 & 1.90 & 7.11 & 0.93 \\
\midrule[0.45pt]
\addlinespace[5pt]
NN+ & \multicolumn{4}{c}{Shallow} & \multicolumn{4}{c}{Medium} & \multicolumn{4}{c}{Deep} \\
\cmidrule(r){1-1} \cmidrule(lr){2-5} \cmidrule(lr){6-9} \cmidrule(lr){10-13}
Rank & \multicolumn{1}{c}{Fcst} & \multicolumn{1}{c}{Mean} & \multicolumn{1}{c}{SD} & \multicolumn{1}{c}{SR} & \multicolumn{1}{c}{Fcst} & \multicolumn{1}{c}{Mean} & \multicolumn{1}{c}{SD} & \multicolumn{1}{c}{SR} & \multicolumn{1}{c}{Fcst} & \multicolumn{1}{c}{Mean} & \multicolumn{1}{c}{SD} & \multicolumn{1}{c}{SR} \\
\midrule[0.45pt]
Low & -1.62 & -1.13 & 8.70 & -0.44 & -0.96 & -0.72 & 8.40 & -0.27 & -0.57 & -0.15 & 7.40 & 0.00 \\
2 & -0.69 & -0.20 & 7.06 & -0.09 & -0.28 & -0.06 & 6.85 & -0.02 & -0.06 & 0.22 & 6.19 & 0.15 \\
3 & -0.27 & 0.21 & 5.65 & 0.13 & 0.01 & 0.28 & 5.53 & 0.18 & 0.12 & 0.44 & 5.26 & 0.30 \\
4 & -0.01 & 0.46 & 4.88 & 0.33 & 0.18 & 0.48 & 4.79 & 0.35 & 0.22 & 0.58 & 4.79 & 0.42 \\
5 & 0.18 & 0.60 & 4.55 & 0.46 & 0.29 & 0.67 & 4.50 & 0.51 & 0.28 & 0.65 & 4.63 & 0.49 \\
6 & 0.35 & 0.81 & 4.43 & 0.64 & 0.39 & 0.80 & 4.51 & 0.62 & 0.34 & 0.76 & 4.66 & 0.56 \\
7 & 0.52 & 0.94 & 4.66 & 0.70 & 0.49 & 0.91 & 4.66 & 0.68 & 0.40 & 0.83 & 4.75 & 0.61 \\
8 & 0.72 & 1.11 & 5.10 & 0.75 & 0.61 & 1.05 & 4.95 & 0.74 & 0.48 & 0.92 & 4.97 & 0.64 \\
9 & 1.01 & 1.41 & 5.98 & 0.81 & 0.79 & 1.25 & 5.56 & 0.78 & 0.59 & 0.98 & 5.35 & 0.62 \\
High & 1.75 & 2.02 & 8.25 & 0.84 & 1.20 & 1.57 & 7.15 & 0.74 & 0.86 & 1.18 & 6.43 & 0.61 \\
\addlinespace[3pt]
H$-$L & 3.37 & 3.15 & 5.80 & 1.84 & 2.16 & 2.29 & 5.18 & 1.47 & 1.42 & 1.33 & 4.32 & 0.89 \\
\midrule[0.45pt]
\addlinespace[5pt]
ResNet+ & \multicolumn{4}{c}{Shallow} & \multicolumn{4}{c}{Medium} & \multicolumn{4}{c}{Deep} \\
\cmidrule(r){1-1} \cmidrule(lr){2-5} \cmidrule(lr){6-9} \cmidrule(lr){10-13}
Rank & \multicolumn{1}{c}{Fcst} & \multicolumn{1}{c}{Mean} & \multicolumn{1}{c}{SD} & \multicolumn{1}{c}{SR} & \multicolumn{1}{c}{Fcst} & \multicolumn{1}{c}{Mean} & \multicolumn{1}{c}{SD} & \multicolumn{1}{c}{SR} & \multicolumn{1}{c}{Fcst} & \multicolumn{1}{c}{Mean} & \multicolumn{1}{c}{SD} & \multicolumn{1}{c}{SR} \\
\midrule[0.45pt]
Low & -1.95 & -1.27 & 8.62 & -0.51 & -2.00 & -1.26 & 8.56 & -0.51 & -2.32 & -1.41 & 8.71 & -0.56 \\
2 & -0.89 & -0.23 & 6.88 & -0.12 & -0.94 & -0.22 & 6.85 & -0.11 & -1.10 & -0.25 & 6.99 & -0.12 \\
3 & -0.40 & 0.24 & 5.52 & 0.15 & -0.44 & 0.24 & 5.48 & 0.15 & -0.54 & 0.22 & 5.51 & 0.14 \\
4 & -0.08 & 0.46 & 4.86 & 0.33 & -0.12 & 0.49 & 4.80 & 0.36 & -0.17 & 0.49 & 4.77 & 0.36 \\
5 & 0.17 & 0.62 & 4.57 & 0.47 & 0.13 & 0.66 & 4.54 & 0.51 & 0.11 & 0.70 & 4.55 & 0.53 \\
6 & 0.38 & 0.84 & 4.55 & 0.64 & 0.34 & 0.82 & 4.57 & 0.62 & 0.35 & 0.82 & 4.56 & 0.62 \\
7 & 0.60 & 0.97 & 4.79 & 0.70 & 0.56 & 0.95 & 4.81 & 0.69 & 0.60 & 0.96 & 4.78 & 0.70 \\
8 & 0.86 & 1.14 & 5.22 & 0.75 & 0.81 & 1.18 & 5.20 & 0.79 & 0.90 & 1.17 & 5.19 & 0.78 \\
9 & 1.23 & 1.42 & 6.20 & 0.79 & 1.18 & 1.38 & 6.04 & 0.79 & 1.32 & 1.37 & 5.95 & 0.80 \\
High & 2.13 & 2.09 & 8.40 & 0.86 & 2.04 & 1.98 & 8.08 & 0.85 & 2.32 & 2.08 & 8.17 & 0.88 \\
\addlinespace[3pt]
H$-$L & 4.08 & 3.36 & 6.03 & 1.92 & 4.04 & 3.25 & 5.75 & 1.94 & 4.63 & 3.48 & 5.82 & 2.07 \\
\midrule[0.45pt]
\end{tabular*}
\par\vspace{2pt}
\begin{minipage}{\linewidth}
\setstretch{1.0}\fontsize{8.0}{8.5}\selectfont From Low to High, the portfolios are formed long-only; H$-$L is the long-short portfolio formed as High minus Low.
Fcst is the mean return forecast; Mean and SD are the mean and standard deviation of realized returns. All three are monthly percentages; SR is annualized.
Each model family uses its completed nonconstant-forecast specifications, as in Table~\ref{tab:FinalBaselinePerformance}; eligible specifications can differ across families.
For portfolios other than H$-$L, statistics are computed separately for each width schedule and depth over the 443 out-of-sample months, then averaged equally within each depth group.
H$-$L Fcst and Mean are computed directly as High minus Low; H$-$L SD and SR are computed from the long-short return series.
\end{minipage}
\end{table}

%% file: ForecastRefinement.tex
\begin{table}[htbp]
\centering
\scriptsize
\setstretch{1.0}
\renewcommand{\arraystretch}{1.12}
\captionsetup{skip=4pt,justification=raggedright,singlelinecheck=false,labelsep=newline,labelfont=bf,textfont=bf}
\caption{\textbf{Forecast Preservation and Refinement by Depth Group}}
\label{tab:ForecastRefinement}
\newcommand{\forecastrefinementdivider}{\tikz[baseline=-0.6ex]{\draw[gray!65,dash pattern=on 1.2pt off 1.2pt,line width=0.25pt] (0,-0.9ex) -- (0,1.7ex);}}
\setlength{\tabcolsep}{2pt}
\begin{tabular*}{\linewidth}{@{\extracolsep{\fill}}l@{\hspace{4pt}\forecastrefinementdivider\hspace{4pt}}lccc@{\hspace{8pt}\forecastrefinementdivider\hspace{4pt}}lccc@{}}
\toprule
 & Model & $P_t$ & $C_{t+1} / \%$ & $V_{t+1} / bps$ & Model & $P_t$ & $C_{t+1} / \%$ & $V_{t+1} / bps$ \\
\midrule
\multirow{3}{*}{Medium} & ResNet & \makebox[3.8em][c]{0.93} & \makebox[3.8em][c]{0.32} & \makebox[3.8em][c]{-0.94} & ResNet+ & \makebox[3.8em][c]{0.81} & \makebox[3.8em][c]{\ensuremath{0.72^{***}}} & \makebox[3.8em][c]{0.68} \\
 & NN & \makebox[3.8em][c]{0.87} & \makebox[3.8em][c]{\ensuremath{-0.74^{*}}} & \makebox[3.8em][c]{\ensuremath{-2.07^{**}}} & NN+ & \makebox[3.8em][c]{0.75} & \makebox[3.8em][c]{-0.04} & \makebox[3.8em][c]{\ensuremath{-15.89^{***}}} \\
 & Difference & \makebox[3.8em][c]{\ensuremath{0.06^{***}}} & \makebox[3.8em][c]{\ensuremath{1.06^{**}}} & \makebox[3.8em][c]{1.13} & Difference & \makebox[3.8em][c]{\ensuremath{0.06^{***}}} & \makebox[3.8em][c]{0.75} & \makebox[3.8em][c]{\ensuremath{16.58^{***}}} \\
\addlinespace[1pt]
\arrayrulecolor{gray!45}
\midrule[0.25pt]
\arrayrulecolor{black}
\multirow{3}{*}{Deep} & ResNet & \makebox[3.8em][c]{0.92} & \makebox[3.8em][c]{\ensuremath{0.58^{**}}} & \makebox[3.8em][c]{\ensuremath{-1.14^{*}}} & ResNet+ & \makebox[3.8em][c]{0.79} & \makebox[3.8em][c]{\ensuremath{1.12^{***}}} & \makebox[3.8em][c]{\ensuremath{4.92^{***}}} \\
 & NN & \makebox[3.8em][c]{0.84} & \makebox[3.8em][c]{-0.99} & \makebox[3.8em][c]{-2.64} & NN+ & \makebox[3.8em][c]{0.64} & \makebox[3.8em][c]{\ensuremath{-1.88^{***}}} & \makebox[3.8em][c]{\ensuremath{-34.45^{***}}} \\
 & Difference & \makebox[3.8em][c]{\ensuremath{0.08^{***}}} & \makebox[3.8em][c]{\ensuremath{1.58^{**}}} & \makebox[3.8em][c]{1.50} & Difference & \makebox[3.8em][c]{\ensuremath{0.15^{***}}} & \makebox[3.8em][c]{\ensuremath{2.99^{***}}} & \makebox[3.8em][c]{\ensuremath{39.36^{***}}} \\
\bottomrule
\end{tabular*}
\par\vspace{4pt}
\begin{minipage}{\linewidth}
\setstretch{1.0}\footnotesize Preservation score $P_t$ is the share of stock pairs whose forecast ordering is unchanged from the corresponding shallow anchor forecasts.
Revision accuracy $C_{t+1}$ is the average pairwise revision score among changed stock pairs.
Revision economic value $V_{t+1}$ measures the economic value of these revisions.
The reported $C_{t+1}$ is in percent; $V_{t+1}$ is in monthly basis points.
$P_t$ is unscaled.
For each model, measures are first averaged over forecast origins for every combination of width schedule and target depth; the resulting specification means then receive equal weight within each depth group.
Difference is the residual model mean minus the feedforward model mean.
Stars test zero using Newey\mbox{-}West standard errors with 12 lags; *, **, and *** denote the 10\%, 5\%, and 1\% levels.
Significance levels are not reported for the model-family $P_t$ estimates.
\end{minipage}
\end{table}

%% file: ThemePairwiseLosses.tex
\begin{table}[!t]
\centering\scriptsize\setstretch{1.0}
\renewcommand{\arraystretch}{1.12}
\setlength{\tabcolsep}{2pt}
\captionsetup{skip=5pt,justification=raggedright,singlelinecheck=false,labelsep=newline,labelfont=bf,textfont=bf}
\caption{\textbf{Forecast Preservation and Revision Losses by Theme}}
\label{tab:theme_pairwise_losses}
\begin{tabular*}{\linewidth}{@{\extracolsep{\fill}}lccc@{\hspace{10pt}}ccc@{}}
\toprule
 & \multicolumn{3}{c}{ResNet} & \multicolumn{3}{c}{NN} \\
\cmidrule(lr){2-4}\cmidrule(lr){5-7}
Theme & $\Delta P_t$ & $\Delta C_{t+1} / \%$ & $\Delta V_{t+1} / bps$ & $\Delta P_t$ & $\Delta C_{t+1} / \%$ & $\Delta V_{t+1} / bps$ \\
\midrule
Accruals & $-0.001^{***}$ & $0.02$ & $0.05$ & $-0.001$ & $0.11$ & $0.01$ \\
Debt Issuance & $0.000$ & $0.04$ & $0.02$ & $0.001$ & $-0.01$ & $0.11$ \\
Investment & $-0.001^{*}$ & $0.01$ & $-0.23$ & $-0.001$ & $-0.16$ & $0.07$ \\
Low Leverage & $-0.003^{***}$ & $0.49^{***}$ & $0.40$ & $-0.011^{***}$ & $0.05$ & $0.57$ \\
Low Risk & $0.002^{**}$ & $-0.29$ & $-0.65$ & $-0.001$ & $-0.53$ & $0.54$ \\
Momentum & $0.003^{***}$ & $-0.36^{***}$ & $-0.32$ & $0.004$ & $0.22$ & $-2.83^{**}$ \\
Profit Growth & $0.001^{*}$ & $-0.01$ & $-0.01$ & $0.002^{*}$ & $0.14$ & $0.51$ \\
Profitability & $-0.003^{***}$ & $-0.05$ & $0.02$ & $-0.003^{***}$ & $0.41^{***}$ & $0.06$ \\
Quality & $0.001$ & $-0.31^{**}$ & $-0.14$ & $-0.002$ & $0.19$ & $0.70$ \\
Seasonality & $0.000$ & $-0.05$ & $-0.02$ & $-0.001$ & $-0.17^{**}$ & $0.02$ \\
Short-Term Reversal & $0.001$ & $0.08$ & $0.38^{*}$ & $0.001$ & $-0.02$ & $-1.59^{*}$ \\
Size & $-0.001$ & $0.22$ & $-0.15$ & $0.002$ & $-0.10$ & $-7.87^{***}$ \\
Value & $0.001$ & $-0.11$ & $-0.17$ & $-0.001$ & $-0.98^{***}$ & $-0.29$ \\
\midrule
\addlinespace[5pt]
 & \multicolumn{3}{c}{ResNet+} & \multicolumn{3}{c}{NN+} \\
\cmidrule(lr){2-4}\cmidrule(lr){5-7}
Theme & $\Delta P_t$ & $\Delta C_{t+1} / \%$ & $\Delta V_{t+1} / bps$ & $\Delta P_t$ & $\Delta C_{t+1} / \%$ & $\Delta V_{t+1} / bps$ \\
\midrule
Accruals & $0.001^{***}$ & $0.00$ & $-0.12$ & $0.001^{***}$ & $-0.08^{***}$ & $-0.13$ \\
Debt Issuance & $0.005^{***}$ & $-0.79^{***}$ & $0.28$ & $0.003^{***}$ & $-0.10^{***}$ & $-0.28$ \\
Investment & $0.001$ & $-1.33^{***}$ & $-1.67$ & $0.004^{***}$ & $-0.40^{***}$ & $-0.98$ \\
Low Leverage & $-0.007^{***}$ & $0.44^{***}$ & $1.83^{***}$ & $-0.003^{***}$ & $0.50^{***}$ & $-1.90$ \\
Low Risk & $0.007^{***}$ & $0.32$ & $0.26$ & $0.009^{***}$ & $-1.45^{***}$ & $-6.58^{***}$ \\
Momentum & $0.022^{***}$ & $0.10$ & $1.68$ & $-0.001$ & $-0.32^{***}$ & $-3.34^{***}$ \\
Profit Growth & $0.003^{***}$ & $1.48^{***}$ & $0.71$ & $0.001^{***}$ & $0.13^{***}$ & $-0.43$ \\
Profitability & $0.003^{***}$ & $-0.42^{***}$ & $1.17$ & $0.000$ & $0.10^{***}$ & $-0.22$ \\
Quality & $0.008^{***}$ & $1.67^{***}$ & $2.05$ & $0.002^{***}$ & $-0.58^{***}$ & $-2.32^{***}$ \\
Seasonality & $0.003^{***}$ & $0.16$ & $-0.66$ & $0.001^{***}$ & $-0.04$ & $-0.80$ \\
Short-Term Reversal & $0.012^{***}$ & $-0.46^{***}$ & $0.82$ & $0.004^{***}$ & $0.30^{*}$ & $-3.23^{*}$ \\
Size & $-0.005^{***}$ & $-0.05$ & $6.11^{***}$ & $-0.008^{***}$ & $1.63^{***}$ & $-4.32$ \\
Value & $0.017^{***}$ & $1.96^{***}$ & $0.30$ & $0.003^{***}$ & $-0.30^{*}$ & $-5.76^{***}$ \\
\bottomrule
\end{tabular*}
\par\vspace{5pt}
\begin{minipage}{\linewidth}\footnotesize\setstretch{1.0}
Entries are full-input minus theme-zeroed measures, with the same theme set to zero in both shallow and deep models and fitted models and other inputs held fixed.
$P$ is the preservation score, $C$ is revision accuracy among changed pairs, and $V$ is the revision economic value averaged over all stock pairs.
$\Delta P$ is unscaled, $\Delta C$ is in percent, and $\Delta V$ is in monthly basis points.
Target depths are 11--20, each compared with its minimum-depth anchor. Eligible specification means over 443 forecast origins receive equal weight.
Stars test each paired loss using Newey\mbox{-}West standard errors with 12 lags and Holm adjustment across 13 themes within each model and measure; *, **, and *** denote the 10\%, 5\%, and 1\% levels.
Cross-model differences are not tested. Theme effects are not additive, and $V$ is not a portfolio return.
\end{minipage}
\end{table}

%% file: SizeExclusionRobustness.tex
\begin{table}[h]
\centering
\scriptsize
\setstretch{1.0}
\renewcommand{\arraystretch}{1.12}
\newcolumntype{Z}[1]{>{\hsize=#1\hsize\linewidth=\hsize\Centering\arraybackslash}X}
\captionsetup{skip=4pt,justification=raggedright,singlelinecheck=false,labelsep=newline,labelfont=bf,textfont=bf}
\caption{\textbf{Market Capitalization Exclusion Tests}}
\label{tab:SizeExclusionRobustness}
\newcommand{\sizeexclusiondivider}{\tikz[baseline=-0.6ex]{\draw[gray!65,dash pattern=on 1.2pt off 1.2pt,line width=0.25pt] (0,-0.9ex) -- (0,1.7ex);}}
\begin{tabularx}{\linewidth}{@{}Z{1.20}Z{0.90}@{\hspace{5pt}\sizeexclusiondivider\hspace{5pt}}Z{0.95}Z{0.95}Z{1.00}@{}}
\toprule
Universe & Depth & $\Delta$ SR & ResNet SR & NN SR \\
\midrule
\multirow{3}{=}{\Centering All Stocks} & Shallow & \makebox[3.2em][c]{\ensuremath{0.01}} & 0.95 & 0.94 \\
 & Medium & \makebox[3.2em][c]{\ensuremath{0.05}} & 0.95 & 0.90 \\
 & Deep & \makebox[3.2em][c]{\ensuremath{0.10}} & 0.92 & 0.82 \\
\arrayrulecolor{gray!45}
\cmidrule[0.25pt]{1-5}
\arrayrulecolor{black}
\multirow{3}{=}{\Centering Exclude Smallest 20\%} & Shallow & \makebox[3.2em][c]{\ensuremath{0.00}} & 0.89 & 0.88 \\
 & Medium & \makebox[3.2em][c]{\ensuremath{0.04}} & 0.89 & 0.85 \\
 & Deep & \makebox[3.2em][c]{\ensuremath{0.07}} & 0.88 & 0.80 \\
\arrayrulecolor{gray!45}
\cmidrule[0.25pt]{1-5}
\arrayrulecolor{black}
\multirow{3}{=}{\Centering Exclude Largest 20\%} & Shallow & \makebox[3.2em][c]{\ensuremath{0.04^{**}}} & 1.44 & 1.41 \\
 & Medium & \makebox[3.2em][c]{\ensuremath{0.09^{**}}} & 1.45 & 1.37 \\
 & Deep & \makebox[3.2em][c]{\ensuremath{0.14^{*}}} & 1.43 & 1.28 \\
\addlinespace[1pt]
\midrule
Universe & Depth & $\Delta$ SR & ResNet+ SR & NN+ SR \\
\midrule
\multirow{3}{=}{\Centering All Stocks} & Shallow & \makebox[3.2em][c]{\ensuremath{0.09}} & 1.92 & 1.84 \\
 & Medium & \makebox[3.2em][c]{\ensuremath{0.47^{***}}} & 1.94 & 1.47 \\
 & Deep & \makebox[3.2em][c]{\ensuremath{1.18^{***}}} & 2.07 & 0.89 \\
\arrayrulecolor{gray!45}
\cmidrule[0.25pt]{1-5}
\arrayrulecolor{black}
\multirow{3}{=}{\Centering Exclude Smallest 20\%} & Shallow & \makebox[3.2em][c]{\ensuremath{0.09^{*}}} & 1.74 & 1.65 \\
 & Medium & \makebox[3.2em][c]{\ensuremath{0.50^{***}}} & 1.76 & 1.27 \\
 & Deep & \makebox[3.2em][c]{\ensuremath{1.08^{***}}} & 1.87 & 0.79 \\
\arrayrulecolor{gray!45}
\cmidrule[0.25pt]{1-5}
\arrayrulecolor{black}
\multirow{3}{=}{\Centering Exclude Largest 20\%} & Shallow & \makebox[3.2em][c]{\ensuremath{0.20^{**}}} & 2.46 & 2.26 \\
 & Medium & \makebox[3.2em][c]{\ensuremath{0.70^{***}}} & 2.51 & 1.81 \\
 & Deep & \makebox[3.2em][c]{\ensuremath{1.67^{***}}} & 2.78 & 1.12 \\
\bottomrule
\end{tabularx}
\par\vspace{4pt}
\begin{minipage}{\linewidth}
\setstretch{1.0}\footnotesize Market capitalization screens are applied on the raw forecasts before portfolio formation.
All Stocks is the full universe; Exclude Smallest 20\% keeps the largest 80\%, and Exclude Largest 20\% keeps the smallest 80\%.
$\Delta$ SR is residual (ResNet or ResNet+) minus feedforward (NN or NN+) annualized Sharpe ratio for value-weighted long-short portfolios.
Within each depth group, SRs and $\Delta$ SRs are averaged equally across the set of width and depth specifications completed by both models.
Stars on $\Delta$ SR use the two-sided Ledoit--Wolf studentized circular block bootstrap test for equal Sharpe ratios; *, **, and *** denote the 10\%, 5\%, and 1\% levels.
\end{minipage}
\end{table}

%% file: TransactionCostAdjustedPerformance.tex
\begin{table}[h]
\centering
\scriptsize
\setstretch{1.0}
\renewcommand{\arraystretch}{1.12}
\newcolumntype{Y}[1]{>{\hsize=#1\hsize\linewidth=\hsize\Centering\arraybackslash}X}
\captionsetup{skip=4pt,justification=raggedright,singlelinecheck=false,labelsep=newline,labelfont=bf,textfont=bf}
\caption{\textbf{Transaction Cost Tests}}
\label{tab:TransactionCostAdjustedPerformance}
\newcommand{\transactioncostdivider}{\tikz[baseline=-0.6ex]{\draw[gray!65,dash pattern=on 1.2pt off 1.2pt,line width=0.25pt] (0,-0.9ex) -- (0,1.7ex);}}
\begin{tabularx}{\linewidth}{@{}Y{0.85}Y{0.72}@{\hspace{5pt}\transactioncostdivider\hspace{5pt}}Y{0.73}@{\hspace{5pt}\transactioncostdivider\hspace{5pt}}Y{1.15}Y{1.00}@{\hspace{5pt}\transactioncostdivider\hspace{5pt}}Y{1.10}Y{1.05}@{}}
\toprule
Depth & Cost & $\Delta$ SR & ResNet SR & Turnover & NN SR & Turnover \\
\midrule
\multirow{3}{=}{\Centering Shallow} & 0 bps & \makebox[3.2em][c]{\ensuremath{0.01}} & 0.95 & \multirow{3}{=}{\Centering 211.75} & 0.94 & \multirow{3}{=}{\Centering 212.22} \\
 & 10 bps & \makebox[3.2em][c]{\ensuremath{0.01}} & 0.84 &  & 0.83 &  \\
 & 50 bps & \makebox[3.2em][c]{\ensuremath{0.02}} & 0.42 &  & 0.40 &  \\
\arrayrulecolor{gray!45}
\cmidrule[0.25pt]{1-7}
\arrayrulecolor{black}
\multirow{3}{=}{\Centering Medium} & 0 bps & \makebox[3.2em][c]{\ensuremath{0.05}} & 0.94 & \multirow{3}{=}{\Centering 212.67} & 0.90 & \multirow{3}{=}{\Centering 209.64} \\
 & 10 bps & \makebox[3.2em][c]{\ensuremath{0.05}} & 0.83 &  & 0.80 &  \\
 & 50 bps & \makebox[3.2em][c]{\ensuremath{0.06}} & 0.42 &  & 0.38 &  \\
\arrayrulecolor{gray!45}
\cmidrule[0.25pt]{1-7}
\arrayrulecolor{black}
\multirow{3}{=}{\Centering Deep} & 0 bps & \makebox[3.2em][c]{\ensuremath{0.10}} & 0.93 & \multirow{3}{=}{\Centering 208.80} & 0.82 & \multirow{3}{=}{\Centering 209.58} \\
 & 10 bps & \makebox[3.2em][c]{\ensuremath{0.10}} & 0.83 &  & 0.72 &  \\
 & 50 bps & \makebox[3.2em][c]{\ensuremath{0.12}} & 0.42 &  & 0.30 &  \\
\addlinespace[1pt]
\midrule
Depth & Cost & $\Delta$ SR & ResNet+ SR & Turnover & NN+ SR & Turnover \\
\midrule
\multirow{3}{=}{\Centering Shallow} & 0 bps & \makebox[3.2em][c]{\ensuremath{0.09}} & 1.92 & \multirow{3}{=}{\Centering 255.68} & 1.84 & \multirow{3}{=}{\Centering 257.53} \\
 & 10 bps & \makebox[3.2em][c]{\ensuremath{0.09}} & 1.77 &  & 1.68 &  \\
 & 50 bps & \makebox[3.2em][c]{\ensuremath{0.13^{**}}} & 1.18 &  & 1.06 &  \\
\arrayrulecolor{gray!45}
\cmidrule[0.25pt]{1-7}
\arrayrulecolor{black}
\multirow{3}{=}{\Centering Medium} & 0 bps & \makebox[3.2em][c]{\ensuremath{0.47^{***}}} & 1.94 & \multirow{3}{=}{\Centering 256.67} & 1.47 & \multirow{3}{=}{\Centering 250.78} \\
 & 10 bps & \makebox[3.2em][c]{\ensuremath{0.49^{***}}} & 1.78 &  & 1.30 &  \\
 & 50 bps & \makebox[3.2em][c]{\ensuremath{0.55^{***}}} & 1.16 &  & 0.61 &  \\
\arrayrulecolor{gray!45}
\cmidrule[0.25pt]{1-7}
\arrayrulecolor{black}
\multirow{3}{=}{\Centering Deep} & 0 bps & \makebox[3.2em][c]{\ensuremath{1.18^{***}}} & 2.07 & \multirow{3}{=}{\Centering 259.49} & 0.89 & \multirow{3}{=}{\Centering 261.97} \\
 & 10 bps & \makebox[3.2em][c]{\ensuremath{1.26^{***}}} & 1.92 &  & 0.66 &  \\
 & 50 bps & \makebox[3.2em][c]{\ensuremath{1.56^{***}}} & 1.30 &  & -0.26 &  \\
\bottomrule
\end{tabularx}
\par\vspace{4pt}
\begin{minipage}{\linewidth}
\setstretch{1.0}\footnotesize Transaction costs are deducted from monthly portfolio returns based on the stated cost rate and turnover.
Cost is measured in basis points (bps) per dollar traded; turnover is the average monthly percentage of portfolio value traded.
$\Delta$ SR is residual (ResNet or ResNet+) minus feedforward (NN or NN+) annualized Sharpe ratio for value-weighted long-short portfolios net of transaction costs.
Within each depth group, model SRs are averaged equally across each model's available width and depth specifications, whereas $\Delta$ SR and turnover are averaged equally across the set of specifications completed by both models.
Stars on $\Delta$ SR use the two-sided Ledoit--Wolf studentized circular block bootstrap test for equal Sharpe ratios; *, **, and *** denote the 10\%, 5\%, and 1\% levels.
\end{minipage}
\end{table}

%% file: MarketStateRobustness.tex
\begin{table}[t]
\centering
\scriptsize
\setstretch{1.0}
\renewcommand{\arraystretch}{1.10}
\newcolumntype{K}[1]{>{\hsize=#1\hsize\linewidth=\hsize\Centering\arraybackslash}X}
\captionsetup{skip=4pt,justification=raggedright,singlelinecheck=false,labelsep=newline,labelfont=bf,textfont=bf}
\caption{\textbf{NBER Business Cycle Tests}}
\label{tab:MarketStateRobustness}
\newcommand{\marketstatedivider}{\tikz[baseline=-0.6ex]{\draw[gray!65,dash pattern=on 1.2pt off 1.2pt,line width=0.25pt] (0,-0.9ex) -- (0,1.7ex);}}
\begin{tabularx}{\linewidth}{@{}K{0.80}K{1.15}K{0.65}@{\hspace{5pt}\marketstatedivider\hspace{5pt}}K{0.90}K{1.15}K{1.15}@{}}
\toprule
State & Period & Months & $\Delta$ Mean & ResNet Mean & NN Mean \\
\midrule
\multirow{5}{=}{\Centering Expansion} & 1987:02--1990:07 & 42 & \makebox[3.2em][c]{\ensuremath{0.82^{***}}} & 3.26 & 2.44 \\
 & 1991:04--2001:03 & 120 & \makebox[3.2em][c]{\ensuremath{0.03}} & 2.23 & 2.21 \\
 & 2001:12--2007:12 & 73 & \makebox[3.2em][c]{\ensuremath{0.25^{*}}} & 1.60 & 1.35 \\
 & 2009:07--2020:02 & 128 & \makebox[3.2em][c]{\ensuremath{-0.09}} & 0.87 & 0.97 \\
 & 2020:05--2023:12 & 44 & \makebox[3.2em][c]{\ensuremath{0.47^{***}}} & 2.23 & 1.77 \\
\arrayrulecolor{gray!45}
\cmidrule[0.25pt]{1-6}
\arrayrulecolor{black}
\multirow{4}{=}{\Centering Recession} & 1990:08--1991:03 & 8\textsuperscript{\textdagger} & \makebox[3.2em][c]{\ensuremath{0.66^{\dagger}}} & 4.52 & 3.86 \\
 & 2001:04--2001:11 & 8\textsuperscript{\textdagger} & \makebox[3.2em][c]{\ensuremath{1.14^{\dagger}}} & 5.52 & 4.38 \\
 & 2008:01--2009:06 & 18 & \makebox[3.2em][c]{\ensuremath{-0.02}} & 2.44 & 2.46 \\
 & 2020:03--2020:04 & 2\textsuperscript{\textdagger} & \makebox[3.2em][c]{\ensuremath{-2.40^{\dagger}}} & -0.74 & 1.66 \\
\addlinespace[1pt]
\midrule
State & Period & Months & $\Delta$ Mean & ResNet+ Mean & NN+ Mean \\
\midrule
\multirow{5}{=}{\Centering Expansion} & 1987:02--1990:07 & 42 & \makebox[3.2em][c]{\ensuremath{2.39^{***}}} & 2.90 & 0.51 \\
 & 1991:04--2001:03 & 120 & \makebox[3.2em][c]{\ensuremath{2.89^{***}}} & 5.62 & 2.73 \\
 & 2001:12--2007:12 & 73 & \makebox[3.2em][c]{\ensuremath{0.61^{**}}} & 2.23 & 1.61 \\
 & 2009:07--2020:02 & 128 & \makebox[3.2em][c]{\ensuremath{0.42^{**}}} & 1.73 & 1.31 \\
 & 2020:05--2023:12 & 44 & \makebox[3.2em][c]{\ensuremath{0.56}} & 3.55 & 2.98 \\
\arrayrulecolor{gray!45}
\cmidrule[0.25pt]{1-6}
\arrayrulecolor{black}
\multirow{4}{=}{\Centering Recession} & 1990:08--1991:03 & 8\textsuperscript{\textdagger} & \makebox[3.2em][c]{\ensuremath{4.37^{\dagger}}} & 5.59 & 1.22 \\
 & 2001:04--2001:11 & 8\textsuperscript{\textdagger} & \makebox[3.2em][c]{\ensuremath{3.71^{\dagger}}} & 7.80 & 4.10 \\
 & 2008:01--2009:06 & 18 & \makebox[3.2em][c]{\ensuremath{1.17}} & 3.98 & 2.81 \\
 & 2020:03--2020:04 & 2\textsuperscript{\textdagger} & \makebox[3.2em][c]{\ensuremath{-3.58^{\dagger}}} & -5.15 & -1.57 \\
\bottomrule
\end{tabularx}
\par\vspace{4pt}
\begin{minipage}{\linewidth}
\setstretch{1.0}\footnotesize The periods are defined according to the NBER business cycle.
$\Delta$ Mean is residual (ResNet or ResNet+) minus feedforward (NN or NN+) monthly return for value-weighted long-short portfolios, reported in percent and averaged equally across months.
Within each comparison set, model returns are first averaged equally across the same completed width and depth specifications in each month.
Stars on $\Delta$ Mean test zero using Newey\mbox{-}West standard errors with 12 lags; *, **, and *** denote the 10\%, 5\%, and 1\% levels.
\textsuperscript{\textdagger} When intervals contain fewer than 13 monthly observations, statistical significance is not reported for these $\Delta$ Mean estimates.
\end{minipage}
\end{table}

%% file: PredictorList.tex
\begingroup
\footnotesize
\setlength{\LTleft}{0pt}
\setlength{\LTright}{\fill}
\setlength{\LTpost}{0pt}
\setstretch{1.0}
\setlength{\tabcolsep}{0pt}
\renewcommand{\arraystretch}{0.88}
\captionsetup{skip=4pt,justification=raggedright,singlelinecheck=false,labelsep=newline,labelfont=bf,textfont=bf}
\begin{longtable}{@{}>{\raggedleft\arraybackslash}p{25pt}@{\hspace{8pt}}>{\raggedright\arraybackslash}p{100pt}@{\hspace{8pt}}>{\raggedright\arraybackslash}p{240pt}@{\hspace{8pt}}>{\raggedright\arraybackslash}p{70pt}@{\hspace{8pt}}>{\raggedright\arraybackslash}p{78pt}@{\hspace{8pt}}>{\raggedright\arraybackslash}p{88pt}@{}}
\caption{\textbf{Predictors for One-Month-Ahead Return Forecasts}}\label{tab:PredictorList}\\
\toprule
\textbf{No.} & \textbf{Acronym} & \textbf{Description} & \textbf{Block} & \textbf{Source} & \textbf{Frequency} \\
\midrule
\endfirsthead
\caption[]{\textbf{Predictors for One-Month-Ahead Return Forecasts (continued)}}\\
\toprule
\textbf{No.} & \textbf{Acronym} & \textbf{Description} & \textbf{Block} & \textbf{Source} & \textbf{Frequency} \\
\midrule
\endhead
\midrule
\multicolumn{6}{r}{Continued on next page}\\
\endfoot
\bottomrule
\endlastfoot
1 & niq\_su & Standardized earnings surprise & Firm chars. & JKP panel & Monthly aligned \\
2 & ret\_6\_1 & Price momentum t-6 to t-1 & Firm chars. & JKP panel & Monthly aligned \\
3 & ret\_12\_1 & 12-month price momentum & Firm chars. & JKP panel & Monthly aligned \\
4 & saleq\_su & Standardized revenue surprise & Firm chars. & JKP panel & Monthly aligned \\
5 & tax\_gr1a & Tax expense surprise & Firm chars. & JKP panel & Monthly aligned \\
6 & ni\_inc8q & Consecutive earnings increases & Firm chars. & JKP panel & Monthly aligned \\
7 & prc\_highprc\_252d & Price to 12-month high & Firm chars. & JKP panel & Monthly aligned \\
8 & resff3\_6\_1 & 6-month residual momentum & Firm chars. & JKP panel & Monthly aligned \\
9 & resff3\_12\_1 & 12-month residual momentum & Firm chars. & JKP panel & Monthly aligned \\
10 & be\_me & Book-to-market equity & Firm chars. & JKP panel & Monthly aligned \\
11 & debt\_me & Debt-to-market & Firm chars. & JKP panel & Monthly aligned \\
12 & at\_me & Assets-to-market & Firm chars. & JKP panel & Monthly aligned \\
13 & ret\_60\_12 & Long-term reversal & Firm chars. & JKP panel & Monthly aligned \\
14 & ni\_me & Earnings-to-price & Firm chars. & JKP panel & Monthly aligned \\
15 & fcf\_me & Free cash flow-to-price & Firm chars. & JKP panel & Monthly aligned \\
16 & div12m\_me & Dividend yield & Firm chars. & JKP panel & Monthly aligned \\
17 & eqpo\_me & Payout yield & Firm chars. & JKP panel & Monthly aligned \\
18 & eqnpo\_me & Net payout yield & Firm chars. & JKP panel & Monthly aligned \\
19 & sale\_gr3 & Sales growth (3 years) & Firm chars. & JKP panel & Monthly aligned \\
20 & sale\_gr1 & Sales growth (1 year) & Firm chars. & JKP panel & Monthly aligned \\
21 & ebitda\_mev & EBITDA to enterprise value & Firm chars. & JKP panel & Monthly aligned \\
22 & sale\_me & Sales-to-market & Firm chars. & JKP panel & Monthly aligned \\
23 & ocf\_me & Cash flow to market & Firm chars. & JKP panel & Monthly aligned \\
24 & ival\_me & Intrinsic value-to-market & Firm chars. & JKP panel & Monthly aligned \\
25 & bev\_mev & Book to enterprise value & Firm chars. & JKP panel & Monthly aligned \\
26 & netdebt\_me & Net debt-to-price & Firm chars. & JKP panel & Monthly aligned \\
27 & eq\_dur & Equity duration & Firm chars. & JKP panel & Monthly aligned \\
28 & capex\_abn & Abnormal investment & Firm chars. & JKP panel & Monthly aligned \\
29 & at\_gr1 & Asset growth & Firm chars. & JKP panel & Monthly aligned \\
30 & ppeinv\_gr1a & Change in PPE and inventory & Firm chars. & JKP panel & Monthly aligned \\
31 & noa\_at & Net operating assets & Firm chars. & JKP panel & Monthly aligned \\
32 & noa\_gr1a & Change in net operating assets & Firm chars. & JKP panel & Monthly aligned \\
33 & lnoa\_gr1a & Change in long-term operating assets & Firm chars. & JKP panel & Monthly aligned \\
\pagebreak
34 & capx\_gr1 & CAPEX growth (1 year) & Firm chars. & JKP panel & Monthly aligned \\
35 & capx\_gr2 & CAPEX growth (2 years) & Firm chars. & JKP panel & Monthly aligned \\
36 & capx\_gr3 & CAPEX growth (3 years) & Firm chars. & JKP panel & Monthly aligned \\
37 & chcsho\_12m & Net stock issues & Firm chars. & JKP panel & Monthly aligned \\
38 & eqnpo\_12m & Equity net payout & Firm chars. & JKP panel & Monthly aligned \\
39 & debt\_gr3 & Growth in book debt (3 years) & Firm chars. & JKP panel & Monthly aligned \\
40 & inv\_gr1 & Inventory growth & Firm chars. & JKP panel & Monthly aligned \\
41 & inv\_gr1a & Inventory change & Firm chars. & JKP panel & Monthly aligned \\
42 & oaccruals\_at & Operating accruals & Firm chars. & JKP panel & Monthly aligned \\
43 & taccruals\_at & Total accruals & Firm chars. & JKP panel & Monthly aligned \\
44 & cowc\_gr1a & Change in current working capital & Firm chars. & JKP panel & Monthly aligned \\
45 & coa\_gr1a & Change in current operating assets & Firm chars. & JKP panel & Monthly aligned \\
46 & col\_gr1a & Change in current operating liabilities & Firm chars. & JKP panel & Monthly aligned \\
47 & nncoa\_gr1a & Change in net noncurrent assets & Firm chars. & JKP panel & Monthly aligned \\
48 & ncoa\_gr1a & Change in noncurrent assets & Firm chars. & JKP panel & Monthly aligned \\
49 & ncol\_gr1a & Change in noncurrent liabilities & Firm chars. & JKP panel & Monthly aligned \\
50 & nfna\_gr1a & Change in net financial assets & Firm chars. & JKP panel & Monthly aligned \\
51 & sti\_gr1a & Change in short-term investments & Firm chars. & JKP panel & Monthly aligned \\
52 & lti\_gr1a & Change in long-term investments & Firm chars. & JKP panel & Monthly aligned \\
53 & fnl\_gr1a & Change in financial liabilities & Firm chars. & JKP panel & Monthly aligned \\
54 & be\_gr1a & Change in common equity & Firm chars. & JKP panel & Monthly aligned \\
55 & oaccruals\_ni & Percent operating accruals & Firm chars. & JKP panel & Monthly aligned \\
56 & taccruals\_ni & Percent total accruals & Firm chars. & JKP panel & Monthly aligned \\
57 & netis\_at & Net total issuance & Firm chars. & JKP panel & Monthly aligned \\
58 & eqnetis\_at & Net equity issuance & Firm chars. & JKP panel & Monthly aligned \\
59 & dbnetis\_at & Net debt issuance & Firm chars. & JKP panel & Monthly aligned \\
60 & niq\_be & Quarterly ROE & Firm chars. & JKP panel & Monthly aligned \\
61 & niq\_be\_chg1 & Change in quarterly ROE & Firm chars. & JKP panel & Monthly aligned \\
62 & niq\_at & Quarterly ROA & Firm chars. & JKP panel & Monthly aligned \\
63 & niq\_at\_chg1 & Change in quarterly ROA & Firm chars. & JKP panel & Monthly aligned \\
64 & ebit\_bev & Return on operating assets & Firm chars. & JKP panel & Monthly aligned \\
65 & ebit\_sale & Profit margin & Firm chars. & JKP panel & Monthly aligned \\
66 & sale\_bev & Asset turnover & Firm chars. & JKP panel & Monthly aligned \\
\pagebreak
67 & at\_turnover & Capital turnover & Firm chars. & JKP panel & Monthly aligned \\
68 & gp\_at & Gross profits-to-assets & Firm chars. & JKP panel & Monthly aligned \\
69 & gp\_atl1 & Gross profits-to-lagged assets & Firm chars. & JKP panel & Monthly aligned \\
70 & ope\_be & Operating profit to equity & Firm chars. & JKP panel & Monthly aligned \\
71 & ope\_bel1 & Operating profit to lagged equity & Firm chars. & JKP panel & Monthly aligned \\
72 & op\_at & Operating profit to assets & Firm chars. & JKP panel & Monthly aligned \\
73 & op\_atl1 & Operating profit to lagged assets & Firm chars. & JKP panel & Monthly aligned \\
74 & cop\_at & Cash operating profit to assets & Firm chars. & JKP panel & Monthly aligned \\
75 & cop\_atl1 & Cash operating profit to lagged assets & Firm chars. & JKP panel & Monthly aligned \\
76 & f\_score & Piotroski F-score & Firm chars. & JKP panel & Monthly aligned \\
77 & o\_score & Ohlson O-score & Firm chars. & JKP panel & Monthly aligned \\
78 & z\_score & Altman Z-score & Firm chars. & JKP panel & Monthly aligned \\
79 & pi\_nix & Taxable to book income & Firm chars. & JKP panel & Monthly aligned \\
80 & at\_be & Book leverage & Firm chars. & JKP panel & Monthly aligned \\
81 & saleq\_gr1 & Sales growth (1 quarter) & Firm chars. & JKP panel & Monthly aligned \\
82 & rd\_me & R\&D-to-market & Firm chars. & JKP panel & Monthly aligned \\
83 & rd\_sale & R\&D-to-sales & Firm chars. & JKP panel & Monthly aligned \\
84 & opex\_at & Operating leverage & Firm chars. & JKP panel & Monthly aligned \\
85 & emp\_gr1 & Hiring rate & Firm chars. & JKP panel & Monthly aligned \\
86 & rd5\_at & R\&D capital to assets & Firm chars. & JKP panel & Monthly aligned \\
87 & age & Firm age & Firm chars. & JKP panel & Monthly aligned \\
88 & dsale\_dinv & Sales change minus inventory change & Firm chars. & JKP panel & Monthly aligned \\
89 & dsale\_drec & Sales change minus receivables change & Firm chars. & JKP panel & Monthly aligned \\
90 & dgp\_dsale & Gross margin change minus sales change & Firm chars. & JKP panel & Monthly aligned \\
91 & dsale\_dsga & Sales change minus SG\&A change & Firm chars. & JKP panel & Monthly aligned \\
92 & sale\_emp\_gr1 & Labor force efficiency & Firm chars. & JKP panel & Monthly aligned \\
93 & tangibility & Asset tangibility & Firm chars. & JKP panel & Monthly aligned \\
94 & kz\_index & Kaplan-Zingales index & Firm chars. & JKP panel & Monthly aligned \\
95 & ocfq\_saleq\_std & Cash flow volatility & Firm chars. & JKP panel & Monthly aligned \\
96 & cash\_at & Cash-to-assets & Firm chars. & JKP panel & Monthly aligned \\
97 & ni\_ar1 & Earnings persistence & Firm chars. & JKP panel & Monthly aligned \\
98 & ni\_ivol & Earnings volatility & Firm chars. & JKP panel & Monthly aligned \\
99 & earnings\_variability & Earnings variability & Firm chars. & JKP panel & Monthly aligned \\
\pagebreak
100 & aliq\_at & Liquidity of book assets & Firm chars. & JKP panel & Monthly aligned \\
101 & aliq\_mat & Liquidity of market assets & Firm chars. & JKP panel & Monthly aligned \\
102 & seas\_1\_1an & Year 1-lagged return, annual & Firm chars. & JKP panel & Monthly aligned \\
103 & seas\_1\_1na & Year 1-lagged return, nonannual & Firm chars. & JKP panel & Monthly aligned \\
104 & seas\_2\_5an & Years 2-5 lagged returns, annual & Firm chars. & JKP panel & Monthly aligned \\
105 & seas\_2\_5na & Years 2-5 lagged returns, nonannual & Firm chars. & JKP panel & Monthly aligned \\
106 & seas\_6\_10an & Years 6-10 lagged returns, annual & Firm chars. & JKP panel & Monthly aligned \\
107 & seas\_6\_10na & Years 6-10 lagged returns, nonannual & Firm chars. & JKP panel & Monthly aligned \\
108 & seas\_11\_15an & Years 11-15 lagged returns, annual & Firm chars. & JKP panel & Monthly aligned \\
109 & seas\_11\_15na & Years 11-15 lagged returns, nonannual & Firm chars. & JKP panel & Monthly aligned \\
110 & seas\_16\_20an & Years 16-20 lagged returns, annual & Firm chars. & JKP panel & Monthly aligned \\
111 & seas\_16\_20na & Years 16-20 lagged returns, nonannual & Firm chars. & JKP panel & Monthly aligned \\
112 & market\_equity & Market equity & Firm chars. & JKP panel & Monthly aligned \\
113 & ivol\_ff3\_21d & FF3 idiosyncratic volatility & Firm chars. & JKP panel & Monthly aligned \\
114 & ivol\_capm\_252d & CAPM idiosyncratic volatility, 252 days & Firm chars. & JKP panel & Monthly aligned \\
115 & ivol\_capm\_21d & CAPM idiosyncratic volatility, 21 days & Firm chars. & JKP panel & Monthly aligned \\
116 & ivol\_hxz4\_21d & q-factor idiosyncratic volatility & Firm chars. & JKP panel & Monthly aligned \\
117 & rvol\_21d & Return volatility & Firm chars. & JKP panel & Monthly aligned \\
118 & beta\_60m & Market beta & Firm chars. & JKP panel & Monthly aligned \\
119 & betabab\_1260d & Frazzini-Pedersen beta & Firm chars. & JKP panel & Monthly aligned \\
120 & beta\_dimson\_21d & Dimson beta & Firm chars. & JKP panel & Monthly aligned \\
121 & turnover\_126d & Share turnover & Firm chars. & JKP panel & Monthly aligned \\
122 & turnover\_var\_126d & Turnover variation & Firm chars. & JKP panel & Monthly aligned \\
123 & dolvol\_126d & Dollar trading volume & Firm chars. & JKP panel & Monthly aligned \\
124 & dolvol\_var\_126d & Dollar volume variation & Firm chars. & JKP panel & Monthly aligned \\
125 & prc & Price per share & Firm chars. & JKP panel & Monthly aligned \\
126 & ami\_126d & Amihud Measure & Firm chars. & JKP panel & Monthly aligned \\
127 & zero\_trades\_21d & Zero trade days, 1 month & Firm chars. & JKP panel & Monthly aligned \\
128 & zero\_trades\_126d & Zero trade days, 6 months & Firm chars. & JKP panel & Monthly aligned \\
129 & zero\_trades\_252d & Zero trade days, 12 months & Firm chars. & JKP panel & Monthly aligned \\
130 & rmax1\_21d & Maximum daily return & Firm chars. & JKP panel & Monthly aligned \\
131 & rskew\_21d & Total skewness & Firm chars. & JKP panel & Monthly aligned \\
132 & iskew\_capm\_21d & CAPM idiosyncratic skewness & Firm chars. & JKP panel & Monthly aligned \\
\pagebreak
133 & iskew\_ff3\_21d & FF3 idiosyncratic skewness & Firm chars. & JKP panel & Monthly aligned \\
134 & iskew\_hxz4\_21d & q-factor idiosyncratic skewness & Firm chars. & JKP panel & Monthly aligned \\
135 & coskew\_21d & Coskewness & Firm chars. & JKP panel & Monthly aligned \\
136 & ret\_1\_0 & Short-term reversal & Firm chars. & JKP panel & Monthly aligned \\
137 & betadown\_252d & Downside beta & Firm chars. & JKP panel & Monthly aligned \\
138 & bidaskhl\_21d & High-low bid-ask spread & Firm chars. & JKP panel & Monthly aligned \\
139 & ret\_3\_1 & Price momentum t-3 to t-1 & Firm chars. & JKP panel & Monthly aligned \\
140 & ret\_9\_1 & Price momentum t-9 to t-1 & Firm chars. & JKP panel & Monthly aligned \\
141 & ret\_12\_7 & Price momentum t-12 to t-7 & Firm chars. & JKP panel & Monthly aligned \\
142 & corr\_1260d & Market correlation & Firm chars. & JKP panel & Monthly aligned \\
143 & rmax5\_21d & Highest 5 daily returns & Firm chars. & JKP panel & Monthly aligned \\
144 & rmax5\_rvol\_21d & Top five returns over volatility & Firm chars. & JKP panel & Monthly aligned \\
145 & ni\_be & Return on equity & Firm chars. & JKP panel & Monthly aligned \\
146 & ocf\_at & Operating cash flow to assets & Firm chars. & JKP panel & Monthly aligned \\
147 & ocf\_at\_chg1 & Change in cash flow to assets & Firm chars. & JKP panel & Monthly aligned \\
148 & mispricing\_perf & Mispricing factor: Performance & Firm chars. & JKP panel & Monthly aligned \\
149 & mispricing\_mgmt & Mispricing factor: Management & Firm chars. & JKP panel & Monthly aligned \\
150 & qmj & Quality minus junk composite & Firm chars. & JKP panel & Monthly aligned \\
151 & qmj\_prof & Quality minus junk profitability & Firm chars. & JKP panel & Monthly aligned \\
152 & qmj\_growth & Quality minus junk growth & Firm chars. & JKP panel & Monthly aligned \\
153 & qmj\_safety & Quality minus junk safety & Firm chars. & JKP panel & Monthly aligned \\
154 & d/p & Dividend-price ratio & Macro & Goyal-Welch & Monthly \\
155 & e/p & Earnings-price ratio & Macro & Goyal-Welch & Monthly \\
156 & b/m & Book-to-market ratio & Macro & Goyal-Welch & Monthly \\
157 & ntis & Net equity expansion & Macro & Goyal-Welch & Monthly \\
158 & tbl & Treasury-bill rate & Macro & Goyal-Welch & Monthly \\
159 & tms & Term spread & Macro & Goyal-Welch & Monthly \\
160 & dfy & Default spread & Macro & Goyal-Welch & Monthly \\
161 & svar & Stock variance & Macro & Goyal-Welch & Monthly \\
\end{longtable}
\par\vspace{4pt}
\noindent
\begin{minipage}{\linewidth}
\RaggedRight\sloppy\setlength{\emergencystretch}{4em}\setstretch{1.0}\footnotesize Forecasts use 153 firm characteristics from \citet{jensen2023ThereReplicationCrisisa} and eight macroeconomic predictors used by \citet{gu2020EmpiricalAssetPricing} and drawn from \citet{welch2008ComprehensiveLookEmpirical}.
Firm chars. denotes the JKP variables, which are median imputed by month and rank normalized.
Macroeconomic predictors enter in levels.
\end{minipage}
\endgroup

%% file: AppendixForecastRankingValues.tex
\begin{table}[!htbp]
\centering\scriptsize\setstretch{1.0}
\renewcommand{\arraystretch}{1.12}
\setlength{\tabcolsep}{2pt}
\captionsetup{font=footnotesize,skip=4pt,justification=raggedright,singlelinecheck=false,labelsep=newline,labelfont=bf,textfont=bf}
\caption{\textbf{Forecast Preservation and Ranking Economic Values by Depth Group}}
\label{tab:appendix_forecast_ranking_values}
\begin{tabular*}{\linewidth}{@{\extracolsep{\fill}}llcccc@{\hspace{10pt}}lcccc@{}}
\toprule
 & Model & $P_t$ & $C_{t+1}/\%$ & $V_{t+1}^{S}/\mathrm{bps}$ & $V_{t+1}^{D}/\mathrm{bps}$ & Model & $P_t$ & $C_{t+1}/\%$ & $V_{t+1}^{S}/\mathrm{bps}$ & $V_{t+1}^{D}/\mathrm{bps}$ \\
\midrule
\multirow{3}{*}{Medium} & ResNet & \makebox[3.8em][c]{0.93} & \makebox[3.8em][c]{0.32} & \makebox[3.8em][c]{\ensuremath{56.72^{***}}} & \makebox[3.8em][c]{\ensuremath{55.78^{***}}} & ResNet+ & \makebox[3.8em][c]{0.81} & \makebox[3.8em][c]{\ensuremath{0.72^{***}}} & \makebox[3.8em][c]{\ensuremath{62.62^{***}}} & \makebox[3.8em][c]{\ensuremath{63.31^{***}}} \\
 & NN & \makebox[3.8em][c]{0.87} & \makebox[3.8em][c]{\ensuremath{-0.74^{*}}} & \makebox[3.8em][c]{\ensuremath{56.90^{***}}} & \makebox[3.8em][c]{\ensuremath{54.83^{***}}} & NN+ & \makebox[3.8em][c]{0.75} & \makebox[3.8em][c]{-0.04} & \makebox[3.8em][c]{\ensuremath{62.62^{***}}} & \makebox[3.8em][c]{\ensuremath{46.73^{***}}} \\
 & Difference & \makebox[3.8em][c]{\ensuremath{0.06^{***}}} & \makebox[3.8em][c]{\ensuremath{1.06^{**}}} & \makebox[3.8em][c]{-0.18} & \makebox[3.8em][c]{0.95} & Difference & \makebox[3.8em][c]{\ensuremath{0.06^{***}}} & \makebox[3.8em][c]{0.75} & \makebox[3.8em][c]{0.00} & \makebox[3.8em][c]{\ensuremath{16.58^{***}}} \\
\addlinespace[1pt]
\arrayrulecolor{gray!45}
\midrule[0.25pt]
\arrayrulecolor{black}
\multirow{3}{*}{Deep} & ResNet & \makebox[3.8em][c]{0.92} & \makebox[3.8em][c]{\ensuremath{0.58^{**}}} & \makebox[3.8em][c]{\ensuremath{56.72^{***}}} & \makebox[3.8em][c]{\ensuremath{55.58^{***}}} & ResNet+ & \makebox[3.8em][c]{0.79} & \makebox[3.8em][c]{\ensuremath{1.12^{***}}} & \makebox[3.8em][c]{\ensuremath{62.62^{***}}} & \makebox[3.8em][c]{\ensuremath{67.54^{***}}} \\
 & NN & \makebox[3.8em][c]{0.84} & \makebox[3.8em][c]{-0.99} & \makebox[3.8em][c]{\ensuremath{57.17^{***}}} & \makebox[3.8em][c]{\ensuremath{54.53^{***}}} & NN+ & \makebox[3.8em][c]{0.64} & \makebox[3.8em][c]{\ensuremath{-1.88^{***}}} & \makebox[3.8em][c]{\ensuremath{62.62^{***}}} & \makebox[3.8em][c]{\ensuremath{28.17^{***}}} \\
 & Difference & \makebox[3.8em][c]{\ensuremath{0.08^{***}}} & \makebox[3.8em][c]{\ensuremath{1.58^{**}}} & \makebox[3.8em][c]{-0.46} & \makebox[3.8em][c]{1.04} & Difference & \makebox[3.8em][c]{\ensuremath{0.15^{***}}} & \makebox[3.8em][c]{\ensuremath{2.99^{***}}} & \makebox[3.8em][c]{0.00} & \makebox[3.8em][c]{\ensuremath{39.36^{***}}} \\
\bottomrule
\end{tabular*}
\par\vspace{4pt}
\begin{minipage}{\linewidth}\footnotesize\setstretch{1.0}
The sample, preservation score $P$, and revision accuracy $C$ are the same as in Table~\ref{tab:ForecastRefinement}.
$V^S$ and $V^D$ are the separate shallow and deep ranking economic values; $V^D-V^S$ equals the revision economic value reported there.
$D$ denotes the target forecast in each depth group: 6--10 for Medium and 11--20 for Deep.
$S$ is its minimum-depth anchor, not an average over all shallow depths.
Eligible specification means over 443 forecast origins receive equal weight, and each anchor inherits its target's weight.
Different eligible width compositions can therefore produce different shallow means.
$P$ is unscaled, $C$ is in percent, and ranking economic values are in monthly basis points.
Difference is the residual model mean minus the feedforward model mean.
Stars test zero using Newey\mbox{-}West standard errors with 12 lags; *, **, and *** denote the 10\%, 5\%, and 1\% levels.
Stars on the two value levels do not test their difference.
Significance levels are not reported for the model-family $P_t$ estimates. Ranking economic values are not portfolio returns.
\end{minipage}
\end{table}

%% file: AppendixThemeValueLosses.tex
\begin{table}[!htbp]
\centering\scriptsize\setstretch{1.0}
\renewcommand{\arraystretch}{1.12}
\setlength{\tabcolsep}{2pt}
\captionsetup{font=footnotesize,skip=4pt,justification=raggedright,singlelinecheck=false,labelsep=newline,labelfont=bf,textfont=bf}
\caption{\textbf{Shallow and Deep Ranking Economic Value Losses by Theme}}
\label{tab:appendix_theme_value_losses}
\begin{tabular*}{\linewidth}{@{\extracolsep{\fill}}lcccc@{\hspace{10pt}}cccc@{}}
\toprule
 & \multicolumn{4}{c}{ResNet} & \multicolumn{4}{c}{NN} \\
\cmidrule(lr){2-5}\cmidrule(lr){6-9}
Theme & $\Delta P_t$ & $\Delta C_{t+1}/\%$ & $\Delta V_{t+1}^{S}/\mathrm{bps}$ & $\Delta V_{t+1}^{D}/\mathrm{bps}$ & $\Delta P_t$ & $\Delta C_{t+1}/\%$ & $\Delta V_{t+1}^{S}/\mathrm{bps}$ & $\Delta V_{t+1}^{D}/\mathrm{bps}$ \\
\midrule
Accruals & $-0.001^{***}$ & $0.02$ & $-0.13$ & $-0.08$ & $-0.001$ & $0.11$ & $-0.14$ & $-0.13$ \\
Debt Issuance & $0.000$ & $0.04$ & $0.77$ & $0.79$ & $0.001$ & $-0.01$ & $0.71$ & $0.83$ \\
Investment & $-0.001^{*}$ & $0.01$ & $1.11$ & $0.88$ & $-0.001$ & $-0.16$ & $1.05$ & $1.12$ \\
Low Leverage & $-0.003^{***}$ & $0.49^{***}$ & $1.54$ & $1.94$ & $-0.011^{***}$ & $0.05$ & $1.27$ & $1.84$ \\
Low Risk & $0.002^{**}$ & $-0.29$ & $4.02$ & $3.37$ & $-0.001$ & $-0.53$ & $4.05$ & $4.58$ \\
Momentum & $0.003^{***}$ & $-0.36^{***}$ & $7.25^{***}$ & $6.94^{***}$ & $0.004$ & $0.22$ & $7.84^{***}$ & $5.02$ \\
Profit Growth & $0.001^{*}$ & $-0.01$ & $0.62$ & $0.61$ & $0.002^{*}$ & $0.14$ & $0.57$ & $1.08$ \\
Profitability & $-0.003^{***}$ & $-0.05$ & $-0.83$ & $-0.81$ & $-0.003^{***}$ & $0.41^{***}$ & $-0.90$ & $-0.84$ \\
Quality & $0.001$ & $-0.31^{**}$ & $4.83^{**}$ & $4.68^{**}$ & $-0.002$ & $0.19$ & $4.76^{**}$ & $5.46^{***}$ \\
Seasonality & $0.000$ & $-0.05$ & $0.76$ & $0.74$ & $-0.001$ & $-0.17^{**}$ & $0.75$ & $0.77$ \\
Short-Term Reversal & $0.001$ & $0.08$ & $7.64^{***}$ & $8.02^{***}$ & $0.001$ & $-0.02$ & $8.02^{***}$ & $6.43^{***}$ \\
Size & $-0.001$ & $0.22$ & $10.23^{***}$ & $10.08^{***}$ & $0.002$ & $-0.10$ & $12.75^{***}$ & $4.88$ \\
Value & $0.001$ & $-0.11$ & $8.47^{***}$ & $8.31^{***}$ & $-0.001$ & $-0.98^{***}$ & $8.63^{***}$ & $8.34^{***}$ \\
\midrule
\addlinespace[5pt]
 & \multicolumn{4}{c}{ResNet+} & \multicolumn{4}{c}{NN+} \\
\cmidrule(lr){2-5}\cmidrule(lr){6-9}
Theme & $\Delta P_t$ & $\Delta C_{t+1}/\%$ & $\Delta V_{t+1}^{S}/\mathrm{bps}$ & $\Delta V_{t+1}^{D}/\mathrm{bps}$ & $\Delta P_t$ & $\Delta C_{t+1}/\%$ & $\Delta V_{t+1}^{S}/\mathrm{bps}$ & $\Delta V_{t+1}^{D}/\mathrm{bps}$ \\
\midrule
Accruals & $0.001^{***}$ & $0.00$ & $0.17$ & $0.05$ & $0.001^{***}$ & $-0.08^{***}$ & $0.17$ & $0.04$ \\
Debt Issuance & $0.005^{***}$ & $-0.79^{***}$ & $0.68$ & $0.95^{*}$ & $0.003^{***}$ & $-0.10^{***}$ & $0.68$ & $0.40$ \\
Investment & $0.001$ & $-1.33^{***}$ & $2.61^{**}$ & $0.93$ & $0.004^{***}$ & $-0.40^{***}$ & $2.61^{**}$ & $1.63^{***}$ \\
Low Leverage & $-0.007^{***}$ & $0.44^{***}$ & $3.09$ & $4.92$ & $-0.003^{***}$ & $0.50^{***}$ & $3.09$ & $1.19$ \\
Low Risk & $0.007^{***}$ & $0.32$ & $11.69^{***}$ & $11.95^{***}$ & $0.009^{***}$ & $-1.45^{***}$ & $11.69^{***}$ & $5.11^{**}$ \\
Momentum & $0.022^{***}$ & $0.10$ & $6.20^{***}$ & $7.88^{***}$ & $-0.001$ & $-0.32^{***}$ & $6.20^{***}$ & $2.85^{***}$ \\
Profit Growth & $0.003^{***}$ & $1.48^{***}$ & $0.70$ & $1.41$ & $0.001^{***}$ & $0.13^{***}$ & $0.70$ & $0.28$ \\
Profitability & $0.003^{***}$ & $-0.42^{***}$ & $0.19$ & $1.36$ & $0.000$ & $0.10^{***}$ & $0.19$ & $-0.02$ \\
Quality & $0.008^{***}$ & $1.67^{***}$ & $4.12^{***}$ & $6.17^{***}$ & $0.002^{***}$ & $-0.58^{***}$ & $4.12^{***}$ & $1.80^{***}$ \\
Seasonality & $0.003^{***}$ & $0.16$ & $1.59^{**}$ & $0.93$ & $0.001^{***}$ & $-0.04$ & $1.59^{**}$ & $0.79^{*}$ \\
Short-Term Reversal & $0.012^{***}$ & $-0.46^{***}$ & $5.75^{***}$ & $6.57^{***}$ & $0.004^{***}$ & $0.30^{*}$ & $5.75^{***}$ & $2.51^{***}$ \\
Size & $-0.005^{***}$ & $-0.05$ & $6.48$ & $12.59^{**}$ & $-0.008^{***}$ & $1.63^{***}$ & $6.48$ & $2.16$ \\
Value & $0.017^{***}$ & $1.96^{***}$ & $10.85^{***}$ & $11.15^{***}$ & $0.003^{***}$ & $-0.30^{*}$ & $10.85^{***}$ & $5.09^{**}$ \\
\bottomrule
\end{tabular*}
\par\vspace{4pt}
\begin{minipage}{\linewidth}\footnotesize\setstretch{1.0}
Entries are full-input minus theme-zeroed measures, with the same theme set to zero in both shallow and deep models and fitted models and other inputs held fixed.
$\Delta P$ and $\Delta C$ use the same definitions and estimates as Table~\ref{tab:theme_pairwise_losses}; $\Delta V^S$ and $\Delta V^D$ are separate shallow and deep ranking economic value losses.
Their difference is that table's $\Delta V$.
Target depths are 11--20 over 443 forecast origins; specification means receive equal weight, and each anchor inherits its target's weight.
$\Delta P$ is unscaled, $\Delta C$ is in percent, and ranking economic value losses are in monthly basis points.
Stars test each paired loss using Newey\mbox{-}West standard errors with 12 lags and Holm adjustment across 13 themes within each model and measure; *, **, and *** denote the 10\%, 5\%, and 1\% levels.
Stars do not test cross-column differences. Theme effects are not additive; ranking economic values are not portfolio returns.
\end{minipage}
\end{table}

%% file: S1D7DiagnosticSummary.tex
\begin{table}[htbp]
\centering
\scriptsize
\setstretch{1.0}
\renewcommand{\arraystretch}{1.12}
\captionsetup{skip=4pt,justification=raggedright,singlelinecheck=false,labelsep=newline,labelfont=bf,textfont=bf}
\caption{\textbf{Forecast Collapse for Width Schedule 1 at Depth 7}}
\label{tab:S1D7DiagnosticSummary}
\resizebox{\linewidth}{!}{%
\begin{tabular}{lcccccc}
\toprule
\textbf{Model} & \textbf{Depth} & \textbf{Fallback} & \textbf{VW SR} & \textbf{Fcst SD} & \textbf{Unique fcsts} & \textbf{Min hidden SD} \\
\midrule
NN s1d3 & 3 & 0 & 0.87 & $2.1\times 10^{-2}$ & 77,387 & $1.4\times 10^{-3}$ \\
NN s1d7 & 7 & 443 & 0.21 & 0 & 1 & 0 \\
ResNet s1d3 & 3 & 0 & 0.88 & $2.8\times 10^{-2}$ & 77,417 & $1.6\times 10^{-1}$ \\
ResNet s1d7 & 7 & 0 & 0.90 & $2.4\times 10^{-2}$ & 77,440 & $2.0\times 10^{-1}$ \\
\bottomrule
\end{tabular}
}
\par\vspace{4pt}
\begin{minipage}{\linewidth}
\setstretch{1.0}\footnotesize Width schedule 1 is shown at depths 3 and 7; depth counts hidden layers, and s1 and d in the model labels denote the width schedule and depth.
Fallback counts full-universe months that use the constant-forecast decile fallback, and VW SR is the annualized Sharpe ratio of the value-weighted long-short portfolio.
Fcst SD and Min hidden SD report cross-sectional standard deviations of forecasts and hidden representations at the 1987 forecast origin for ensemble member 0; Unique fcsts counts forecasts distinct at 12 decimal places.
These diagnostics show whether fitted forecasts retain the dispersion needed for decile sorting; they do not identify the cause of forecast collapse.
The portfolio performance tables, heatmaps, and inference exclude constant-forecast specifications. Figure~\ref{fig:DecileMonotonicitySelected} retains them, assigning tied forecasts to deciles using a fixed random seed.
\end{minipage}
\end{table}

%% file: EqualWeightedRobustness.tex
\begin{table}[htbp]
\centering
\small
\setstretch{1.0}
\renewcommand{\arraystretch}{1.12}
\captionsetup{skip=4pt,justification=raggedright,singlelinecheck=false,labelsep=newline,labelfont=bf,textfont=bf}
\caption{\textbf{Equal-Weighted Portfolio Performance and Sharpe Ratio Differences by Depth}}
\label{tab:EqualWeightedRobustness}
\resizebox{\linewidth}{!}{%
\begin{tabular}{llrrrrrrr}
\toprule
\textbf{Comparison} & \textbf{Depth} & \textbf{Res. SR} & \textbf{FF SR} & \textbf{$\Delta$ SR} & \textbf{$t$} & \textbf{$\Delta R^2_{\mathrm{OOS}}$} & \textbf{$\Delta$ FF5 $\alpha$} & \textbf{$\Delta$ Turn.} \\
\midrule
\multirow{3}{*}{ResNet vs NN} & Shallow & \makebox[3.2em][c]{\ensuremath{2.40}} & \makebox[3.2em][c]{\ensuremath{2.31}} & \makebox[3.2em][c]{\ensuremath{0.09^{***}}} & 6.46 & 0.01 & 0.13 & 0.55 \\
 & Medium & \makebox[3.2em][c]{\ensuremath{2.46}} & \makebox[3.2em][c]{\ensuremath{2.29}} & \makebox[3.2em][c]{\ensuremath{0.17^{***}}} & 4.22 & 0.01 & 0.08 & 0.63 \\
 & Deep & \makebox[3.2em][c]{\ensuremath{2.42}} & \makebox[3.2em][c]{\ensuremath{2.13}} & \makebox[3.2em][c]{\ensuremath{0.29^{***}}} & 3.74 & 0.02 & 0.04 & 1.17 \\
\addlinespace[1pt]
\arrayrulecolor{gray!45}
\midrule[0.25pt]
\arrayrulecolor{black}
\multirow{3}{*}{ResNet+ vs NN+} & Shallow & \makebox[3.2em][c]{\ensuremath{3.53}} & \makebox[3.2em][c]{\ensuremath{3.28}} & \makebox[3.2em][c]{\ensuremath{0.25^{**}}} & 2.49 & 0.06 & 0.50 & 0.09 \\
 & Medium & \makebox[3.2em][c]{\ensuremath{3.52}} & \makebox[3.2em][c]{\ensuremath{2.60}} & \makebox[3.2em][c]{\ensuremath{0.92^{***}}} & 3.94 & 0.13 & 1.48 & 4.35 \\
 & Deep & \makebox[3.2em][c]{\ensuremath{3.83}} & \makebox[3.2em][c]{\ensuremath{1.60}} & \makebox[3.2em][c]{\ensuremath{2.22^{***}}} & 6.40 & 0.24 & 3.04 & -12.15 \\
\bottomrule
\end{tabular}
}
\par\vspace{4pt}
\begin{minipage}{\linewidth}
\setstretch{1.0}\footnotesize The table repeats the full-universe long-short portfolio tests using equal-weighted decile returns.
Res. SR and FF SR are the annualized Sharpe ratios of the residual and feedforward models; each $\Delta$ subtracts the feedforward model from the residual model.
$\Delta$ FF5 $\alpha$ and $\Delta$ Turn. are monthly percentages, while $\Delta R^2_{\mathrm{OOS}}$ is reported in percentage points.
The $R^2_{\mathrm{OOS}}$ measure follows the calibrated out-of-sample $R^2$ definition in Table~\ref{tab:FinalBaselinePerformance} and is a stock-level statistic independent of portfolio weighting.
Depth counts hidden layers: 1--5 for shallow, 6--10 for medium, and 11--20 for deep.
Within each depth group, width and depth specifications completed by both models receive equal weight.
Stars on $\Delta$ SR use the two-sided Ledoit--Wolf studentized circular block bootstrap test for equal Sharpe ratios; $t$ is $\Delta$ SR divided by its Ledoit--Wolf standard error.
*, **, and *** denote the 10\%, 5\%, and 1\% levels.
\end{minipage}
\end{table}

%% file: PairedSpecificationDepthProfile.tex
\begingroup
\setlength{\LTleft}{0pt}
\setlength{\LTright}{\fill}
\setstretch{1.0}
\fontsize{6.1}{6.7}\selectfont
\setlength{\tabcolsep}{0pt}
\renewcommand{\arraystretch}{1.02}
\captionsetup{skip=4pt,justification=raggedright,singlelinecheck=false,labelsep=newline,labelfont=bf,textfont=bf}
\begin{longtable}{@{\hspace{24pt}}>{\raggedright\arraybackslash}p{55pt}@{\hspace{2pt}}>{\raggedleft\arraybackslash}p{25pt}@{\hspace{2pt}}>{\raggedright\arraybackslash}p{34pt}@{\hspace{2pt}}>{\raggedright\arraybackslash}p{57pt}@{\hspace{2pt}}>{\raggedleft\arraybackslash}p{29pt}@{\hspace{2pt}}>{\raggedleft\arraybackslash}p{25pt}@{\hspace{2pt}}>{\raggedleft\arraybackslash}p{30pt}@{\hspace{2pt}}>{\raggedleft\arraybackslash}p{29pt}@{\hspace{2pt}}>{\raggedleft\arraybackslash}p{25pt}@{\hspace{2pt}}>{\raggedleft\arraybackslash}p{39pt}@{\hspace{2pt}}>{\raggedleft\arraybackslash}p{21.3pt}@{\hspace{2pt}}>{\raggedleft\arraybackslash}p{21.3pt}@{\hspace{2pt}}>{\raggedleft\arraybackslash}p{21.3pt}@{\hspace{2pt}}>{\raggedleft\arraybackslash}p{21.3pt}@{\hspace{2pt}}>{\raggedleft\arraybackslash}p{21.3pt}@{\hspace{2pt}}>{\raggedleft\arraybackslash}p{21.3pt}@{\hspace{2pt}}>{\raggedleft\arraybackslash}p{21.3pt}@{\hspace{2pt}}>{\raggedleft\arraybackslash}p{21.3pt}@{\hspace{2pt}}>{\raggedleft\arraybackslash}p{21.3pt}@{\hspace{2pt}}>{\raggedleft\arraybackslash}p{21.3pt}@{\hspace{2pt}}>{\raggedleft\arraybackslash}p{21.3pt}@{}}
\caption{\hspace*{24pt}\textbf{Full ResNet+ and NN+ Value-Weighted Portfolio Sorts by Width and Depth}}\label{tab:PairedSpecificationDepthProfile}\\
\toprule
\textbf{Width} & \textbf{Depth} & \textbf{Model} & \textbf{Spec} & \textbf{$R^2_{\mathrm{OOS}}$} & \textbf{Scale} & \textbf{VW~SR} & \textbf{FF5~$\alpha$} & \textbf{$t(\alpha)$} & \textbf{Turnover} & \textbf{D1} & \textbf{D2} & \textbf{D3} & \textbf{D4} & \textbf{D5} & \textbf{D6} & \textbf{D7} & \textbf{D8} & \textbf{D9} & \textbf{D10} & \textbf{H$-$L} \\
\midrule
\endfirsthead
\caption[]{\hspace*{24pt}\textbf{Full ResNet+ and NN+ Value-Weighted Portfolio Sorts by Width and Depth (continued)}}\\
\toprule
\textbf{Width} & \textbf{Depth} & \textbf{Model} & \textbf{Spec} & \textbf{$R^2_{\mathrm{OOS}}$} & \textbf{Scale} & \textbf{VW~SR} & \textbf{FF5~$\alpha$} & \textbf{$t(\alpha)$} & \textbf{Turnover} & \textbf{D1} & \textbf{D2} & \textbf{D3} & \textbf{D4} & \textbf{D5} & \textbf{D6} & \textbf{D7} & \textbf{D8} & \textbf{D9} & \textbf{D10} & \textbf{H$-$L} \\
\midrule
\endhead
\midrule
\multicolumn{21}{r}{Continued on next page}\\
\endfoot
\bottomrule
\addlinespace[4pt]
\multicolumn{21}{@{\hspace{24pt}}p{\dimexpr\linewidth-24pt\relax}@{}}{\scriptsize The table reports all available full-universe model, width, and depth specifications after excluding constant forecasts.
Width gives the width schedule, Depth counts hidden layers, and Spec is the model, width schedule, and depth label.
$R^2_{\mathrm{OOS}}$ is reported in percent and follows the calibrated out-of-sample $R^2$ definition in Table~\ref{tab:FinalBaselinePerformance}: monthly cross-sectional demeaning, 2.5\% forecast trimming from each tail, and nonnegative scalar rescaling before computing out-of-sample $R^2$.
Scale is the fitted nonnegative rescaling coefficient.
VW SR is the annualized Sharpe ratio of the value-weighted long-short portfolio.
Fama\mbox{-}French five-factor $\alpha$, turnover, D1 through D10, and H$-$L are monthly percentages; $t(\alpha)$ uses Newey\mbox{-}West standard errors with 12 lags.
D1 and D10 are the lowest and highest forecast deciles, and H$-$L is D10 minus D1.}\\
\endlastfoot
S1=[32] & 1 & ResNet+ & resnetps1d1 & 0.33 & 0.75 & 2.11 & 3.79 & 7.77 & 268.70 & -1.32 & -0.22 & 0.21 & 0.50 & 0.56 & 0.77 & 0.90 & 1.15 & 1.51 & 2.45 & 3.77 \\
S1=[32] & 1 & NN+ & nnps1d1 & 0.33 & 0.75 & 2.11 & 3.79 & 7.77 & 268.70 & -1.32 & -0.22 & 0.21 & 0.50 & 0.56 & 0.77 & 0.90 & 1.15 & 1.51 & 2.45 & 3.77 \\
S1=[32] & 2 & ResNet+ & resnetps1d2 & 0.33 & 0.79 & 1.95 & 3.26 & 9.11 & 260.67 & -1.37 & -0.26 & 0.14 & 0.45 & 0.54 & 0.74 & 0.96 & 1.13 & 1.36 & 1.97 & 3.34 \\
S1=[32] & 2 & NN+ & nnps1d2 & 0.32 & 0.83 & 2.09 & 3.49 & 7.62 & 261.60 & -1.38 & -0.38 & 0.13 & 0.40 & 0.53 & 0.77 & 1.03 & 1.15 & 1.38 & 2.30 & 3.68 \\
S1=[32] & 3 & ResNet+ & resnetps1d3 & 0.36 & 0.75 & 1.85 & 3.33 & 7.81 & 246.43 & -1.32 & -0.27 & 0.31 & 0.53 & 0.61 & 0.80 & 0.98 & 1.03 & 1.29 & 1.94 & 3.26 \\
S1=[32] & 3 & NN+ & nnps1d3 & 0.32 & 0.85 & 1.98 & 3.37 & 7.44 & 265.04 & -1.42 & -0.33 & 0.11 & 0.41 & 0.53 & 0.78 & 1.01 & 1.14 & 1.38 & 2.16 & 3.58 \\
S1=[32] & 4 & ResNet+ & resnetps1d4 & 0.34 & 0.74 & 1.92 & 3.31 & 7.88 & 243.62 & -1.27 & -0.22 & 0.28 & 0.47 & 0.69 & 0.77 & 0.97 & 1.03 & 1.28 & 2.04 & 3.31 \\
S1=[32] & 4 & NN+ & nnps1d4 & 0.34 & 0.95 & 2.15 & 3.55 & 7.96 & 265.30 & -1.47 & -0.40 & 0.11 & 0.44 & 0.51 & 0.76 & 0.90 & 1.26 & 1.40 & 2.32 & 3.78 \\
\end{longtable}
\endgroup

%% file: ResNetVsNNDepthProfile.tex
\begingroup
\setlength{\LTleft}{0pt}
\setlength{\LTright}{\fill}
\setstretch{1.0}
\fontsize{6.1}{6.7}\selectfont
\setlength{\tabcolsep}{0pt}
\renewcommand{\arraystretch}{1.02}
\captionsetup{skip=4pt,justification=raggedright,singlelinecheck=false,labelsep=newline,labelfont=bf,textfont=bf}
\begin{longtable}{@{\hspace{24pt}}>{\raggedright\arraybackslash}p{55pt}@{\hspace{2pt}}>{\raggedleft\arraybackslash}p{25pt}@{\hspace{2pt}}>{\raggedright\arraybackslash}p{34pt}@{\hspace{2pt}}>{\raggedright\arraybackslash}p{57pt}@{\hspace{2pt}}>{\raggedleft\arraybackslash}p{29pt}@{\hspace{2pt}}>{\raggedleft\arraybackslash}p{25pt}@{\hspace{2pt}}>{\raggedleft\arraybackslash}p{30pt}@{\hspace{2pt}}>{\raggedleft\arraybackslash}p{29pt}@{\hspace{2pt}}>{\raggedleft\arraybackslash}p{25pt}@{\hspace{2pt}}>{\raggedleft\arraybackslash}p{39pt}@{\hspace{2pt}}>{\raggedleft\arraybackslash}p{21.3pt}@{\hspace{2pt}}>{\raggedleft\arraybackslash}p{21.3pt}@{\hspace{2pt}}>{\raggedleft\arraybackslash}p{21.3pt}@{\hspace{2pt}}>{\raggedleft\arraybackslash}p{21.3pt}@{\hspace{2pt}}>{\raggedleft\arraybackslash}p{21.3pt}@{\hspace{2pt}}>{\raggedleft\arraybackslash}p{21.3pt}@{\hspace{2pt}}>{\raggedleft\arraybackslash}p{21.3pt}@{\hspace{2pt}}>{\raggedleft\arraybackslash}p{21.3pt}@{\hspace{2pt}}>{\raggedleft\arraybackslash}p{21.3pt}@{\hspace{2pt}}>{\raggedleft\arraybackslash}p{21.3pt}@{\hspace{2pt}}>{\raggedleft\arraybackslash}p{21.3pt}@{}}
\caption{\hspace*{24pt}\textbf{Full ResNet and NN Value-Weighted Portfolio Sorts by Width and Depth}}\label{tab:ResNetVsNNDepthProfile}\\
\toprule
\textbf{Width} & \textbf{Depth} & \textbf{Model} & \textbf{Spec} & \textbf{$R^2_{\mathrm{OOS}}$} & \textbf{Scale} & \textbf{VW~SR} & \textbf{FF5~$\alpha$} & \textbf{$t(\alpha)$} & \textbf{Turnover} & \textbf{D1} & \textbf{D2} & \textbf{D3} & \textbf{D4} & \textbf{D5} & \textbf{D6} & \textbf{D7} & \textbf{D8} & \textbf{D9} & \textbf{D10} & \textbf{H$-$L} \\
\midrule
\endfirsthead
\caption[]{\hspace*{24pt}\textbf{Full ResNet and NN Value-Weighted Portfolio Sorts by Width and Depth (continued)}}\\
\toprule
\textbf{Width} & \textbf{Depth} & \textbf{Model} & \textbf{Spec} & \textbf{$R^2_{\mathrm{OOS}}$} & \textbf{Scale} & \textbf{VW~SR} & \textbf{FF5~$\alpha$} & \textbf{$t(\alpha)$} & \textbf{Turnover} & \textbf{D1} & \textbf{D2} & \textbf{D3} & \textbf{D4} & \textbf{D5} & \textbf{D6} & \textbf{D7} & \textbf{D8} & \textbf{D9} & \textbf{D10} & \textbf{H$-$L} \\
\midrule
\endhead
\midrule
\multicolumn{21}{r}{Continued on next page}\\
\endfoot
\bottomrule
\addlinespace[4pt]
\multicolumn{21}{@{\hspace{24pt}}p{\dimexpr\linewidth-24pt\relax}@{}}{\scriptsize The table reports all available full-universe model, width, and depth specifications after excluding constant forecasts.
Width gives the width schedule, Depth counts hidden layers, and Spec is the model, width schedule, and depth label.
$R^2_{\mathrm{OOS}}$ is reported in percent and follows the calibrated out-of-sample $R^2$ definition in Table~\ref{tab:FinalBaselinePerformance}: monthly cross-sectional demeaning, 2.5\% forecast trimming from each tail, and nonnegative scalar rescaling before computing out-of-sample $R^2$.
Scale is the fitted nonnegative rescaling coefficient.
VW SR is the annualized Sharpe ratio of the value-weighted long-short portfolio.
Fama\mbox{-}French five-factor $\alpha$, turnover, D1 through D10, and H$-$L are monthly percentages; $t(\alpha)$ uses Newey\mbox{-}West standard errors with 12 lags.
D1 and D10 are the lowest and highest forecast deciles, and H$-$L is D10 minus D1.}\\
\endlastfoot
S1=[32] & 1 & ResNet & resnets1d1 & 0.21 & 0.47 & 0.88 & 1.96 & 6.94 & 211.01 & -0.49 & 0.10 & 0.49 & 0.47 & 0.68 & 0.69 & 0.79 & 0.87 & 0.94 & 1.26 & 1.75 \\
S1=[32] & 1 & NN & nns1d1 & 0.21 & 0.47 & 0.88 & 1.96 & 6.94 & 211.01 & -0.49 & 0.10 & 0.49 & 0.47 & 0.68 & 0.69 & 0.79 & 0.87 & 0.94 & 1.26 & 1.75 \\
S1=[32] & 2 & ResNet & resnets1d2 & 0.21 & 0.47 & 0.87 & 1.93 & 6.83 & 210.68 & -0.45 & 0.10 & 0.51 & 0.51 & 0.69 & 0.70 & 0.77 & 0.89 & 0.93 & 1.27 & 1.71 \\
S1=[32] & 2 & NN & nns1d2 & 0.21 & 0.49 & 0.85 & 1.88 & 6.64 & 211.00 & -0.46 & 0.11 & 0.45 & 0.55 & 0.66 & 0.70 & 0.80 & 0.87 & 0.93 & 1.23 & 1.69 \\
S1=[32] & 3 & ResNet & resnets1d3 & 0.21 & 0.48 & 0.88 & 1.97 & 6.55 & 211.61 & -0.47 & 0.11 & 0.50 & 0.45 & 0.72 & 0.67 & 0.78 & 0.88 & 0.95 & 1.28 & 1.75 \\
S1=[32] & 3 & NN & nns1d3 & 0.20 & 0.49 & 0.87 & 1.86 & 6.58 & 215.51 & -0.45 & 0.11 & 0.52 & 0.51 & 0.71 & 0.72 & 0.79 & 0.90 & 0.94 & 1.23 & 1.68 \\
S1=[32] & 4 & ResNet & resnets1d4 & 0.21 & 0.48 & 0.88 & 1.98 & 6.85 & 210.60 & -0.51 & 0.13 & 0.47 & 0.51 & 0.70 & 0.68 & 0.79 & 0.87 & 0.92 & 1.25 & 1.75 \\
S1=[32] & 4 & NN & nns1d4 & 0.18 & 0.54 & 0.89 & 1.88 & 7.07 & 220.04 & -0.35 & 0.15 & 0.54 & 0.51 & 0.71 & 0.74 & 0.83 & 0.93 & 0.91 & 1.30 & 1.64 \\
\end{longtable}
\endgroup